\documentclass{article}

\usepackage{arxiv}
\renewcommand\normalsize{\fontsize{12pt}{14pt}\selectfont}
\newcommand{\OMIT}[1]{}
\usepackage[utf8]{inputenc} 
\usepackage[T1]{fontenc}    
\usepackage{hyperref}       
\hypersetup{
  colorlinks   = true, 
  urlcolor     = blue, 
  linkcolor    = black, 
  citecolor   =  blue 
}     
\usepackage{url}            
\usepackage{booktabs}       
\usepackage{amsfonts}       
\usepackage{nicefrac}       
\usepackage{microtype}      
\usepackage{lipsum}
\usepackage{graphicx}
\graphicspath{ {./images/} }
\usepackage{apxproof}
\usepackage{amsthm}
\usepackage{lipsum}
\usepackage{xcolor}
\usepackage{amsmath}
\usepackage{tabularx}
\usepackage{amssymb} 

\usepackage[ruled]{algorithm2e} 

\SetAlFnt{\small}
\SetAlCapFnt{\small}
\SetAlCapNameFnt{\small}
\SetAlCapHSkip{0pt}
\IncMargin{-\parindent}

\usepackage{doi}

\usepackage{amsmath}
\newtheorem{theorem}{Theorem}[section]
\newtheorem{lemma}[theorem]{Lemma}
	\newtheorem{corollary}[theorem]{Corollary}
	\newtheorem{proposition}[theorem]{Proposition}
	
	\newtheorem{definition}[theorem]{Definition}

	\theoremstyle{definition}

	\newtheorem*{note*}{Note}
	\newtheorem{remark}[theorem]{Remark}
	\newtheorem*{remark*}{Remark}

\theoremstyle{definition}

\DeclareTextFontCommand{\emph}{\em}

\newif\ifcomments
\commentstrue
\ifcomments
\newcommand{\bruce}[1]{\textcolor{blue}{(Bruce: #1)}}
\newcommand{\koosha}[1]{\textcolor{green}{(Koosha: #1)}}
\newcommand{\changes}[1]{\textcolor{red}{(changes: #1)}}
\newcommand{\new}[1]{\textcolor{brown}{( #1)}}
\else
\newcommand{\bruce}[1]{}
\newcommand{\koosha}[1]{}
\newcommand{\changes}[1]{}
\newcommand{\new}[1]{}
\fi

\usepackage{dsfont}

\usepackage[shortlabels]{enumitem}

\newlist{primenumerate}{enumerate}{1}
\setlist[primenumerate,1]{label={\arabic*$'$}}

\usepackage{booktabs,tabularx,amsmath}
\usepackage{lipsum}

\usepackage{xcolor}
\ExplSyntaxOn
\NewDocumentCommand{\problemStatement}{mm}
 {
  \arteche_problemstatement:nn { #1 } { #2 }
 }

\prop_new:N \l_arteche_problemstatement_body_prop

\cs_new_protected:Nn \arteche_problemstatement:nn
 {
  \prop_set_from_keyval:Nn \l_arteche_problemstatement_body_prop { #2 }
  \begin{center}
  \begin{tabularx}{\columnwidth}{@{}lX@{}}
  \toprule
  \multicolumn{2}{@{}c@{}}{\textsc{#1}}\tabularnewline
  \midrule
  \prop_map_function:NN \l_arteche_problemstatement_body_prop \__arteche_problemstatemet_do:nn
  \bottomrule
  \end{tabularx}
  \end{center}
}

\cs_new_protected:Nn \__arteche_problemstatemet_do:nn
 {
  \bfseries \tl_rescan:nn { } { #1 }: & #2 \\
 }

\ExplSyntaxOff

\title{The Computational Complexity of Variational Inequalities and Applications in Game Theory}

\author{ Bruce M. Kapron\hspace{2mm}\href{https://orcid.org/0000-0002-3295-543X}{\includegraphics[scale=0.05]{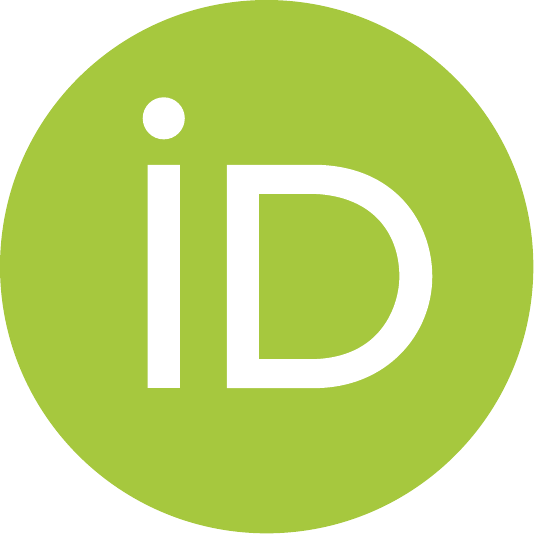}}\ \\
	Department of Computer Science\\
	University of Victoria\\
	Victoria, BC\\
	\texttt{bmkapron@uvic.ca} \\
	\And
	Koosha Samieefar\hspace{1mm}\href{https://orcid.org/0000-0001-8960-9282}{\includegraphics[scale=0.06]{orcid.pdf}} \\
	Department of Computer Science\\
	University of Victoria\\
	Victoria, BC \\
	\texttt{koosha021@gmail.com} \\
}
\begin{document}
\maketitle
\begin{abstract}
We present a computational formulation for the approximate version of several variational inequality problems, investigating their computational complexity and establishing PPAD-completeness. Examining applications in computational game theory, we specifically focus on two key concepts: resilient Nash equilibrium, and multi-leader-follower games—domains traditionally known for the absence of general solutions. In the presence of standard assumptions and relaxation techniques, we formulate problem versions for such games that are expressible in terms of variational inequalities, ultimately leading to proofs of PPAD-completeness.
\end{abstract}

\begin{titlepage}

\maketitle


\setcounter{tocdepth}{1} 
\tableofcontents

\end{titlepage}

\section{Introduction}

A variational inequality (VI) is a problem specified by a closed convex set $\mathcal{R}\in \mathbb{R}^n$ and a function  $F:\mathbb{R}^n \rightarrow \mathbb{R}^n$, and asks for some $x^* \in \mathcal{R}$ such that for every $y\in \mathcal{R}$, $\langle F(x^*),y-x^* \rangle \ge 0$. This formulation generalizes a number of familiar problems, including systems of equations, convex optimization, equilibria, and fixed points.
 Quasi-variational inequalities (QVI) extend this notion by introducing a set-valued mapping, and dependence of the feasible set on an external variable. Even further generalizations lead to generalized quasi-variational inequalities (GQVI) in which  $F:\mathbb{R}^n \rightarrow \mathbb{R}^n$ becomes a correspondence,   extending classical QVIs as well as the generalized variational inequality (GVI) studied in \cite{Fang}. 
These problems find application in diverse fields including optimization theory, economics, and engineering, serving as a modeling and solution tool for a broad spectrum of real-world problems \cite{Giannessi1995,Panagiotopoulos1,Panagiotopoulos2,Pang2005}. This is especially true in cases where the aim is to find a solution subject to a particular inequality condition that is defined by a set of functions or operators.  For example, under some assumptions such as differentiability, Debreu-Rosen style games (see \cite{Social,concavegames,rosen}) can be expressed as QVIs, aiding their analysis through variational techniques, and offering a unified framework for various multiplayer, non-cooperative games \cite{Pang2005}.

Finding a solution to more generalized forms of variational inequalities such as QVI  is typically more intricate compared to a standard variational inequality due to their increased generality. Commencing with the formal definition of QVI proposed by Bensoussan and Lion (\cite{Bensoussan1,Bensoussan4,Bensoussan3,Bensoussan2}), researchers have engaged in an exploration of algorithmic solutions, focusing on conditions governing convergence. Various numerical and mathematical techniques, including fixed-point methods, penalty methods, and projection methods, can be used to find solutions to QVIs \cite{fixpointbook,pang2004,Giannessi1995,Pang2005,Rockafellar1998}.  In many instances, directly solving QVIs can be challenging. Researchers often resort to regularization techniques or approximate the problem to enhance its suitability for numerical methods. These techniques may yield an approximate solution within a specified tolerance of the exact solution. In such cases, the complexity is contingent upon the chosen approximation accuracy and the convergence rate of the algorithm. The study of the computational aspects of these variational inequalities is in its early stages of development. Specifically, the majority of research papers focus on examining whether a solution to a problem exists \cite{battaexsy,chan,existenceqvi,qviexist2,pangchan2,YAOexsitqvi}. For example, in \cite{chan},  an existence result for the GQVI  was proved using the Eilenberg-Montgomery fixed point theorem (see \cite{ELMfixp}).


Strong Nash equilibrium \emph{(SNE)} is a concept in game theory that refines the traditional notion of Nash Equilibrium. In SNE, not only does each player have no incentive to deviate from their strategy given the strategies of others unilaterally, but this equilibrium also withstands deviations by any possible coalition of players \cite{Strongnash}. In other words, even if a group of players collaboratively tries to deviate from the equilibrium, they cannot collectively gain from such deviations.  This concept had faced criticism for being excessively "strong", particularly in environments permitting unlimited private communication (for more information, see\cite{Coalationproof}).  Furthermore, a strong Nash equilibrium is weakly Pareto-efficient, a fact which may be verified by considering a deviation of the ``grand coalition'' of all players.  In conclusion, the existence of such equilibria is unlikely (even if mixed strategies are allowed) \cite{Coalationproof}, except under strict conditions and assumptions (extensively discussed in \cite{StrongkNash,socialcoalition,StrongNashexist,strongnashvotingexists}.) In contrast, for \emph{coalition-proof Nash equilibrium (CPNE)}, limitations on private communication are imposed \cite{Coalationproof}. CPNE emphasizes stability against coordinated actions by any subset of players, regardless of the coalition's size while SNE goes further by requiring that not only should deviations by any coalition be unsuccessful, but any subset of players (including the entire group) should also be unable to deviate successfully. In contrast to the NP-completeness of deciding whether an SNE exists \cite{Conitzer,Gatti2013,gatti2017verification}, the problem of finding an SNE is in smoothed polynomial time\cite{strongnashnothard2,strongnashnothard3,strongnashnothard1}\footnote{These results concern the problem of finding (mixed strategy and exact solution of) SNE in bi-matrix games which is a special case. Specifically, they show that if there is a mixed–strategy, the payoffs restricted to the actions in the support must satisfy strict geometric conditions.}. We may also restrict ourselves to coalitions of size at most $t$. This limitation appears reasonable in practice, as forming and coordinating large coalitions can be challenging. This fact in turn motivated the definition of \emph{$t$-strong Nash equilibrium} and \emph{$t$-resilient Nash equilibrium}\cite{Abraham1,StrongkNash}. Informally, here, no member of a coalition of size up to $t$ can do better, even if the whole coalition defects.  Apart from the technical interest of these variations, numerous papers explore their practical applications \cite{Abraham1,Abraham3,coalitionproofapplications,resilientapplication}.


 Multi-leader-follower games are a class of games in which multiple agents, referred to as \emph{leaders} and \emph{followers}, interact strategically to achieve their respective objectives. Leaders and followers often have conflicting objectives or interests, in particular aiming to maximize their own profit or minimize their loss.  This setup resembles a bi-level program, where the leaders engage in competitive, non-cooperative Nash games at the upper level, making their decisions while considering the followers' responses. After the leaders have made their choices, the followers then engage in a parametric, non-cooperative game at the lower level, with the strategies of the leaders treated as external parameters.  The concept of multi-leader-follower games has a variety of applications that arise from situations where there are multiple oligopoly firms operating in the market \cite{mlfapp1,mlfapp2,mlfapp3,mlfapp4,Pang2005}. Oligopoly markets are markets dominated by a small number of suppliers.   The simplest form of the multi-leader-follower game is a Stackelberg game \cite{stacklbergbook,Sherali_1984,Stackelberg} in which one leader and multiple followers react to the leader’s strategies. These games have applications in various fields, including economics, engineering, and multi-agent systems \cite{Stackelbergapp1,Stackelbergapp2,Stackelbergapp3}.  Traditional game theory provides solution concepts for analyzing and solving multi-leader follower games. Common solution concepts include Stackelberg equilibrium (a hierarchical equilibrium where the leader moves first and followers react) and sub-game perfect equilibrium (a refinement of Nash equilibrium that considers strategies at each stage of the game).  While the multi-leader-follower problem provides a sound mathematical framework with a well-defined solution concept and practical applications, its elevated level of complexity and technical intricacies make it computationally challenging or intractable. Specifically, it resembles an equilibrium problem in a more complex form of Debreu-Rosen style games, requiring each leader to solve a non-convex mathematical program with equilibrium constraints \cite{bilevel2,Pang2005,bilevel}. This formulation faces two significant issues: a potential absence of an equilibrium solution due to non-convexity, and computational intractability. In response to these challenges, \cite{Pang2005} proposed a careful analysis and selection of ``remedial models'' aimed at deriving sensible equilibrium solutions. An alternative approach involves examining a specific class of multi-leader-follower games \cite{mlfapp3} that adhere to particular yet reasonable assumptions.

Problems for which  the existence of a solution follows from Kakutani's fixed-point theorem, particularly those related to games introduced by Debreu and Rosen, as well as quasi-variational inequalities, have not been extensively explored from a complexity standpoint. In \cite{PPADcomplexity}, a problem known as {\sc Kakutani} was introduced, with a brief overview of its inclusion in the complexity class PPAD.
The primary challenge in formulating {\sc Kakutani} as a computational problem lies in the limitations of conventional approaches to explicitly and succinctly represent convex sets. Methods such as the convex hull of a point set or a convex polytope defined by linear inequalities prove overly restrictive and fail to encompass vital practical applications of Kakutani's fixed-point theorem, such as its application to the games mentioned earlier. A more suitable computational formulation of the Kakutani problem has recently been introduced by leveraging the computational convex geometry methods introduced in \cite{Grostel}. These methods represent convex sets using separation oracles. This formulation was subsequently applied in \cite{concavegames}, which considers the computational complexity of finding approximate equilibrium solutions for Rosen-style games.

\subsection*{Our Contribution}

A primary objective of algorithmic game theory is to categorize the complexity of key economic concepts. PPAD-completeness has become a significant unifying principle within this endeavor \cite{Daskalakis2006,concavegames}. 
The main results of this paper focus on the computational complexity of finding approximate solutions for variational inequalities, and variants, namely quasi-variational inequalities and generalized quasi-variational inequalities, and their applications in game theory. We give formulations of these problems that are general enough for a variety of game-theoretic applications and prove that they are PPAD-complete. PPAD-hardness can be shown directly by the fact that well-known PPAD-hard problems such as the Nash equilibrium problem may be directly formulated as variational inequality problems. The main challenge addressed in this paper is establishing membership in PPAD for the problems we consider. This requires extensions and also combination of techniques and results from a variety of previous works \cite{chan,Harker,Pang2005,concavegames}. Note that our proofs also provide an alternative method for showing the PPAD-completeness for the Rosen-style games considered in \cite{concavegames} and even for Debreu-style games, using variational inequalities (see Proposition \ref{geneqvari2})\footnote{An alternative proof is available in recent work \cite{STOC24}.}.

Our formulation of the computational version of these problems uses weak and strong separation oracles. Informally, strong (weak) separation oracles can verify the membership (almost membership) of a point in given correspondences, or sets. Simpler methods of representing convex sets, such as using convex hulls or convex polytopes, are too restrictive and often not useful for capturing the practical applications of Kakutani's fixed point in game theory, as discussed in \cite{concavegames}.  We can similarly motivate this approach for variational inequalities. As we are concerned with the complexity of the problems, we would like to establish PPAD membership for their most general form.  Moreover, using convex hull/polytope representations would require new techniques to replace the use of Berge's maximum theorem and Kakutani's fixed point and would likely not significantly simplify our proofs\footnote{For more information, see both the definitions of Kakutani's fixed point and also, Theorem E.1 in \cite{concavegames}.}. In addition, in different application areas, including economics, variational inequality problems with more general convex sets arise, e.g., the approach to modeling Walrasian equilibria given in \cite{walrasian}. Another simple example of equilibrium problems in which general convex sets are required to present the constraint is the problem of finding generalized close equilibrium in which each player's alternative admissible strategies are limited to strategies that are statistically close to their strategies, not the whole mixed strategy sets \cite{MFCS}.

\begin{theorem}[informal]
\label{informal1}
Finding an approximate solution to computational variants of GQVI, QVI and VI is PPAD-complete where:

\begin{itemize}
    \item The correspondences\footnote{For example, an instance of a GQVI has two correspondences, a QVI has only one, and a VI has none.} are convex-valued and Lipschitz and given by either a strong or weak separation oracle. 
    \item The functions are convex and are given by circuits.
    \item The sets are convex and given by either a strong or weak separation oracle.
\end{itemize}

\end{theorem}

 Next, applying the computational versions of the variational inequality problems, we demonstrate that several central problems in game theory are PPAD-complete under reasonable assumptions. First, we address challenges to finding an $L/F$-equilibrium for a multi-leader-follower game. In addition to the complexity of the leaders' optimization problem, the non-convexity of the leader's constraints may mean that a $L/F$-equilibrium does not always exist. We address the non-convexity of the leader's constraints by following an algebraic relaxation/restriction approach introduced in \cite{Pang2005} which remedies the non-existence of a $L/F$-equilibrium by using a \emph{Karush–Kuhn–Tucker condition} (KKT) representation  and additionally introducing various relaxations/restrictions of such a representation.

\begin{theorem}[informal]
Finding an approximate remedial solution (proposed by \cite{Pang2005}) of an $L/F$-equilibrium in a multi-leader follower game is PPAD-complete under the following assumptions:

\begin{itemize}
    \item The correspondences are convex-valued given by linear arithmetic circuits that represent a strong or weak separation oracle. 
    \item The functions are represented by linear arithmetic circuits.
    \item The sets are convex given by linear arithmetic circuits which represent either a strong or weak separation oracle.
\end{itemize}
\end{theorem}

 Finally, we demonstrate that given specific continuity and  concavity conditions (in particular a more restricted condition that we call multi-concavity), notions such as $t$-resilient Nash and some other related notions can be formulated using a modified version of variational inequalities. Informally, the multi-concavity of the (multivariate) function $f(x_1,\dots,x_n)$  means that $f(x_{i_1},\dots,x_{i_t},.)$  is concave for any $\{i_1,\dots,i_t\}\subset \{1,\dots,n\}$.

\begin{theorem}[informal]
Finding an approximate solution to $t$-resilient Nash equilibrium, $t$-strong Nash equilibrium, coalition-proof Nash equilibrium (for a given constant $t$) are PPAD-complete under the following assumptions\footnote{Note that decision version of $t$-resilient Nash equilibrium without the multi-concavity condition (and also strong Nash equilibrium \cite{Conitzer}) is NP-hard in general. }:

\begin{itemize}
   \item The multi-concave utility functions are represented by linear arithmetic circuits.

\end{itemize}
\end{theorem}

\section{Preliminaries}\label{defs}

This section is dedicated to presenting fundamental definitions and basic facts used in this paper. For various definitions, we adhere to their classical definitions, which are framed as either maximization or minimization problems, to maintain consistency with the original terminology\footnote{For example, in the context of generalized (social) Nash equilibrium and related problems, the problems are often specified as minimization problems. }.

\subsection{Some Elementary Definitions}

\begin{definition}
 For sets $X,Y$, a \emph{correspondence} (also called a \emph{set-valued map} or a  \emph{a point-to-set map}) from $X$ to $Y$ is a function $\mathcal{R}:X\rightarrow 2^Y$. Sometimes we will denote such a mapping by $\mathcal{R}: X\rightrightarrows Y$. 
\end{definition}

\begin{definition}  Let $X$ be a convex, non-empty, and compact set.
   \begin{itemize}
     \item We let $\mathrm{d}(x, z)$ denote the Euclidean distance between  two points $x, z \in \mathbb{R}^m$.
     \item  The projection map of a point $x$ to the set $X$ is $\Pi_{X}(x)=\arg \inf _{z \in X} \mathrm{~d}(x, z)$.
     \item  The closed $\epsilon$-parallel body of $X$ is defined as $\overline{\mathrm{B}}(X, \epsilon):=\bigcup_{x \in X}\{z \in \mathcal{M}: \mathrm{d}(x, z) \leq \epsilon\}$
     Note that $\mathcal{M}$ is the universe. The inner closed $\epsilon$-parallel body of $X$ is $\overline{\mathrm{B}}(X,-\epsilon)=\{x \in X: \overline{\mathrm{B}}(x, \epsilon) \subseteq X\}$.

 \end{itemize}
   
\end{definition}
\begin{remark}
    The elements of $\overline{\mathrm{B}}(X,-\epsilon)$ can be viewed as the points "deep inside of $X$", while $\overline{\mathrm{B}}(X, \epsilon)$ as the points that are "almost inside of $X$". 
\end{remark}

\subsection{Semi-Continuity and Convexity}

\begin{definition}
   The correspondence $\mathcal{R}$ is \emph{upper semicontinuous} at $x \in X$ if for every open neighborhood $V$ of $\mathcal{R}(x)$ there is an open neighborhood $U$ of $x$ such that for every $y \in U$, $\mathcal{R}(y)\subseteq V$.
\end{definition}

\begin{definition}
    The correspondence $\mathcal{R}$ is \emph{lower semicontinuous} at $x$ if for any open set $V$ such that $\mathcal{R}(x)\cap V $ is nonempty, there exists a neighbourhood $U$ of $x$  such  that $\mathcal{R}(y)\cap  V$ is nonempty for all $y \in U$. 
\end{definition}

\begin{definition}
    A set $C \subseteq \mathbb{R}^n$ is called convex if, for any two points $x_1, x_2 \in C$ and for any $\alpha \in [0, 1]$, the following holds:
\[
\alpha x_1 + (1 - \alpha) x_2 \in C.
\]
\end{definition}

\begin{definition}
    A vector-valued function $f: \mathbb{R}^n \rightarrow \mathbb{R}^m$ is convex if, for any $x_1, x_2 \in \mathbb{R}^n$ and $\alpha \in [0, 1]$, the following holds component-wise:
\[
f(\alpha x_1 + (1-\alpha) x_2) \preceq \alpha f(x_1) + (1-\alpha) f(x_2),
\]
where $\preceq$ denotes element-wise inequality.
\end{definition}

\begin{definition}
    A function $f: \mathbb{R}^n \rightarrow \mathbb{R}$ is said to be a convex function with respect to the variable $x_i$ if, for any fixed values of the other variables $x_j$ (where $j \neq i$) and for any $x_{i1}, x_{i2} \in \mathbb{R}$ and $\alpha \in [0, 1]$, the following holds:
\[
f\left(x_{i1}, x_j\right) \leq \alpha f\left(x_{i2}, x_j\right) + (1 - \alpha) f\left(x_{i1}, x_j\right),
\]
for all $x_j$.
\end{definition}

\subsection{Kakutani's Fixed Point Theorem}

\begin{theorem}[ \cite{Kakutani}]
\label{thm:Kakutani}
Suppose $S$ is a  non-empty, compact and convex subset of $\mathbb{R}
^n$ and $\mathcal{R}:S\rightrightarrows 2^S$ is a correspondence with the following properties has a fixed point, i.e., a point $s \in S$ such that $s \in \mathcal{R}(s)$.
\begin{enumerate}
\item  $\mathcal{R}$ is upper semi-continuous on $S$;
\item $\mathcal{R}(s)$ is non-empty and convex for all $s\in S$.
\end{enumerate}

\end{theorem}

\begin{definition}
    
We say that $\mathcal{R}$ has a \emph{closed graph} if $\{(x,y)~|~y\in \mathcal{R}(x)\}$ is a closed subset of $X \times Y$ in the product topology.
\end{definition}

By the \emph{closed graph theorem}, since $S$ is compact, the first condition in Theorem~\ref{thm:Kakutani} may be replaced by  $\mathcal{R}$ having a closed graph.

\subsection{Basic Game Theory}

\begin{definition}
    Suppose that there are $k$ \emph{players}, where each player $i$,  $1 \le i \le k$ has a set $S_i$ of possible \emph{pure strategies}. Each also player has a \emph{utility function} $u_i:S_1\times\dots\times S_k\rightarrow \mathbb{R}$, and a game is specified by the tuple $\mathcal{U}=(u_1,\dots,u_k)$ and $\mathcal{S}=(S_1,\dots,S_k)$. Player $i$'s utility when each player $j$, $1\le j\le k$ plays $s_j$ is $u_i(s_1,\dots,s_k)$. 
\end{definition}

\begin{definition}

An \emph{$\epsilon$-approximate Nash equilibrium} is a strategy profile\footnote{A \emph{strategy profile} is a vector that includes all players' strategies.} $\textbf{s}=(s^*_1,\dots,s^*_k)$ such that for $1 \le i\le k$ and 
\begin{equation}
\label{eq:bestresp}
    u_i(\textbf{s}^*) +\epsilon \ge u_i(s_i,\textbf{s}^*_{-i}), \quad\quad\quad\quad \forall s_i\in S_i
\end{equation}
\end{definition}

Where $\textbf{s}_{-i}^*$ includes all strategies except $i$. A positive answer to the existence of a solution was provided by Nash (for the exact case, i.e, $\epsilon=0$) in his celebrated papers \cite{Nash,Nash0} for bi-matrix games and mixed strategies. Nash's result follows in a fairly direct way from Kakutani's fixed-point theorem
\cite{Kakutani,Nash0} (and also Brouwer's fixed point \cite{Nash}).

\begin{remark}
    We can assume that each $S_i$ is a bounded subset of  $\mathbb{R}^{n_i}$ and each player $i$ controls $dim(S_i)=n_i$ variables of the strategy profile. Furthermore, $\sum_{i=1}^k n_i=n$.
\end{remark}

A natural way to limit the notion of a player's best response is to suppose that player $i$, when considering responses to $\textbf{s}^*_{-i}$, is constrained to those lying in some subset of $S_i$, determined by $\textbf{s}^*_{-i}$. This constraint might be given as a correspondence $\mathcal{R}:S_{-i}\rightrightarrows S_i$, where $S_{-i}$ denotes $S_1\times\dots\times S_{i-1}\times S_{i+1}\times\dots\times S_k$. We then have the following condition:
\begin{equation*}
\label{eq:bestsocresp}
    u_i(s^*_i,\textbf{s}^*_{-i})+\epsilon \ge u_i(s_i,\textbf{s}^*_{-i}), \quad\quad\quad\quad \forall s_i\in \mathcal{R}_i(\textbf{s}^*_{-i})\tag{**}
\end{equation*}
Given this notion of best response, Debreu generalized the definition of Nash equilibrium:
\begin{definition} A \emph{ generalized (social)  equilibrium}  is a (pure) strategy profile $\textbf{s}=(s^*_1,\dots,s^*_n)$ which satisfies $s^*_i\in \mathcal{R}_i(s^*_{-i})$ and (\ref{eq:bestsocresp}) for $1 \le i \le n$.
\end{definition}


\begin{remark}
Following the approach in some applications to generalized (Nash) equilibria and multi-leader-follower games \cite{Generalized,Pang2005},
we may assume that each player $i$'s constraints are represented by the non-empty set motivated by :
\[\mathcal{R}_{i}\left( \textbf{x}_{-i}\right) \equiv\left\{x_{i} \in S_i ~|~ g_{i}\left(x_{i}, \textbf{x}_{-i}\right) \leq 0, \thickspace h_{i}\left(x_{i}\right) \leq 0\right\},\]
where $x_{i} \in S_{i}$ and $g_{i}: \Pi_{j=1 }^k S_j \rightarrow \mathbb{R}^{m_i}$ and $h_{i}: S_\mathrm{i} \rightarrow \mathbb{R}^{l_i}$ are continuously differentiable and $g_{i}\left(\cdot, \textbf{x}_{-i}\right)$ is convex in $\textbf{x}_{-i}$. Note that $m_i$ and $l_i$ are some given constants.
   
\end{remark}

\subsection{Variational Inequalities}


\begin{definition}
   Given a function $F:\mathbb{R}^m \rightarrow \mathbb{R}^m$  and a correspondence  $\mathcal{R}:\mathbb{R}^m \rightrightarrows \mathbb{R}^m$, an approximate solution to QVI  $(\mathcal{R}, F)$ is a vector $x^* \in \mathcal{R}(x^*)$  such that:
$$
\left(y-x^*\right)^T F\left(x^*\right)+\epsilon \geq 0, \quad \forall y \in \mathcal{R}\left(x^*\right)
$$
\end{definition}

\begin{definition}

  When $\mathcal{R}(x)$ is independent of $x$, say, $\mathcal{R}(x)=\mathcal{R}$ for all $x$, the QVI becomes the VI for which a solution $x^* \in \mathcal{R}$ satisfies:
$$
(y-x^*)^T F(x^*)+\epsilon \geq 0, \quad \forall y \in \mathcal{R}
$$  
\end{definition}







Generalized quasi-variational inequalities (GQVI) extend the notion of quasi-variational inequality (QVI).  GQVI is a unification of QVI and generalized VI \cite{chan}.

\begin{definition}
    Given a correspondence $\mathcal{F}: \mathbb{R}^m \rightrightarrows \mathbb{R}^m $ and a correspondence $\mathcal{R}:\mathbb{R}^m \rightrightarrows \mathbb{R}^m$, an approximate solution to GQVI $(\mathcal{R}, \mathcal{F})$ consists of two vectors $x^* \in \mathcal{R}(x^*)$ and $w^*\in \mathcal{F}(x^*)$ such that:
$$
\left(y-x^*\right)^T w^* +\epsilon \geq 0, \quad \forall y \in \mathcal{R}\left(x^*\right)
$$
\end{definition}


\begin{proposition}\label{geneqvari2}
  Suppose that we have a game with concave and continuously differentiable utilities $\mathcal{U}=(u_1,\dots,u_k)$,  strategies $\mathcal{S}=(S_1,\dots,S_k)$ and constraints $\mathcal{R}=(\mathcal{R}_1,\dots,\mathcal{R}_k)$  which are closed convex subsets of $\mathcal{S}=(S_1,\dots,S_k)$.  The problem of finding an approximate generalized equilibrium of this game can be transformed into finding an approximate solution to a QVI problem\footnote{This proposition is a simple extension of Harker's transformation (\cite{Harker}) for the approximate case. The proof is available in Appendix \ref{AppendixElementaryProofsection}}. 
\end{proposition}

\begin{corollary}\label{nashvi2}
    
  Suppose that we have a game with the same conditions on  $\mathcal{U}=(u_1,\dots,u_k)$,  and strategies $\mathcal{S}=(S_1,\dots, S_k)$. The problem of finding an approximate Nash equilibrium of this game can be transformed into finding an approximate solution to a VI problem.
\end{corollary}

\subsection{Equilibria and Coalitions}

Significant effort has been dedicated to addressing deviations by coalitions of players in game theory, dating back to the work of Aumann\cite{Strongnash}. Here we review some of the relevant notions. Consider a $k$-player game $(\mathcal{G},\mathcal{U},\mathcal{S})$. Let $\mathcal{J}$ be the set of proper subsets of $\{1, \ldots, k\}$, $J \in \mathcal{J}$ denote a \emph{coalition} and $S_J \equiv \prod_{j \in J} S_j$  \emph{the possible strategy set} for the coalition $J$. Strong Nash equilibrium is particularly useful in areas such as the study of voting systems, in which a certain level of stability and resistance to coordinated deviations is required. Somewhat confusingly, the concept of a strong Nash equilibrium is unrelated to that of a weak Nash equilibrium.

\begin{definition}[\cite{Strongnash}]
    In a game with $k$-players, a strategy profile $\textbf{s}^*$ is considered a \emph{strong Nash equilibrium} if, for any coalition $J \subseteq \mathcal{J}$ and any deviation $\textbf{s}_J'\in \Pi_{j\in J} S_j$ by the players in the coalition $J$ , there exists a player $j\in J$ such that the following condition holds:
\[
u_j\left(\textbf{s}^*\right) \geq u_j\left(\textbf{s}_J', \textbf{s}_{-J}^*\right)
\]where $-J$ denote the complement of $J$ in $[k]$ (which is equal to $[k]-J$).
\end{definition}

\begin{remark}
It is equivalent to consider a strong Nash equilibrium to be a strategy profile in which for any $J\in \mathcal{J}$, there exists no player $j\in J$ and an alternative strategy $\textbf{s}_J'$, such that the condition $u_j\left(\textbf{s}^*\right) < u_j\left(\textbf{s}_J', \textbf{s}_{-J}^*\right)
$ holds.

\end{remark}

\begin{definition}[\cite{StrongkNash}]
     A \emph{$t$-strong Nash equilibrium} is a Nash equilibrium in which no coalition of size $t$ or less of players can deviate so that all its members strictly benefit.
\end{definition}

The idea of a coalition-proof Nash equilibrium is relevant in situations where players can talk about their strategies but can't make binding commitments (see \cite{Coalationproof}). In simple terms, it emphasizes immunity to deviations (alternative strategies) that are self-enforcing.

\begin{definition}[\cite{Coalationproof}]
   For each $\textbf{s}_{-J}^0 \in S_{-J}$, we use $\mathcal{G} / \textbf{s}_{-J}^0$ to denote the game induced on coalition $J$ by the strategies $s_{-J}^0$ for coalition $-J$, i.e., for all $j \in J$  and $\textbf{s}_J \in S_J$ we define the utilities ($\bar{u}_j: S_J \rightarrow \mathbb{R}$) of the members of coalition $J$  to be $ u_j(\textbf{s}_J, \textbf{s}_{-J}^0)$. Next, \emph{self-enforceability} and \emph{coalition-proofness} are defined by mutual recursion as follows:
     \begin{enumerate}
         \item In a single-player game $ s^* \in S_1$ is a coalition-proof Nash equilibrium if and only if $\textbf{s}^*$ maximizes $u_1(\textbf{s})$.
         \item Let $k>1$ and assume that coalition-proof Nash equilibrium has been defined for games with fewer than $k$ players. Then,
         \begin{itemize}
             \item For any game $\mathcal{G}$ with $k$ players, $\textbf{s}^* \in S$ is self-enforcing if, for all $J \in \mathcal{J}, \textbf{s}_J^*$ is a coalition-proof Nash equilibrium in the game $\mathcal{G} / \textbf{s}_{-J}^*$.
             \item  For any game $\mathcal{G}$ with $k$ players, $\textbf{s}^* \in S$ is a coalition-proof Nash equilibrium if it is self-enforcing and if there does not exist another self-enforcing strategy vector $\textbf{s} \in S$ such that $u_j\left(\textbf{s}^*\right)<u_j(\textbf{s})$ for all $j\in [k]$.
         \end{itemize}
     \end{enumerate}
\end{definition}

 Next, we define a more restricted equilibrium notion called $t$-resilient Nash equilibrium, introduced in \cite{Abraham1}, which is fundamental to applications such as secret sharing and multi-party computation. 


\begin{definition}[\cite{Abraham1}]\label{resilient}
 Given a nonempty set $J \subseteq [k] , \textbf{s}^*_J \in S_J$ is a group best response for $J$ to $\textbf{s}^*_{-J} \in S_{-J}$ if, for all $\textbf{s}^\prime_J \in S_J$ and all $j \in J$, we have:
$$
u_j\left(\textbf{s}^*\right) \geq u_j\left(\textbf{s}^\prime_J, \textbf{s}^*_{-J}\right) .
$$ A  strategy profile $\textbf{s}^*\in \mathcal{S}$ is a \emph{$t$-resilient Nash equilibrium} if,  $\forall J \subseteq [k]$ with $|J| \leq t, s^*_J$ is a group best response for $J$ to $\textbf{s}^*_{-J}$\footnote{We call a strategy \emph{strongly resilient} if it is $t$-resilient for all $ t \leq k-1$.}.  
\end{definition}
 
\begin{remark}\label{relationshipremark}
        The $t$-resilient Nash equilibrium problem is the most restricted of these three problems, as every $t$-resilient Nash equilibrium is also a $t$-strong Nash equilibrium. Every strong Nash equilibrium is also inherently a coalition-proof Nash equilibrium. This follows by definition; the set of possible deviations in the latter is a subset of that in the former. $1$-resilient and $1$-strong Nash equilibria are equivalent to the Nash equilibrium problem. Furthermore,  coalition-proof Nash equilibrium for two players is equivalent to the two-player Nash equilibrium problem.  The connection between Nash and coalition-proof Nash equilibrium is less straightforward for more than two players. The paper \cite{Coalationproof} presents an example of a three-player game lacking coalition-proof Nash equilibria. To our knowledge, no other definitive inclusion relationship between these two classes has been established for multiplayer games.  
\end{remark}

\subsection{Multi-Leader-Follower Games}\label{mlfsection}

The multi-leader-follower $L/F$-(Nash) equilibrium is a solution concept for multi-leader-follower games and can be described as a collection of strategies employed by leaders and followers. In this equilibrium, no individual player, whether a leader or a follower, can improve their utility (or minimize their \emph{loss} or \emph{regret} functions) by unilaterally altering their current strategy. Stackelberg games (multi-leader-follower games with one leader) can be seen as mathematical programs with equilibrium constraints (MPEC). In this case, the followers' problems are replaced by a constraint given by their optimality conditions. In a broader context, an MPEC is an optimization problem that encompasses two sets of variables, namely decision variables and response variables \cite{MPEG3,MPEG,MPEG2}. The framework commonly used to represent the multi-leader-follower game is referred to as the equilibrium problem with equilibrium constraints (EPEC). An EPEC \cite{thesis2,Gabriel2012,thesis3,Hu2007,thesis1,SU200774} is essentially an equilibrium problem composed of multiple parametric MPECs, each of which incorporates other players' strategies as parameters.


For simplicity, we consider games that have two leaders, labeled by $\mathrm{I}$ and $\mathrm{II}$, and $k$ followers, labeled by $i \in [k]$.  We denote the \emph{domain} of strategies of leaders $\mathrm{I}$ and $\mathrm{II}$ by $X^{\mathrm{I}}$ and $X^{\mathrm{II}}$ respectively. The leaders’ loss functions are denoted by $\phi_{\mathrm{I}}(x_\mathrm{I}, x_{\mathrm{II}}, y)$ and $\phi_{\mathrm{II}}(x_\mathrm{I}, x_\mathrm{II}, y)$. The notation implies that the loss of each leader is determined by both its own strategies and those of the opposing leader, as well as the strategies of the followers represented as the vector $y$. The followers also respond to the leaders’ strategies as follows.  For each follower $i=1, \dots, k$, let $\theta_i(x_{\mathrm{I}}, x_{\mathrm{II}}, y)$ and $\mathcal{R}_i(x_{\mathrm{I}}, x_{\mathrm{II}}, y_{-i})$ denote follower $i$'s loss function and available constrained strategy set respectively. This strategy set depends on the pair of strategies $(x_{\mathrm{I}}, x_{\mathrm{II}}) \in X^{\mathrm{I}} \times X^{\mathrm{II}}$. For each such pair $(x_{\mathrm{I}}, x_{\mathrm{II}})$, the followers' problem is modeled by a Debreu game (generalized Nash equilibrium problem) parameterized by the leaders' strategies. Let $Y(x_{\mathrm{I}}, x_{\mathrm{II}})$ denote the set of such solutions (not necessarily a singleton). Each element $\bar{y} \in Y(x_{\mathrm{I}}, x_{\mathrm{II}}) \subseteq \Pi_{i=1}^k S_{i}$ is a tuple $(\bar{y}_i)_{i=1}^k$ where for each follower $i$, $\bar{y}_i$ is a solution of problem \ref{taa}. The tuple $(x_{\mathrm{I}}, x_{\mathrm{II}}, \bar{y}_{-i})$ is external to the minimization program \ref{taa}, and $y_i$ is the primary variable that must be computed.\begin{equation}\label{taa}
    \begin{aligned}
& \operatorname{Min} \theta_i\left(x_{\mathrm{I}}, x_{\mathrm{II}}, \bar{y}_{-i}, y_i\right) \\
& \text { s.t } y_i \in \mathcal{R}_i\left(x_{\mathrm{I}}, x_{\mathrm{II}}, \bar{y}_{-i}\right),
\end{aligned}
\end{equation}

We now define the concept of equilibrium in multi-leader-follower games. A pair $(x^{*}_\mathrm{I}, x^{*}_\mathrm{II}) \in X^{\mathrm{I}} \times X^{\mathrm{II}}$ is called a \emph{$L/F$-equilibrium}, if there exists $(y^{*}_\mathrm{I}, y^{*}_\mathrm{II})$ such that $(x^{*}_\mathrm{I}, y^{*}_\mathrm{I})$ is an optimal solution of leader $\mathrm{I}$'s problem, which tries to find a pair $(x_{\mathrm{I}}, y_{\mathrm{I}})$ to the following:
\begin{equation}
\begin{aligned}
& \operatorname{Min} \phi_{\mathrm{I}}\left(x_{\mathrm{I}}, x^{*}_\mathrm{II}, y_{\mathrm{I}}\right) \\
& \text { s.t } x_{\mathrm{I}} \in X^{\mathrm{I}} \\
& \text { and } y_{\mathrm{I}} \in Y\left(x_{\mathrm{I}}, x^{*}_\mathrm{II}\right)
\end{aligned}
\end{equation}
and $(x^{*}_\mathrm{II}, y^{*}_\mathrm{II})$ is an optimal solution of leader $\mathrm{II}$'s problem, which tries to find a pair $\left(x_{\mathrm{II}}, y_{\mathrm{II}}\right)$ to the following problem:
\begin{equation}
\begin{aligned}
& \operatorname{Min} \phi_{\mathrm{II}}\left(x^{*}_\mathrm{I}, x_{\mathrm{II}}, y_{\mathrm{II}}\right) \\
& \text { s.t } x_{\mathrm{II}} \in X^{\mathrm{II}} \\
& \text { and } \quad y_{\mathrm{II}} \in Y\left(x^{*}_\mathrm{I}, x_{\mathrm{II}}\right)
\end{aligned}
\end{equation}
In the given definition,  the equilibrium strategies of the followers, represented as $y^*_{\mathrm{I}}$ and $y^*_{\mathrm{II}}$, belong to the same set $Y(x^*_{\mathrm{I}}, x^*_{\mathrm{II}})$. However, they are not required to be the same. This flexibility arises because these strategies are based on different anticipations by Leader $\mathrm{I}$ and Leader $\mathrm{II}$ of how the followers collectively respond to the pair of strategies $(x^*_{\mathrm{I}}, x^*_{\mathrm{II}})$.   One could introduce another variant of this problem by enforcing $y^*_{\mathrm{I}}=y^*_{\mathrm{II}}$. However, even in cases where $Y(x_{\mathrm{I}}, x_{\mathrm{II}})$ contains only one unique response for all pairs of strategies $(x_{\mathrm{I}}, x_{\mathrm{II}})$, a $L/ F$-equilibrium may not exist, as demonstrated in \cite{Pang2005}.

\subsection{Some Computational Complexity Definitions}

\begin{definition}
	TFNP (Total Functional Nondeterministic Polynomial Time) is the class of total search problems with solutions that are poly-time verifiable. Formally, given a poly-balanced poly-time relation $R(x,y)$ the associated  \emph{NP search problem} is the (partial) multi-valued function $Q(x)=\{y~|~R(x,y)\}$. The problem is \emph{total} if $Q(x)$ is nonempty for all $x$. The class \emph{FNP} consists of all NP search problems, while TFNP consists of all total NP search problems. A search problem $Q$ is in \emph{FP} if there is a  poly-time function $Q^\prime$ such that for all $x$, $Q^\prime(x)\in Q(x)$. The \emph{decision problem} associated with $R$ is to determine, given $x$, whether there is some $y$ such that $R(x,y)$. NP is the class of all such decision problems.
	
\end{definition}

\begin{definition}

We can extend the notion of many-one reduction to total search problems as follows:
For two total search problems $R$ and $S$ we say $R \leq_{m} S$ if there exist poly-time computable functions  $f, g$ such that for all $x,y$  if $(f(x), y) \in S$  then $(x, g(y)) \in R$.

\end{definition}

\begin{remark}
  This form of reduction between search problems is equivalent to a Cook reduction with one call to the oracle.
\end{remark}

The subclass PPAD of TFNP plays a central role in the complexity of the equilibrium problems in game theory and is defined as follows:
\begin{definition}
The \emph{end-of-the-line (EOTL)} problem is defined as follows:
Given a directed graph $\mathcal{G}$, represented by two poly-sized circuits which return the predecessor and the successor of a node (represented in binary), where each vertex has at most one predecessor and successor and a vertex $u$ in $\mathcal{G}$ with no predecessor, find another vertex $v\neq u$ with no predecessor or no successor. A total search problem is in \emph{PPAD} if it is many-one reducible to EOTL.
	
\end{definition}


  Linear arithmetic circuits are a restriction of general arithmetic circuits which do not allow general multiplication gates \cite{Etessami,Fearnley}. Such circuits can efficiently approximate any polynomially computable function; in particular, they can approximate functions represented by well-behaved general circuits (for more information see Appendix \ref{arithmeticcircuits}).  Linear arithmetic circuits have a variety of useful properties such as Lipschitzness.
 
\begin{definition}
    A \emph{linear arithmetic circuit} $C$ is a circuit represented as a directed acyclic graph with nodes labeled either as input nodes or as output nodes or as gate nodes with one of the following possible gates $\{+,-, \min , \max , \times \zeta\}$, where the $\times \zeta$ gate refers to the multiplication by a constant. Also, rational constants are allowed. We use $\operatorname{size}(C)$ to refer to the number of nodes of $C$ also including the constants.
\end{definition}


\section{Computational Complexity of Variational Inequalities}\label{AppendixvariationalSection}

In this section, we define the computational version of the generalized quasi-variational inequality problem (GQVI) and establish its PPAD-completeness.


\subsubsection*{Strong Separation Oracles}

We implicitly represent the convex set $X$ via a polynomially-sized circuit that computes (weak/strong) separation oracles for $X$. While this approach captures most game theoretic applications, it introduces technical difficulties, in particular when dealing with projection errors and optimizing convex functions \cite{concavegames}. Without loss of generality, we restrict our attention to the metric space $(\mathbb{R}^m, L2)$  where the inputs and the outputs are restricted to a well-bounded compact box denoted by $\mathbb{R}^{m*}$.   For simplicity, initially, we investigate strong separation oracles defined as follows:


\problemStatement{ A strong Separation Oracle (via a circuit $C_{\mathcal{R}(x)}$)}{
  Input={A vector $z \in \mathbb{Q}^m \cap \mathbb{R}^{m*}$.
  },
  Output={$(a, b) \in \mathbb{Q}^m \times \mathbb{Q}$ such that the threshold $b \in[0,1] \cap \mathbb{Q}$ denotes the membership of $z$ in $\mathcal{R}(x)$. More precisely:
  \begin{itemize}
      \item If $z \in \mathcal{R}(x)$ then $b>\frac{1}{2}$ and the vector $a \in \mathbb{Q}^m$ will be $\bot$. In other words, $a$ is meaningful only when $b \leq \frac{1}{2}$
\item  $b \leq \frac{1}{2}$ and vector $a$ , with $\|a\|_{\infty}=1$, defines a  separating hyperplane $\mathcal{H}(a, z):=\left\{y \in\mathbb{R}^{m*}:\langle a, y-z\rangle=0\right\}$ between the vector $z$ and the set $\mathcal{R}(x)$ such that $\langle a, y-z\rangle \leq 0$ for every $y \in \mathcal{R}(x)$.
  \end{itemize}
  }
}

 Towards proving PPAD-membership of quasi-variational inequalities, we will need to use and modify some computational problems that were introduced in \cite{concavegames}. These problems can be solved using the sub-gradient ellipsoid central cut method \cite{Grostel,concavegames}.  This method requires access to the sub-gradients of the functions that we do optimization on. Given the set of gates that can appear in a linear arithmetic circuit, such a circuit is differentiable if and only if it represents a linear (affine) function. In general, due to the existence of $max$ or $min$ gates, the linear arithmetic circuits may not be differentiable. However, as noted in \cite{concavegames}, it is possible to compute a vector belonging to the subgradients of these functions using automatic differentiation techniques, eliminating the need for a dedicated circuit to perform this computation \cite{Barton2018}.

\problemStatement{Strong Constrained Convex Optimization Problem}{
  Input={A zeroth and first order oracle for the convex function $F: \mathbb{R}^{m*} \rightarrow \mathbb{R}$, two rational numbers $\delta,\epsilon>0$ and a strong separation oracle $\mathrm{SO}_{\mathcal{R}}$ for a non-empty closed convex set $\mathcal{R} \subseteq$  $\mathbb{R}^{m*}$.
  },
  Output={One of the following cases:
  \begin{itemize}
      \item (Violation of non-emptiness):A failure symbol $\perp$ with a polynomial-sized witness that certifies that $\overline{\mathrm{B}}(\mathcal{R},-\delta)=\emptyset$.
\item (Approximate minimization):A vector $z \in \mathbb{Q}^m \cap \mathcal{R}$, such that $F(z) \leq \min _{y \in \mathcal{R}} F(y)+\epsilon$.
  \end{itemize}
  }
}

The strong approximate version of the projection problem, denoted by $\widetilde{\Pi}^{\epsilon}_{X}(x)$, is an instance of a strong constrained optimization problem.
\problemStatement{Strong Approximate Projection Problem}{
  Input={A rational number $\epsilon>0$ and a strong separation oracle $\mathrm{SO}_{X}$ for a non-empty closed convex set $X \subseteq \mathbb{R}^{m*}$ and a vector $x$ that belongs to $\mathbb{Q}^m \cap X$. 
  },
  Output={One of the following cases:
  \begin{itemize}
      \item (Violation of non-emptiness)
A failure symbol $\perp$ followed by a polynomial-sized witness that certifies that $\overline{\mathrm{B}}(X,-\epsilon)=\emptyset$,
\item  (Approximate Projection)
A vector $z \in \mathbb{Q}^m \cap X$, such that:
$$
\|z-x\|_2^2 \leq \min _{y \in X}\|x-y\|_2^2+\epsilon
$$
  \end{itemize}
  }
}

\subsection{Computational Generalized Variational Inequalities (GQVI)}

 The computational complexity of solving variational inequalities is influenced by the problem's dimension, the properties of the set-valued mapping, and the representation of sets and mappings.  We follow the assumption that the sets are bounded by a box that is a subset of $\mathbb{R}^m$ and denote this by $\mathbb{R}^{m*}$.  We defer the computational definition of QVI and VI to Section \ref{Espcases}.  To satisfy some essential properties such as Lipschitz continuity of the correspondences and functions syntactically, we assume that the inputs are given as linear arithmetic circuits\footnote{Note that given the discussions of \cite{Fearnley,concavegames}, we do not lose representation power (see Appendix \ref{arithmeticcircuits})}. We also can consider circuits that are sums of monomials, for which Lipschitz continuity and computing the sub-gradients is straightforward.

\begin{remark}\label{remarkQVI}
    The case (Violation of convexity) is meaningful as output whenever the form of utilities is explicitly given (e.g. sum of monomials) otherwise convexity holds as a promise. The given approximate solution does not satisfy the properties $x^* \in \mathcal{R}(x^*)$. However, this assumption is not required for the applications we study in game theory. We also can impose reasonable assumptions on this problem (GQVI) such as $\mathcal{R}$ being a symmetric correspondence (i.e, if $y\in \mathcal{R}(x)$, then $x\in \mathcal{R}(y)$), however, this might detract generality of the problem.
\end{remark}

\problemStatement{$GQVI(\mathcal{F},\mathcal{R})$ With Linear Arithmetic Circuits Representing SO}{
  Input={ We receive as input all the following:
  \begin{itemize}
      \item  A linear arithmetic circuit $C_{\mathcal{R}}$ which represents a strong separation oracle for a correspondence $\mathcal{R}: \mathbb{R}^{m*} \rightrightarrows  \mathbb{R}^{m*}$,
      \item A linear arithmetic circuit $C_\mathcal{F}$ which represents a strong separation oracle for a  $\gamma$-strongly convex valued correspondence $\mathcal{F}: \mathbb{R}^{m*} \rightrightarrows  \mathbb{R}^{m*}$,
      \item    An accuracy parameter $\beta$. 
  \end{itemize}},
  Output={One of the following cases:
  \begin{itemize}
      \item (Violation of non-emptiness): A vector $x \in \mathbb{R}^{m*}$ such that $\overline{\mathrm{B}}(\mathcal{R}(x),-\epsilon)=\emptyset$ or $x \in \mathbb{R}^{m*}$ such that $\overline{\mathrm{B}}(\mathcal{F}(x),-\epsilon)=\emptyset$,
\item 
(Violation of $\gamma$-strong convexity):
There vectors $x, p, q\in \mathbb{R}^{m*}$ and two constants $\epsilon>0$ and $\lambda\in (0,1)$ such that:
$$
\begin{aligned}
\widetilde{\Pi}_{\mathcal{F}(x)}^{\epsilon}(\lambda p+ (1-\lambda)q ) & > \lambda \cdot \widetilde{\Pi}_{\mathcal{F}(x)}^{\epsilon}(p)+(1-\lambda) \cdot \widetilde{\Pi}_{\mathcal{F}(x)}^{\epsilon}(q)
\\&-\frac{\lambda(1-\lambda)}{2} \cdot \gamma\cdot\left\|p-q\right\|_2^2+\epsilon
\end{aligned}
$$
\item Two tuples $(x,w)$ and $(x^*,w^*)$ with $\|(x,w)-(x^*,w^*)\|^2_2\leq \beta $ such that  $x^*\in \mathcal{R}(x) $ and  $w^*\in \mathcal{F}(x)$ and $\left(y-x\right) ^T w^* +\beta \geq 0, \quad \forall y \in \mathcal{R}(x)$   
  \end{itemize}
  }
}

The following theorem states that the computational version of GQVI is PPAD-complete. Our approach to establishing PPAD membership for variational inequality problems leverages the computational version of Kakutani's fixed point theorem and the robust version of Berge's maximum theorem noted in \cite{concavegames} . Our approach, while similar to that of \cite{concavegames}, has key adaptations. Specifically, we utilize \cite{chan}'s proof of the existence of Generalized Quasi-Variational Inequalities (GQVI) instead of relying on Rosen's existence theorem \cite{rosen}. This choice leads to the construction of a different correspondence, which we subsequently send to computational version of Kakutani's fixed point theorem. Note that the paper \cite{chan} uses the Eilenberg-Montgomery fixed-point theorem to establish the existence of GQVI, while we instead use Kakutani's fixed-point theorem. PPAD-hardness follows directly from the PPAD-hardness of problems such as Nash equilibrium (see Corollary \ref{nashvi2}), which are easily formulated using a simpler format of variational inequalities (VI). The proof of inclusion in PPAD is relegated to Section \ref{Variationalsectionproof}.

\begin{theorem}\label{gqviinppad}
The above-mentioned generalized quasi-variational inequality problem $GQVI(\mathcal{F},\mathcal{R})$ for linear arithmetic circuits representing $\mathrm{SO}_{\mathcal{R}}$ and $\mathrm{SO}_\mathcal{F}$ is  PPAD-complete.
\end{theorem}

\section{Applications In Game Theory}

 Next, applying the computational versions of the variational inequality problems, we demonstrate that several central problems in game theory are PPAD-complete under reasonable assumptions. We assume that the separation oracles are given as linear arithmetic circuits to represent convex sets and correspondences.   First, we have the problem of finding a remedial solution to $L/F$-equilibrium for multi-leader-follower games. Subsequently, we demonstrate that given specific continuity and concavity conditions (in particular a more restricted condition that we call multi-concavity), notions such as $t$-resilient Nash and some other related notions can be formulated by a modified version of variational inequalities. 


\subsection{Computational Complexity of Multi-Leader-Follower Games}\label{lfremdial}

In addition to the complexity of the leaders' optimization problem, the non-convexity of the leader's constraints may mean that a $L/F$-equilibrium does not always exist. We address the non-convexity of the leader's constraints by following an algebraic relaxation/restriction approach introduced in \cite{Pang2005} which remedies the non-existence of a $L/F$-equilibrium by using a \emph{Karush–Kuhn–Tucker condition} (KKT) representation of the graph of $Y$ and additionally introducing various relaxations/restrictions of such a representation\footnote{Note that we only investigate the approximate version for the remedial model of \cite{Pang2005}. There are a variety of hardness results for Stackelberg games (e.g. \cite{stackplex,Marchesi2021}).}. As discussed, we  will assume follower $i$'s constraints are represented by a non-empty set.
\[\mathcal{R}_{i}\left(x_{\mathrm{I}}, x_{\mathrm{II}}, y_{-i}\right) \equiv\left\{y_{i} \in S_i ~|~ g_{i}\left(x_{\mathrm{I}}, x_{\mathrm{II}}, y\right) \leq 0,\thickspace h_{i}\left(x_{\mathrm{I}}, x_{\mathrm{II}}, y_{i}\right) \leq 0\right\}\]
where $x_{\mathrm{I}}, x_{\mathrm{II}} \in X^{\mathrm{I}} \times X^{\mathrm{II}}$ are the leaders' strategies and  $g_{i}:X^{\mathrm{I}} \times X^{\mathrm{II}} \times \Pi_{j=1 }^k S_j \rightarrow \mathbb{R}^{m_i}$ and $h_{i}: X^{\mathrm{I}} \times X^{\mathrm{II}} \times S_i \rightarrow \mathbb{R}^{l_i}$ are continuously  Lipschitz and differentiable. Note that $m_i$ and $l_i$ are some given constants. By incorporating the KKT conditions of the followers' optimization problems with $x^*_{\mathrm{II}}$ as an exogenous variable (representing second leader's optimal solution), leader $\mathrm{I}$'s optimization problem  can be formulated as finding a solution $(x_\mathrm{I},y_\mathrm{I},\lambda_{(i,\mathrm{I})},\mu_{(i,\mathrm{I})})$   to the following problem:
\begin{equation*}\label{ORGS}
\resizebox{0.961\textwidth}{!}{$
\begin{aligned}
&\operatorname{Min} \phi_{\mathrm{I}}\left(x_{\mathrm{I}}, x^*_{\mathrm{II}}, y_{\mathrm{I}}\right) \\
& \text{ s.t } x_{\mathrm{I}} \in X^{\mathrm{I}} \\
& \forall i \in [k], \thickspace \partial_{y_{(i,\mathrm{I})}} \theta_{i}\left(x_{\mathrm{I}}, x^*_{\mathrm{II}}, y_{\mathrm{I}}\right) + \sum_{j=1}^{m_{i}} \lambda^{j}_{(i,\mathrm{I})} \partial_{y_{(i,\mathrm{I})}} g^j_{i}\left(x_{\mathrm{I}}, x^*_{\mathrm{II}}, y_{\mathrm{I}}\right) +\sum_{j=1}^{l_i} \mu^{j}_{(i,\mathrm{I})} \partial_{y_{(i,\mathrm{I})}} h^j_{i}(x_{\mathrm{I}}, x^*_{\mathrm{II}}, y_{(i,\mathrm{I})})=0 \\
& \forall i \in [k], \thickspace 0 \leq \lambda_{(i,\mathrm{I})} \perp g_{i}\left(x_{\mathrm{I}}, x^*_{\mathrm{II}}, y_{\mathrm{I}}\right) \leq 0 \\
& \forall i \in [k], \thickspace 0 \leq \mu_{(i,\mathrm{I})} \perp h_{i}(x_{\mathrm{I}}, x^*_{\mathrm{II}}, y_{(i,\mathrm{I})}) \leq 0
\end{aligned}
$}
\end{equation*}

In general, the constraints in the aforementioned problem lack convexity with respect to their variables. This difficulty is addressed in \cite{Pang2005}, by considering a restriction that still encompasses a wide range of practical models. More precisely, we will assume that the functions $\theta_{i}$, $g_{i}$, and $h_{i}$ take the following forms:
\begin{equation}\label{thera}
\resizebox{0.961\textwidth}{!}{$
\begin{aligned}
   & \theta_{i}\left(x_{\mathrm{I}}, x_{\mathrm{II}}, y\right) \equiv \frac{1}{2}\left(y_{i}\right)^{T} M_{i} y_{i} + \left(c_{i}\left(x_{\mathrm{I}}, x_{\mathrm{II}}, y_{-i}\right)\right)^{T} y_{i} + \psi_{i} \left(x_{\mathrm{I}}, x_{\mathrm{II}}, y_{-i}\right) \\
   & g_{i}\left(x_{\mathrm{I}}, x_{\mathrm{II}}, y\right) \equiv A_{(i,\mathrm{I})} x_{\mathrm{I}} + A_{(i,\mathrm{II})} x_{\mathrm{II}} + \sum_{j=1}^{k} B_{i, j} y_{j} \text{ and }
   h_{i}\left(x_{\mathrm{I}}, x_{\mathrm{II}}, y_{i}\right) \equiv C_{(i,\mathrm{I})} x_{\mathrm{I}} + C_{(i,\mathrm{II})} x_{\mathrm{II}} + D_{i} y_{i}
\end{aligned}
$}
\end{equation}
for some matrices $M_{i} \in \mathbb{R}^{n_{i} \times n_{i}}$, which are symmetric positive semidefinite, $A_{(i,\mathrm{I})} \in \mathbb{R}^{m_{i} \times n_{\mathrm{I}}}$,  $A_{(i,\mathrm{II})} \in \mathbb{R}^{m_{i} \times n_{\mathrm{II}}}$, $B_{i, j} \in \mathbb{R}^{m_{i} \times n_{j}}$, $C_{(i,\mathrm{I})} \in \mathbb{R}^{l_{i} \times n_{\mathrm{I}}}$, $C_{(i,\mathrm{II})} \in \mathbb{R}^{l_{i} \times n_{\mathrm{II}}}$, $D_{i} \in \mathbb{R}^{l_{i} \times n_{i}}$, affine functions $c_{i}: \mathbb{R}^{n_\mathrm{I}+n_\mathrm{II}+n_{-i}} \rightarrow \mathbb{R}^{n_{i}}$, and arbitrary real-valued functions $\psi_i: \mathbb{R}^{n_\mathrm{I}+n_{\mathrm{II}}+n_{-i}} \rightarrow \mathbb{R}$. 

\begin{remark}
    Recall that $n_\mathrm{I}$, $n_\mathrm{II}$ and $n_i$ denote the variables that the leaders and the follower $i$ control (out of $n_\mathrm{I}+n_\mathrm{II}+n$ where $n=\sum_{i=1}^{k} n_i$) respectively. Furthermore, $n_{-i}$ denote the variables that all other followers other than $i$ control. 
\end{remark}

With these modifications, we formulate leader $\mathrm{I}$'s optimization problem as follows: With $x^*_{\mathrm{II}}$ as an exogenous variable ,  find  a solution $(x_\mathrm{I},y_\mathrm{I},\lambda_{(i,\mathrm{I})},\mu_{(i,\mathrm{I})})\in X^{\mathrm{I}} \times X^{\mathrm{II}} \times \Pi_{j=1}^k S_j \times \mathbb{R}^{m_{i}+l_{i}}$ to the following problem:
\begin{equation}\label{naghes}
    \begin{aligned}
        &\operatorname{Min} \phi_{\mathrm{I}}\left(x_{\mathrm{I}}, x^*_{\mathrm{II}}, y_{\mathrm{I}}\right) \\&
 \text{ s.t } x_{\mathrm{I}} \in X^{\mathrm{I}}\\
   &c_{i}\left(x_{\mathrm{I}}, x^*_{\mathrm{II}}, y_{(-i,\mathrm{I})}\right)+  M_{i} y_{(i,\mathrm{I})}+(B_{i,i})^T \lambda_{(i,\mathrm{I})} +D_i^T \mu_{(i,\mathrm{I})}=0\\ 
   &\forall i \in [k],\quad \left(\lambda_{(i,\mathrm{I})}\right)^T\left[A_{(i,\mathrm{I})} x_{\mathrm{I}}+A_{(i,\mathrm{I})} x^*_{\mathrm{II}}+\sum_{j=1}^k B_{i, j} y_{(j,\mathrm{I})}\right]=0\\
     &\forall i \in [k],\quad\left(\mu_{(i,\mathrm{I})}\right)^T\left[C_{(i,\mathrm{I})} x_{\mathrm{I}}+C_{(i,\mathrm{I})} x^*_{\mathrm{II}}+D_i y_{(i,\mathrm{I})} \right]=0
\end{aligned}
\end{equation}

Recall that for leader $\mathrm{I}$, we denote an anticipated strategy of all $k$ followers except $i$ by $y_{(-i,\mathrm{I})}$. With the exception of the set $X^{\mathrm{I}}$, which might have nonlinear (though convex) constraints, and the orthogonality conditions, the remaining constraints in leader $\mathrm{I}$'s problem are linear. Informally, the approach described in \cite{Pang2005} for handling these nonconvex constraints is described as follows: Each leader has access to limited information about the followers' reactions, characterized by specific "favorable" sets  $Z^{\mathrm{I}}(x_\mathrm{I}, x_\mathrm{II})$ and $Z^{\mathrm{II}}(x_{\mathrm{I}}, x_{\mathrm{II}})$, respectively, and subject to which the leaders optimize their objective functions.  In general, such partial information could be of one of two kinds, \emph{restricted} or \emph{relaxed} information as follows:
\begin{equation}\label{restrictrel}
    \begin{gathered}
Z^{\mathrm{I}}\left(x_{\mathrm{I}}, x_{\mathrm{II}}\right) \subseteq Y\left(x_{\mathrm{I}}, x_{\mathrm{II}}\right) \mbox{ or }
Y\left(x_{\mathrm{I}}, x_{\mathrm{II}}\right) \subseteq Z^{\mathrm{I}}\left(x_{\mathrm{I}}, x_{\mathrm{II}}\right)
\end{gathered}
\end{equation}

A restricted response is an equilibrium response, whereas a relaxed one may not be. 
 Similar classifications apply to the partial responses available to leader $\mathrm{II}$. By separating their responses, one leader may have restricted follower responses while the other has relaxed ones. The reference \cite{Pang2005} suggests substituting the two non-convex orthogonality conditions in leader $\mathrm{I}$'s constraints with the following  conditions:
\[
\left(x_{\mathrm{I}}, y_{\mathrm{I}}, \lambda_{(i,\mathrm{I})}\right) \in W_{(i,\mathrm{I})}\left(x_{\mathrm{II}}\right) \quad \text { and } \quad\left(x_{\mathrm{I}}, y_{(i,\mathrm{I})}, \mu_{(i,\mathrm{I})}\right) \in V_{(i,\mathrm{I})}\left(x_{\mathrm{II}}\right)
\]
where $W_{(i,\mathrm{I})}\left(x_{\mathrm{II}}\right) \subseteq  X^{\mathrm{I}} \times \Pi_{j=1}^k S_j \times \mathbb{R}^{m_{i}} $ and $V_{(i,\mathrm{I})}\left(x_{\mathrm{II}}\right) \subseteq  X^{\mathrm{I}} \times S_i \times \mathbb{R}^{l_{i}}$ are appropriate polyhedral sets such that together they represent either a restriction or a relaxation of the complementarity constraints in $Y(x_{\mathrm{I}}, x^*_{\mathrm{II}})$. For all $i=1, \dots, k$, we define $Z^{\mathrm{I}}\left(x_{\mathrm{I}}, x_{\mathrm{II}}\right)$ to be the set of all tuples $y_{\mathrm{I}}=\left(y_{(j,\mathrm{I})}\right)_{j=1}^{k}$ for which there exists $\left(\lambda_{\mathrm{I}}, \mu_{\mathrm{I}}\right)$ that satisfies:
\[
\begin{gathered}
c_{i}\left(x_{\mathrm{I}}, x_{\mathrm{II}}, y_{(-i,\mathrm{I})}\right)+M_{i} y_{(i,\mathrm{I})}+\left(B_{i, i}\right)^{T} \lambda_{(i,\mathrm{I})}+\left(D_{i}\right)^{T} \mu_{(i,\mathrm{I})}=0 \\
0 \leq \lambda_{(i,\mathrm{I})}, \quad A_{i,\mathrm{I}} x_{\mathrm{I}}+A_{i,\mathrm{II}} x_{\mathrm{II}}+\sum_{j=1}^{k} B_{i, j} y_{(j,\mathrm{I})} \leq 0 \\
0 \leq \mu_{(i,\mathrm{I})}, \quad C_{(i,\mathrm{I})} x_{\mathrm{I}}+C_{(i,\mathrm{II})} x_{\mathrm{II}}+D_{i} y_{(i,\mathrm{I})} \leq 0
\end{gathered}
\]

We call the elements of $Z^{\mathrm{I}}\left(x^*_{\mathrm{I}}, x_{\mathrm{II}}\right)$ the followers' partial responses anticipated by (or available to) leader $\mathrm{I}$ and categorize such responses as restricted or relaxed based on Equation \ref{restrictrel}\footnote{We can construct a strong separation oracle for this set given the polyhedral sets' separation oracles.}.

Finally, in terms of the partial responses, leader $\mathrm{I}$, given $x^*_{\mathrm{II}}$ as an exogenous variable, must find a solution
$\left(x_{\mathrm{I}}, y_{\mathrm{I}}\right)$ to the following surrogate optimization problem:
\begin{equation}\label{leader1O}
\begin{aligned}
& \operatorname{Min} \phi_{\mathrm{I}}\left(x_{\mathrm{I}}, x^*_{\mathrm{II}}, y_{\mathrm{I}}\right) \\
& \text { s.t } x_{\mathrm{I}} \in X^{\mathrm{I}} \\
& \text { and } \quad\left(x_{\mathrm{I}}, y_{\mathrm{I}}\right) \in \operatorname{graph} Z^{\mathrm{I}}\left(\cdot, x^*_{\mathrm{II}}\right)
\end{aligned}
\end{equation}

For leader $\mathrm{II}$, the optimization problem with the surrogate complementarity conditions  can be defined similarly by considering $Z^{\mathrm{II}}(x^*_\mathrm{I},x_\mathrm{II})$, $ W_{(i,\mathrm{II})}(x^*_{\mathrm{I}})$ and $V_{(i,\mathrm{II})}(x^*_{\mathrm{I}})$ where $x^*_{\mathrm{I}}$ will be an exogenous variable:
\begin{equation}\label{leader2O}
  \begin{aligned}
& \operatorname{Min} \phi_{\mathrm{II}}\left(x^*_{\mathrm{I}}, x_{\mathrm{II}}, y_{\mathrm{I}}\right) \\
& \text { s.t } x_{\mathrm{II}} \in X^{\mathrm{II}} \\
& \text { and } \quad\left(x_{\mathrm{II}}, y_{\mathrm{II}}\right) \in \operatorname{graph} Z^{\mathrm{II}}\left( x^*_{\mathrm{I}},\cdot \right)
\end{aligned}
\end{equation}

\begin{remark}\label{remarkG}
   Note that by combining the given constraints we are able to write Equation \ref{leader1O} in terms of  correspondences $G_\mathrm{I}$ and $G_\mathrm{II}$. In particular, the optimization problem for leader $\mathrm{I}$ can be written as:
\begin{equation}\label{leader1G}
\begin{aligned}
& \operatorname{Min} \phi_{\mathrm{I}}\left(x_{\mathrm{I}}, x^*_{\mathrm{II}}, y_{\mathrm{I}}\right) \\
& \text { s.t } \left(x_{\mathrm{I}}, y_{\mathrm{I}}\right) \in G_\mathrm{I} \left(x^*_{\mathrm{II}}\right)
\end{aligned}
\end{equation}
and for leader $\mathrm{II}$, we have:
\begin{equation}\label{leader2G}
\begin{aligned}
& \operatorname{Min} \phi_{\mathrm{II}}\left(x^*_{\mathrm{I}}, x_{\mathrm{II}}, y_{\mathrm{I}}\right) \\
& \text { s.t } \left(x_{\mathrm{II}}, y_{\mathrm{II}}\right) \in G_\mathrm{II}  \left(x^*_{\mathrm{I}}\right)
\end{aligned}
\end{equation}

\end{remark}

\begin{remark}
  A solution of the optimization problem in Remark \ref{remarkG} is a generalized Nash equilibrium (for more information see proof of Proposition  \ref{geneqvari2}).
\end{remark}

\begin{definition}
    A \emph{remedial $L/F$-equilibrium} is a pair of leaders' strategies $(x^*_{\mathrm{I}}, x^*_{\mathrm{II}})$ such that there exists a pair $(y^*_{\mathrm{I}}, y^*_ {\mathrm{II}})$ such that $(x^*_{\mathrm{I}}, y^*_ {\mathrm{I}})$ and $(x^*_{\mathrm{II}}, y^*_ {\mathrm{II}})$ constitute a solution of the two leaders' surrogate optimization problems \ref{leader1O} and \ref{leader2O}, respectively.
\end{definition}

We define $Sol(\ref{leader1O})$ and $Sol(\ref{leader2O})$ to be the minimized value (of $\phi_{\mathrm{I}}$ and $\phi_{\mathrm{II}}$ given their respective exogenous variables) of a solution of the two leaders' surrogate optimization problems \ref{leader1O} and \ref{leader2O}.  We are now ready to define the computational version of finding an equilibrium state a PPAD-completeness result in this setting.

We could consider alternative definitions that omit $Z^{\mathrm{I}}$ and $Z^{\mathrm{II}}$ and consider strong separation oracles for the above-mentioned appropriate polyhedral sets as inputs. This is because all of the requirements for constructing separation oracles for $Z^{\mathrm{I}}$ and $Z^{\mathrm{II}}$ can be assured by the linearity of the loss functions of the followers and having the possibility of representing the approximate polyhedral sets by separation oracles. We chose this version for the sake of a simpler presentation.

\begin{remark}
   The non-emptiness exceptions can be distinguished by a slightly generalized ellipsoid algorithm we discuss in Appendix \ref{AppendixCentralcutellipsoidsection}.
\end{remark}

\problemStatement{Remedial $L/F$ Equilibrium With  Strong Separation Oracles}{
  Input={We receive as input all the following:
  \begin{itemize}
   \item  Two linear arithmetic circuits representing the convex loss functions  $(\phi_{\mathrm{I}}$,$\phi_{\mathrm{II}})$ for two leaders,
      \item Two linear arithmetic circuits representing strong separation oracles for non-empty, convex, and compact sets $X^{\mathrm{I}}$ and $X^{\mathrm{II}}$ for the leaders $\mathrm{I}$ and $\mathrm{II}$,
      \item Two linear arithmetic circuits representing strong separation oracles for  $Z^{\mathrm{I}}$ and $Z^{\mathrm{II}}$ that represent restricted or relaxed responses of the followers for each leader respectively that are two non-empty, convex-valued, and compact correspondences. 
      \item  An accuracy parameter $\beta$. 
  \end{itemize}},
  Output={One of the following cases:
  \begin{itemize}
      \item (Violation of non-emptiness): A certificate indicating at least one of the $X^{\mathrm{I}}$, $X^{\mathrm{II}}$, $Z^{\mathrm{I}}\left(x,\cdot\right)$ for some $x\in X^{\mathrm{I}}$ or $Z^{\mathrm{II}}\left( \cdot,x\right)$ for some $x\in X^{\mathrm{II}}$ is empty.
       \item (Violation of convexity) of any of the loss functions of the inputs,
\item (Approximate minimization): 
Vectors $(x^*_{\mathrm{I}}, y^*_{\mathrm{I}}, x^*_ {\mathrm{II}},y^*_ {\mathrm{II}})$  having the following relationship:
\begin{itemize}
    \item $\phi_{\mathrm{I}}\left(x^*_{\mathrm{I}}, x^*_{\mathrm{II}}, y^*_{\mathrm{I}}\right)\leq \beta+ \text{ Sol(\ref{leader1O}) }  $
    \item $\phi_{\mathrm{II}}\left(x^*_{\mathrm{I}}, x^*_{\mathrm{II}}, y^*_{\mathrm{II}}\right)\leq \beta+ \text{ Sol(\ref{leader2O}) }  $
\end{itemize}
  \end{itemize}
  }
}

Before proceeding to proof of PPAD-completeness of the above-mentioned problem, given the format proposed in Remark \ref{remarkG}, we might assume that we can use Proposition \ref{geneqvari2} which can transform this problem to the QVI format under some conditions. However, this proposition assumes the differentiability of the loss functions (utilities). Since these are now represented by linear arithmetic circuits, which as mentioned above are not necessarily differentiable, we need to take a more general approach. To begin, we have the following:

\begin{proposition}\label{GQVIAPPRox}
  Suppose that we have a game with convex loss functions $\mathcal{\phi}=(\phi_\mathrm{1},\dots,\phi_\mathrm{k})$,  and strategies $\mathcal{X}=(X^\mathrm{1},\dots,X^\mathrm{k})$ represented by linear arithmetic circuits and constraints $\mathcal{G}=(G_\mathrm{1},\dots,G_\mathrm{k})$ represented by linear arithmetic circuits which are closed convex subsets of $\mathcal{X}=(X^\mathrm{1},\dots,X^\mathrm{k})$.  The problem of finding an approximate generalized (Nash) equilibrium of this game can be transformed into finding an approximate solution to a GQVI problem. 
\end{proposition}
\begin{proof}
For simplicity, we prove the proposition for $k=2$ players.  Define the the set-valued mapping $\mathcal{F}(x):\mathbb{R}^{n_\mathrm{I}+n_\mathrm{II}}\rightarrow \mathcal{P}\left(\mathbb{R}^{n_\mathrm{I}+n_\mathrm{II}}\right)$ to be:
$$\mathcal{F}(x)=\partial_{x_\mathrm{I}} \phi_{\mathrm{I}}\left(x_{\mathrm{I}}, x_{\mathrm{II}}\right) \times \partial_{x_{\mathrm{II}}} \phi_{\mathrm{II}}\left(x_{\mathrm{I}}, x_{\mathrm{II}}\right)$$

Note that $\mathcal{F}(x)$ is a correspondence as the sub-differentials are not necessarily unique due to the definition. $\mathcal{P}(A)$ also denotes the power set of A. We formulate the following problem by concatenating the first-order optimality conditions of all leaders' problems where the goal is finding a vector $x=(x_\mathrm{I},x_\mathrm{II}) \in \mathcal{X}$ (where $\mathcal{X}=X^{\mathrm{I}} \times X^{\mathrm{II}}$) such that:   
$$
\exists w \in \mathcal{F}\left(x\right), \quad \left(y-x\right)^{\top} w +\epsilon \geq 0, \quad \forall y \in \mathcal{G}(x)
$$

Note that $\mathcal{G}(x)=(G_\mathrm{I}(x_\mathrm{I}), G_\mathrm{II}(x_\mathrm{II}))$. The functions $\phi_\mathrm{I}$ and $\phi_\mathrm{II}$ are linear arithmetic circuits and they are piece-wise linear convex functions. In a piece-wise linear convex function, kink points (where two linear or nonlinear segments meet) exhibit several key behaviors, particularly in the context of convexity. Kink points are typically points of non-differentiability, meaning that the gradient does not exist at these points. This occurs because the slope (or gradient) changes abruptly between the segments. At these points, the function does not have a unique gradient but rather a subdifferential. The subdifferential at a kink point includes all possible slopes of the adjacent segments. In this case, despite the non-differentiability, convexity is still preserved at kink points and any local minimum is also a global minimum. This applies even if the minimum occurs at a kink point.

\end{proof}

\begin{remark}
    $\mathcal{F}(x)$ can have a strong separation oracle provided  that $\phi_{\mathrm{I}}$ and $\phi_{\mathrm{II}}$ have strong separation oracle, due to results of \cite{Barton2018}.   
\end{remark}

Finally, to establish inclusion in PPAD, we must consider solutions of the computational variant of the transformed GQVI from Proposition \ref{GQVIAPPRox}. For the computational formulation of GQVI,  a solution consists of two tuples $(x,w)$ and $(x^*,w^*)$ with $\|(x,w)-(x^*,w^*)\|^2_2\leq \beta $ such that  $x^*\in \mathcal{G}(x) $ and  $w^*\in \mathcal{F}(x)$ and $\left(y-x\right) ^T w^* +\beta \geq 0$, for every $y \in \mathcal{G}(x)$.
We take $x^*\in \mathcal{G}(x)$  as the final approximate solution for remedial $L/F$-equilibrium. Note $x^*$ that is $\beta$-statistically close to $x$, the following lemma can imply that the strategy $x^*$, is a $f(\beta)$-approximate $L/F$-equilibrium.  Note that for our proof, we do not need the condition $x^*\in \mathcal{G}(x^*)$ due to the definition of $G_\mathrm{I}$ and $G_\mathrm{II}$.


\begin{lemma}[Generalized Penalty Lemma]\label{blowup}
	 Suppose that $x$ is a generalized equilibrium in game $(S,\mathcal{R},\Theta)$ with $L$-Lipchitz loss functions $\Theta$\footnote{With respect to each norm, we can have different constants.}. Let $x^*$ be a strategy profile that is $\beta$-statistically close to $x$ and $x^*\in \mathcal{R}(x^*)$. Then, $x^*$ is a $f(\beta)$-approximate generalized equilibrium of this game where $f(\beta)=L\beta$\footnote{This lemma is a simple modification of the Penalty Lemma in \cite{MFCS}}.
\end{lemma}

The above-mentioned lemma can be applied to remedial L/F-equilibrium as remedial L/F-equilibrium can be written in the format of a generalized equilibrium problem.


\begin{theorem}\label{remedial}
In a multi-leader-follower game, the problem of finding an approximate remedial $L/F$-equilibrium where the constraints of the followers are given by strong separation oracles is PPAD-complete.
\end{theorem}
\begin{proof}
    Consider the format given in Remark \ref{remarkG}, which allows us to transform the problem to a GQVI, so that is is in PPAD.  PPAD-hardness of this problem follows from the hardness of finding a mixed Nash equilibrium (see \cite{Chen}) in a game with $2$ leaders, where the loss functions represent the expected payoff of mixed strategies where the followers have only one strategy and no restrictions. 
      
\end{proof}

\subsection{Computational Complexity of Resilient Nash Equilibrium}



 We now proceed to the computational complexity of a total version of resilient Nash and related notions.  Assume that we have a game $(\mathcal{G},\mathcal{S},\mathcal{U})$ with $k$ players. When $t$ is not a constant, neither $t$-resilient nor $t$-strong Nash (and also strong Nash when $t=k$ and coalition-proof Nash) are likely to be in PPAD according to their definitions (they are required to satisfy an exponential in $t$ number of conditions). Therefore, we will assume that either the number of players in the game we discuss is constant or the parameter $t$ is constant.  Note that we know that $t$-resilient Nash equilibrium is defined for $t\leq k-1$ \cite{Abraham1}.

\problemStatement{Total $t$-Resilient Nash}{
  Input={We receive as input all the following:
  \begin{itemize}
      \item  A game $\mathcal{G}$ with $k$ players represented by $k$ linear arithmetic circuits that represent $k$ concave utility functions,
      \item The strategy sets $\mathcal{S}$, 
      \item  An accuracy parameter $\epsilon$.
  \end{itemize}},
  Output={One of the following cases:
  \begin{itemize}
      \item (Violation of $t$-multi-concavity): A subset $J\in \mathcal{J}$ of at most size $t$, one index $j$, three vectors\footnote{Recall that we declared that each player $j$ will play from $S_j$ and each player controls $n_j$ variables.  In computational game theory, the computational results relevant to strong Nash notions are usually restricted to mixed strategies and normal-form games. We kept the definitions closer to the generalized Nash equilibrium notion for consistency which is also more general. } $\textbf{x}_{J}, \textbf{y}_{J} \in  S_J$ and $ \textbf{x}_{-J} \in S_{-J}$  such that for some $\lambda\in(0,1)$:
\[
\begin{aligned}
  u_j\left(\lambda  \textbf{x}_J+(1-\lambda)  \textbf{y}_J, \textbf{x}_{-J}\right)
&<\lambda  u_j\left(\textbf{x}_J, \textbf{x}_{-J}\right)+(1-\lambda) u_j\left(\textbf{y}_J, \textbf{x}_{-J}\right)\end{aligned}\]
\item One vector $s^*$ which represents the strategy profile of all players that satisfies:
$$
\left(\forall J\in \mathcal{J} \mbox{ s.t } |J|\leq t\right), \forall s_J^\prime\in S_J, \forall j\in J: u_j\left(\textbf{s}^*\right) +\epsilon \geq u_j\left(\textbf{s}^\prime_J, \textbf{s}^*_{-J}\right)
$$  
  \end{itemize}
  }
}


\begin{theorem}
    The problem of finding a total $t$-resilient Nash equilibrium is PPAD-complete.
\end{theorem}
\begin{proofsketch}
    We need to further generalize the techniques of \cite{concavegames} and also the definitions of variational inequality problems to obtain a version more appropriate for this application.  This stems from the fact that for a $t$-resilient Nash equilibrium with conditions such as multi-concavity, we need to involve exponentially many equations in the constant $t$ (see the definition of $t$-resilient Nash equilibrium). Specifically, we need to make necessary adjustments to both the constrained optimization problem and the robust version of Berge's maximum theorem of \cite{concavegames}. Roughly, we show that a matrix form of variational inequalities may still be solved using Kakutani's fixed point, by an appropriate transformation. Solving this modified constrained optimization problem also requires a slightly generalized sub-gradient ellipsoid central cut method. More information is available in  Appendix \ref{Appendixtotalresilientinppad} and \ref{AppendixCentralcutellipsoidsection}. The transformation from the $t$-resilient Nash equilibrium problem to our generalized version of variational inequalities is also slightly different from the famous transformation introduced by \cite{Harker} (compare Corollary \ref{nashvi1} and Lemma \ref{reslient2k}).  Dealing with the non-differentiability of linear arithmetic circuits is similar. For PPAD-hardness, we consider the problem of finding a mixed approximate Nash equilibrium ($1$-resilient Nash equilibrium) in which the expected payoff is concave. 
\end{proofsketch}

\begin{remark}
    We also do not need to use the generalized Penalty Lemma due to the fact that we use the generalized version of VI for this proof not the generalized version of GQVI.
\end{remark}

\begin{corollary}
    The problem of finding a  $t$-strong Nash equilibrium ($t\leq k-1$ and $t$ is a constant) and coalition-proof Nash equilibrium (with constant players) under multi-concavity are PPAD-complete.
\end{corollary}

\section{Proof of Inclusion in PPAD of GQVI}\label{Variationalsectionproof}

 This section is organized as follows. First, we provide an overview of the contributions from \cite{concavegames} and offer a comparison with our work.  Next, we give computational definitions for several forms of variational inequalities including more general forms QVI and VI without a restriction to linear arithmetic circuits.

\subsection{Some Essential Elements from (\cite{concavegames})}\label{EC23section}

In this section, we present additional definitions and theorems from \cite{concavegames} that we use directly, indirectly, or with modifications. Specifically, for strong separation oracles, in addition to the definition of the strong constrained convex optimization and the strong approximate projection problem, we provide the computational definition of Kakutani's fixed point, finding an approximate generalized equilibrium in concave games and their PPAD-completeness statement. Also, we provide the robust version of Berge's maximum theorem which is essential in the transformation of the generalized equilibrium problem to Kakutani's fixed point (together with the strong constrained convex optimization problem).

\begin{definition} We also need to define Hausdorff semi-continuity notions:
\begin{itemize}
 
    \item We define the set-point distance of a point $x$ from a set $X$ to be $\operatorname{dist}(x, X):=\inf _{z \in X} \mathrm{~d}(x, z)$. 
    \item Let $X$ and $Y$ be two non-empty sets. We  define their Hausdorff distance  $\mathrm{d}_{\mathrm{H}}(X, Y)$ for two sets $X$ and $Y$ to be:
    \[\max \left\{\sup _{x \in X} \operatorname{dist}(x, Y), \sup _{y \in Y} \operatorname{dist}(X, y)\right\} =\inf \{\epsilon \geq 0: X \subseteq \overline{\mathrm{B}}(Y, \epsilon) \wedge y \subseteq \overline{\mathrm{B}}(X, \epsilon)\}\].
    \item A correspondence $\mathcal{F}: X \rightrightarrows Y$ is called (Hausdorff) upper semi-continuous at a point $x_0 \in X$ if and only if for every $\epsilon>0$ there exists a neighborhood $N$ of $x_0$ such that for all $x \in N\left(x_0\right)$, $\mathcal{F}(x) \subseteq \overline{\mathrm{B}}\left(\mathcal{F}\left(x_0\right), \epsilon\right)$.
    \item  A correspondence $\mathcal{F}: X \rightrightarrows Y$ is called (Hausdorff) lower semi-continuous at a point $x_0 \in X$ if and only if for every $\epsilon>0$ there is a neighborhood $N$ of $x_0$ such that $\mathcal{F}\left(x_0\right) \subseteq \overline{\mathrm{B}}(\mathcal{F}(x), \epsilon)$ for all $x \in N\left(x_0\right)$.
    \item  A correspondence $\mathcal{F}: X \rightrightarrows Y$ is called  $L$-Hausdorff Lipschitz continuous if there exists a positive constant $L \in \mathbb{R}^+$ such that, for all $x_1$ and $x_2$ in $X, \mathrm{d}_{\mathrm{H}}\left(\mathcal{F}\left(x_1\right), \mathcal{F}\left(x_2\right)\right) \leq L \cdot\mathrm{d}\left(x_1, x_2\right)$. If the function $\mathcal{F}(x)$ is compact for every $x$ and possesses global Hausdorff Lipschitz properties, then the conditions for both upper and lower semi-continuity are inherently satisfied.
    \item A real-valued correspondence $\mathcal{F}$ on  Euclidean space satisfies  Holder's condition, or is $(q,p)$-Holder continuous (where $q>0$ and $p>1$), when there are real positive constants $\kappa$ and $c$ such that $d_H(\mathcal{F}(y_1),\mathcal{F}(y_2))\leq  \kappa||(y_1)-(y_2)||^q_p+c$.
\end{itemize}
\end{definition}

\subsection{Computational Kakutani's Fixed Point Problem}

We define the computational version of Kakutani's fixed point problem and restate the results of \cite{concavegames}. The definition of the computational problem of finding a Kakutani fixed-point is given as follows:
\problemStatement{Kakutani With A strong Separation Oracle (via $C_{\mathcal{R}(x)}$)}{
  Input={A circuit $C_\mathcal{R}$ that represents strong separation oracle for a $L$-Hausdorff Lipschitz correspondence $\mathcal{R}:\mathbb{R}^{m*} \rightrightarrows\mathbb{R}^{m*}$ and an accuracy parameter $\alpha$. 
  },
  Output={One of the following cases:
  \begin{itemize}
      \item (Violation of non-emptiness): A vector $x \in$ $\mathbb{R}^{m*}$ such that $\overline{\mathrm{B}}(\mathcal{R}(x),-\epsilon)=\emptyset$.
\item  (Violation of L-Hausdorff Lipschitzness)
Four vectors $p, q, z, w \in\mathbb{R}^{m*}$ and a constant $\epsilon>0$ such that $w=\widetilde{\Pi}_{\mathcal{R}(q)}^{\epsilon}(q)$ and $z=\widetilde{\Pi}_{\mathcal{R}(p)}^{\epsilon}(w)$ but $\|z-w\|>L\|p-q\|+3\epsilon$.
\item  Vectors $x, z \in\mathbb{R}^{m*}$ such that $\|x-z\| \leq$ $\alpha$ and $z \in \mathcal{R}(x) \Leftrightarrow d(x, \mathcal{R}(x)) \leq \alpha$.
  \end{itemize}
  }
}

\begin{theorem}[\cite{concavegames}]
    The  Kakutani fixed-point problem  for a strong separation oracle is PPAD-complete.
\end{theorem}

\begin{remark}
    The computational  Kakutani problem with separation oracles seeks the inherently weaker assurance that each point in the space has a singular projection onto $\mathcal{R}$. The convexity of $\mathcal{R}$ (i.e., $\mathcal{R}$ is a convex valued correspondence) is adequate for ensuring the uniqueness of the nearest point (for more information, see \cite{concavegames}).
\end{remark}

\subsection{Using Arithmetic Circuits}

   In general, when we are restricted to queries to a zeroth and first-order oracle or when inspecting a circuit or Turing machine description, evaluating consistency between function values and their gradients poses computational challenges. For a computational version of Kakutani's fixed-point, an essential requirement that we need to impose on these oracles is that the output which is a separating hyperplane, should have polynomial bit complexity with respect to the given input parameters. The proofs in \cite{concavegames} utilize instances where this consistency can be assured syntactically by using linear arithmetic circuits which are proven to be useful and general enough as we do not lose representation power (see \cite{Fearnley,concavegames} and Appendix \ref{arithmeticcircuits}).

\subsection{Robust Berge's Maximum Theorem}

The paper  \cite{concavegames}  leverages the following quantified version of Berge’s theorem for general convex functions.

\begin{theorem}[\cite{concavegames}]\label{berge}
  (Robust Berge Maximum Theorem). Let $A \subseteq \mathbb{R}^n$ and $B \subseteq \mathbb{R}^{m}$. Consider a continuous function $f: A \times B \rightarrow \mathbb{R}$ that is $\mu$-strongly concave $\forall a \in A, L$-Lipschitz in $A \times B$ and also a $L^{\prime}$-Hausdorff Lipschitz, non-empty, convex-set, compact-valued correspondence $g: A \rightrightarrows B$. Define $f^*(a)=\max _{b \in g(a)} f(a, b)$ and $g^*(a)=\arg \max _{b \in g(a)} f(a, b)$. Then, we observe $f^*$ is continuous and $g^*$ is upper semi-continuous and single-valued, i.e., continuous. Furthermore, $f^*$ and $g^*$ are Lipschitz and $\left(L^{\prime}+2 \sqrt{\frac{4}{\mu}} \sqrt{\left(L+L \cdot L^{\prime}\right)}\right)-(1 / 2,p)$-Holder continuous respectively (for sufficiently small differences).
\end{theorem}


\subsection{Strongly Concave Games}

The paper \cite{concavegames} establishes PPAD-completeness of an instance of concave games introduced in the celebrated work of Rosen \cite{rosen} (also called \emph{concave games} \cite{100challenge,concavegames,rosen}).

\problemStatement{Strongly Concave Games Problem With SO}{
  Input={We receive as input all the following:
  \begin{itemize}
      \item $k$ circuits representing the utility functions $\left(u\right)_{i=1}^k$ for all $k$ players,
      \item A Lipschitzness parameter $L$, a strong concavity parameter $\mu$, and accuracy parameter $\epsilon$, 
      \item $\mathcal{S}=\Pi_{i=1}^k S_i$ that is a convex set called \emph{the strategy domain} where $S_i$ represent the strategy domain for each player $i$,
      \item An arithmetic circuit representing a strong separation oracle for the common constraint $\mathcal{R} $ that is a non-empty, convex, and compact subset of $\mathcal{S}$,
      \item Accuracy parameter $\epsilon$.
  \end{itemize}
  },
  Question={We output one of the following:
  \begin{itemize}
   \item (Violation of non-emptiness):
A certificate that  $\overline{\mathrm{B}}(\mathcal{R},-\epsilon)=\emptyset$.
 \item (Violation of Lipschitz continuity):
A certification that there exist at least two vectors $x, y \in \mathcal{S} $ and an index $i \in[n]$ such that $\left|u_i(x)-u_i(y)\right|>L \|x-y\|$.
\item  (Violation of strong concavity):
An index $i \in[n]$, three vectors $x_i, y_i \in S_i, \textbf{x}_{-i} \in S_{-i}=\Pi_{j=1,j\neq i}^k S_j$ and a number $\mu \in[0,1]$ such that:
\[
\begin{aligned}
  u_i\left(\lambda  x_i+(1-\lambda)  y_i, \textbf{x}_{-i}\right)
&<\lambda  u_i\left(x_i, \textbf{x}_{-i}\right)+(1-\lambda) u_i\left(y_i, \textbf{x}_{-i}\right)\\ &+\frac{\lambda(1-\lambda)}{2} \mu \cdot\left\|\left(x_i, \textbf{x}_{-i}\right)-\left(y_i, \textbf{x}_{-i}\right)\right\|_2^2  
\end{aligned}\]
\item An $\epsilon$-approximate generalized equilibrium  in the sense of {\rm (\ref{eq:bestcommonresp})}.
  \end{itemize}
  }
}

\begin{theorem}[\cite{concavegames}]\label{concaveppad}
    The Strongly Concave Games for a strong separation oracle is PPAD-complete.
\end{theorem}
\begin{proofsketch}
    A complete proof is available in \cite{concavegames}. The difference (compared to Rosen's proof \cite{rosen}) lies in how $\rho$ is defined:

\[\rho(x,y)=\sum_{i\in [n]}u_i(y_i,x_{-i})-\gamma \cdot\|y\|_2^2\]
    Here the constant value $\gamma$ ensures that $\rho$ stays a $2\gamma$-strongly concave function in $y$. Using $\rho(x,y)$, the proof constructs the final correspondence $F$ and proves non-emptiness. Specifically, the proof constructs a separation oracle for $F$ by utilizing the robust version of Berge's maximum theorem. This is done in polynomial time by using the constrained convex optimization problem (given a separation oracle that represents the convex set $S$). The reduction concludes by proving a solution provided by Kakutani's fixed point given $F$ as the input can form an approximate solution to the generalized Nash equilibrium problem in concave games using some techniques that are similar to Rosen's proof.
\end{proofsketch}
\begin{remark}
  We will use this approach later in the proof of inclusion in PPAD of variational inequality problems.  However, we utilize the proof of the existence of variational  inequalities provided by \cite{chan} instead of the techniques provided by Rosen.
\end{remark}

\subsection{A Direct Comparison to \cite{concavegames}}

We now compare our approach to \cite{concavegames}. We will use the definition of the strong approximate projection problem from Section~\ref{AppendixvariationalSection} with no modification in the computational definition of the problems we provide. For this paper, we require a different version of the strong constrained convex optimization problem with different inputs and outputs. As mentioned, this problem can be solved using the sub-gradient ellipsoid central cut algorithm. However, for the modified versions we define, for every single problem namely, multi-leader-followers and resilient Nash equilibrium, we need to modify the sub-gradient ellipsoid central cut algorithm to make it appropriate for different inputs and outputs (especially for emptiness exceptions as discussed). For multi-leader-follower games, the modifications are easier while for resilient Nash equilibrium, we need to do more modest modifications\footnote{Appendix \ref{AppendixCentralcutellipsoidsection} provides all the details}. We directly apply Kakutani's fixed point in proofs of inclusion in PPAD, however, the input to this problem differs significantly compared to \cite{concavegames} in the reductions we provide. Specifically, we also require more generalized definitions of these problems and also more careful error handling. Also, for the resilient Nash problem we need to handle more complex inputs that are a combination of vectors and matrices\footnote{For more information, see the definition of MGQVI.}.   

\newpage

\subsection{General Computational Definition of Variational Inequality Problems}

 We provide another version of definitions for variational inequalities that are not restricted to linear arithmetic circuits.  The computational problem of finding an approximate solution to GQVI for separation oracles is defined as follows:

\problemStatement{$GQVI(\mathcal{F},\mathcal{R})$ With strong Separation Oracles}{
  Input={ We receive as input all the following:
  \begin{itemize}
      \item  A circuit $C_{\mathcal{R}}$ which represents a strong separation oracle for a $L_{\mathcal{R}}$-Hausdorff Lipschitz correspondence $\mathcal{R}: \mathbb{R}^{m*} \rightrightarrows  \mathbb{R}^{m*}$,
      \item A circuit $C_\mathcal{F}$ which represents a strong separation oracle for a $L_\mathcal{F}$-Hausdorff Lipschitz and $\gamma$-strongly convex valued correspondence $\mathcal{F}: \mathbb{R}^{m*} \rightrightarrows  \mathbb{R}^{m*}$,
      \item    An accuracy parameter $\beta$. 
  \end{itemize}},
  Output={One of the following cases:
  \begin{itemize}
      \item (Violation of non-emptiness): A vector $x \in \mathbb{R}^{m*}$ such that $\overline{\mathrm{B}}(\mathcal{R}(x),-\epsilon)=\emptyset$ or $x \in \mathbb{R}^{m*}$ such that $\overline{\mathrm{B}}(\mathcal{F}(x),-\epsilon)=\emptyset$,
\item  (Violation of $L_{\mathcal{R}}$-Hausdorff Lipschitzness of $\mathcal{R}$):
Four vectors $p, q, z, w \in \mathbb{R}^{m*}$ and a constant $\epsilon>0$ such that $w=\widetilde{\Pi}_{\mathcal{R}(q)}^{\epsilon}(q)$ and $z=\widetilde{\Pi}_{\mathcal{R}(p)}^{\epsilon}(w)$ but $\|z-w\|>L_{\mathcal{R}}\|p-q\|+3 \epsilon$.
\item  (Violation of $L_\mathcal{F}$-Hausdorff Lipschitzness of $\mathcal{F}$)
Four vectors $p, q, z, w \in \mathbb{R}^{m*}$ and a constant $\epsilon>0$ such that $w=\widetilde{\Pi}_{\mathcal{F}(q)}^{\epsilon}(q)$ and $z=\widetilde{\Pi}_{\mathcal{F}(p)}^{\epsilon}(w)$ but $\|z-w\|>L_\mathcal{F}\|p-q\|+3 \epsilon$,
\item 
(Violation of strong convexity):
There vectors $x, p, q\in \mathbb{R}^{m*}$ and two constants $\epsilon>0$ and $\lambda\in (0,1)$ such that:
$$
\begin{aligned}
\widetilde{\Pi}_{\mathcal{F}(x)}^{\epsilon}(\lambda p+ (1-\lambda)q ) & > \lambda \cdot \widetilde{\Pi}_{\mathcal{F}(x)}^{\epsilon}(p)+(1-\lambda) \cdot \widetilde{\Pi}_{\mathcal{F}(x)}^{\epsilon}(q)
\\&-\frac{\lambda(1-\lambda)}{2} \cdot \gamma\cdot\left\|p-q\right\|_2^2+\epsilon
\end{aligned}
$$
\item Two tuples $(x,w)$ and $(x^*,w^*)$ with $\|(x,w)-(x^*,w^*)\|\leq \beta $ such that  $x^*\in \mathcal{R}(x) $ and  $w^* \in \mathcal{F}(x)$ and $\left(y-x\right) ^T w^* +\beta \geq 0, \quad \forall y \in \mathcal{R}(x)$   
  \end{itemize}
  }
}

\begin{remark}
    With regards to the constants, see Theorem 3.15 and Lemma 3.16  in \cite{concavegames}.
\end{remark}

\subsection{Special Cases: Quasi Variational Inequalities (QVI) and Variational Inequalities (VI)}\label{Espcases}

Below, we define a computational variant of QVI that receives one strong separation oracle and one circuit that represents a function (in contrast to a correspondence in GQVIs) as an input of the problem. Defining the variants that receive linear arithmetic circuits for these two sub-cases is similar.

\problemStatement{$QVI(F,\mathcal{R})$ With a strong Separation Oracle}{
  Input={We receive as input all the following:
  \begin{itemize}
      \item  A circuit $C_{\mathcal{R}}$ which represents a strong separation oracle for a $L_{\mathcal{R}}$-Hausdorff Lipschitz correspondence $\mathcal{R}: \mathbb{R}^{m*} \rightrightarrows  \mathbb{R}^{m*}$,
      \item A circuit $C_F$ which represents a $L_F$-Hausdorff Lipschitz function $F: \mathbb{R}^{m*} \rightarrow  \mathbb{R}^{m*}$, 
      \item   An accuracy parameter $\beta$. 
  \end{itemize}},
  Output={One of the following cases:
  \begin{itemize}
      \item (Violation of non-emptiness): A vector $x \in \mathbb{R}^{m*}$ such that $\overline{\mathrm{B}}(\mathcal{R}(x),-\epsilon)=\emptyset$, 
\item  (Violation of $L_{\mathcal{R}}$-Hausdorff Lipschitzness of $\mathcal{R}$):
Four vectors $p, q, z, w \in \mathbb{R}^{m*}$ and a constant $\epsilon>0$ such that $w=\widetilde{\Pi}_{\mathcal{R}(q)}^{\epsilon}(q)$ and $z=\widetilde{\Pi}_{\mathcal{R}(p)}^{\epsilon}(w)$ but $\|z-w\|>L_{\mathcal{R}}\|p-q\|+3 \epsilon$,
\item  (Violation of $L_F$-Hausdorff Lipschitzness of $F$):
Two vectors $p, q \in \mathbb{R}^{m*}$ such that $\|F(p)-F(q)\|>L_F\|p-q\|$,
\item  Two vectors $x^*$ and $x$ with $\|x - x^*\|\leq \beta $ such that  $x^*\in \mathcal{R}(x) $ and $\left(y-x\right) ^T F(x^*) +\beta \geq 0, \quad \forall y \in \mathcal{R}(x)$ 
  \end{itemize}
  }
}

Next, we introduce a  computational version of VI. In this problem, the input consists of a strong separation oracle for a set (as opposed to a correspondence in QVIs) and a circuit representing a function (similar to QVI). 

\problemStatement{$VI(F)$ With a strong Separation Oracle}{
  Input={We receive as input all the following:
  \begin{itemize}
      \item  A circuit $C_{\mathcal{R}}$ which represents a strong separation oracle for a non-empty closed convex set $\mathcal{R}$,
      \item A circuit $C_F$ which represents a $L_F$-Hausdorff Lipschitz function $F: \mathbb{R}^{m*} \rightarrow  \mathbb{R}^{m*}$, 
      \item  An accuracy parameter $\beta$.
  \end{itemize}},
  Output={One of the following cases:
  \begin{itemize}
      \item (Violation of non-emptiness): A symbol $\bot$ with a polynomial-sized witness that certifies  $\overline{\mathrm{B}}(\mathcal{R},-\epsilon)=\emptyset$,
\item  (Violation of $L_F$-Hausdorff Lipschitzness of $F$):
Two vectors $p, q \in \mathbb{R}^{m*}$ such that $\|F(p)-F(q)\|>L_F\|p-q\|$,
\item One vector $x^*$ with the condition $x^*\in \mathcal{R} $ such that  $\left(y-x^*\right) ^T F(x^*) +\beta  \geq 0, \quad \forall y \in \mathcal{R}$  
  \end{itemize}
  }
}

    \begin{theorem}\label{gqviinppadnolinear}
The generalized quasi-variational inequality problem $GQVI(\mathcal{F},\mathcal{R})$ for  $\mathrm{SO}_{\mathcal{R}}$ and $\mathrm{SO}_\mathcal{F}$ is  PPAD-complete.
\end{theorem}                      
\begin{proof}
For this version, by assumption, we must have access to the sub-gradients of $\Phi$ compared to linear arithmetic circuits.  The proof of Theorem \ref{gqviinppad} is similar.

  Without loss of generality, we can consider $[-1,1]^m$ instead of $\mathbb{R}^{m*}$ (similar to \cite{concavegames}).   The proof is organized as follows. First, we define the following:
  \[\Phi(y,x,w)=-(y-x)^T w-\gamma(||y||^2_2)\]
  \[\Pi(x,w)=\{y\in \mathcal{R}(x) ~|~\Phi(y,x,w)>max_{y\in \mathcal{R}(x)}\Phi(y,x,w)-\epsilon \}\]
  We show that for a constant $\kappa^\prime$, $\Psi(x,w)=(\Pi(x,w),\mathcal{F}(x))$ is $\kappa^\prime$-Lipschitz continuous in $(x,w)$.  Then, we construct a strong separation oracle for $\Psi$.  Finally, we show that an approximate Kakutani's fixed point of this correspondence will provide an approximate solution to the given GQVI.


To show that $\Psi$ is $\kappa^\prime$-Lipschitz continuous, we show  that $d_H(\Psi(x_1,w_1),\Psi(x_2,w_2))\leq \kappa^\prime||(x_1,w_1)-(x_2,w_2)||^q_p+c$ for some constants $q$,$p$ and $c$. Similar to Theorem E1 in \cite{concavegames}, here $\Phi(y,x,w)$ is $(2\gamma)$-strongly concave in $y$. This follows by the fact that $-(y-x)^T w$ is $0$-strongly concave function of $y$  and also $-\gamma(||y||^2_2)$ is $2\gamma$-strongly concave function of $y$  by Lemma \ref{norm2}.  This technique is similar to Theorem \ref{concaveppad} as we need to ensure that for some $\mu$, $\mu$-strong concavity holds to leverage  Berge's maximum theorem.   In addition, $\Phi$ is $G$-Lipschitz continuous where $G=m+2\gamma m$ for all $x$ and $w$.

The procedure to show that $\Pi(x,w)$ is $\kappa$-Lipschitz continuous in $x$ and $w$ is also similar to Theorem E1 of \cite{concavegames}. Let $H(x,w)=\{y\in \mathcal{R}(x) ~|~\Phi(y,x,w)=max_{y\in \mathcal{R}(x)} \Phi(y,x,w)\}$. By optimality KKT conditions for maximization of a concave function with respect to the constraint set $\mathcal{R}(x)$ for all $y\in \mathcal{R}(x) $ we have:
\[
\partial \Phi(y,x,w)^T(y^*-y)\geq 0 \text{ where } y^*=argmax_{y\in \mathcal{R}(x)}\Phi(y,x,w)
\]Furthermore, $(2 \gamma)$-strong-concavity of $\Phi(\cdot,x,w)$ results the following inequality:
\[
\Phi\left(y^*,x,w\right)-\Phi(y,x,w) \geq \partial \Phi\left(y^*,x,w\right)^{\top}\left(y^*-y\right)+\gamma\left\|y-y^*\right\|_2^2 
\]Finally, by combining the previous inequalities we have:
\[
\Phi\left(y^*,x,w\right)-\Phi(y,x,w) \geq \gamma\left\|y-y^*\right\|_2^2 
\]
In conclusion, for $y^*\in H(x,w)$ and $y\in \Pi(x,w)$, we have the following inequality:
\begin{equation}\label{kapp}
    \gamma\left\|y-y^*\right\|^2_2 \leq \epsilon \text{ or equivalently } \left\|y-y^*\right\|_2 \leq \sqrt{\frac{\epsilon}{\gamma}}
\end{equation}



Recall that we showed  $\Phi(\cdot,x,w)$ is strongly concave. Now, we can apply Theorem \ref{berge}  (Berge's Maximum Theorem) to $f((x,w),y)=\Phi(y,x,w)$ and $g((w,x))=\mathcal{R}(x)$,  where $g$ is a non-empty, compact, convex-valued correspondence.  By Theorem \ref{berge}, there exists a constant $\kappa$ such that $H$ is $\kappa$-Lipschitz (and also Holder) continuous. In conclusion, we have the following:
\begin{equation}\label{kapp2}
    d_H(H(x_1,w_1),H(x_2,w_2))\leq \kappa||(x_1,w_1)-(x_2,w_2)||^{\frac{1}{2}}_2+c
\end{equation}
Combining Equations \ref{kapp} and \ref{kapp2}  results in:
\begin{equation} \label{kapp3}
\begin{split}
 d_H(\Pi(x_1,w_1),\Pi(x_2,w_2)) &\leq d_H(H(x_1,w_1),H(x_2,w_2))\\
 & + d_H(\Pi(x_1,w_1),H(x_1,w_1))\\ &+d_H(\Pi(x_2,w_2),H(x_2,w_2))
\end{split}
\end{equation}
In conclusion, we have the following:

\begin{equation}\label{kapp4}
    d_H(\Pi(x_1,w_1),\Pi(x_2,w_2))\leq \kappa||(x_1,w_1)-(x_2,w_2)||^{\frac{1}{2}}_2+2\sqrt{\frac{\epsilon}{\gamma}}+c
\end{equation}
Since $\mathcal{F}$ is $L_\mathcal{F}$-Hausdorff Lipschitz, by the definition of $\Psi$ we can deduce that there exists a constant $\kappa^\prime$: 

\begin{equation}\label{kapp5}
    d_H(\Psi(x_1,w_1),\Psi(x_2,w_2))\leq \kappa^\prime||(x_1,w_1)-(x_2,w_2)||^{\frac{1}{2}}_2+c^\prime
\end{equation}
This shows that $\Psi$ is  $\kappa^\prime$-Hausdorff Lipschitz (and also $(\frac{1}{2},2)$-Holder) continuous.
To establish a reduction employing the computational Kakutani's problem variant, it is essential to establish the existence of a bounded-radius ball within the correspondence $\Psi(x,w)$.  We also know  that $\overline{\mathrm{B}} (\mathcal{R}(x),-\epsilon)$ and also $\overline{\mathrm{B}} (\mathcal{F}(x),-\epsilon)$ are  non-empty\footnote{We know that $\overline{\mathrm{B}} (\mathcal{R}(x),-\epsilon)$ is non-empty.}. Let $y_x^*=argmax_{y\in \mathcal{R}(x)}\Phi(y,x,w)$. Then, we define a none-empty region $V_x= \overline{\mathrm{B}}(y_x^{\star},+\frac{\epsilon}{G})$  where $G$ is the Lipschitz constant of $\Phi$. By Lipschitzness of $\Phi\left(\cdot,x,w\right)$, we have:
\[ \forall y \in V_x, \quad \left|\Phi\left(y, x,w\right)-\Phi\left(y^*_x,x,w\right)\right| \leq G\frac{\epsilon}{G}=\epsilon\]
 In conclusion, $\Phi(V_x,x,w)\subseteq [\Phi\left(y^*_x,x,w\right)-\epsilon,\Phi\left(y^*_x,x,w\right)]$ is non-empty. Then, we can conclude that $\Phi(V,x,w) \subseteq\Pi(x,w)$ (see the definition of $\Pi$). We also know that for any $x\in \mathcal{R}(x)$, $\overline{\mathrm{B}} (\mathcal{F}(x),-\epsilon)$ is non-empty,  then, we can find a non-empty region $C_x$ such that  $C_x \subseteq\overline{\mathrm{B}} (\mathcal{F}(x),-\epsilon)$.  This is sufficient to show that for all $x$ and $w$, $ \Psi(x,w)$ is a non-empty correspondence. 


Now we proceed to constructing a strong separation oracle for $\Psi(x,w)$ leveraging a slightly modified version of the Strong Constrained Convex Optimization framework stated in the following:

\problemStatement{Modified Strong Constrained Convex Optimization Problem}{
  Input={A zeroth and first order oracle for the concave function $G: \mathbb{R}^{m*}\rightarrow \mathbb{R}$, two rational numbers $\delta,\epsilon>0$ and a strong separation oracle $\mathrm{SO}_{\mathcal{R}}$ for a non-empty closed convex-valued correspondence $\mathcal{R}:\mathbb{R}^{m*}\rightrightarrows  \mathbb{R}^{m*}$ and one input $x$.
  },
  Output={One of the following cases:
  \begin{itemize}
      \item (Violation of non-emptiness)
A failure symbol $\perp$ with a polynomial-sized witness that certifies that $\overline{\mathrm{B}}(\mathcal{R}(x),-\delta)=\emptyset$,
\item (Approximate Maximization)
A vector $z \in \mathbb{Q}^m \cap \mathcal{R}(x)$, such that $G(z) +\epsilon  \geq \max _{y \in \mathcal{R}(x)} G(y)$ for all values.
  \end{itemize}
  }
}

The problem mentioned above can be solved in polynomial time using the sub-gradient ellipsoid central cut method given the separation oracle as it satisfies all of the requirements such as the convexity of the domain (for more information, see Appendix \ref{AppendixCentralcutellipsoidsection} and also Appendix G in \cite{concavegames})\footnote{Note that the modification in this problem is that it is a maximization problem and we have an input where $\mathcal{R}(x)$ becomes a set as $\mathcal{R}$ is convex-valued.}. We can compute a solution $y^*\in \mathcal{R}(x)$ using the sub-gradient ellipsoid central cut method such that $\Phi\left(y^*,x,w\right)\geq max_{y\in \mathcal{R}(x)} \Phi\left(y,x,w\right)-\epsilon $.  Thus, we can substitute a SO for $\Psi$ by considering a strong separation oracle for the following set:
\[
\overline{\Psi}_{s}(x,w)=\{(y,w)\in (\mathcal{R}(x),\mathcal{F}(x))~|~- \Phi(y,x,w)\leq \gamma^\prime\}
\]where $\gamma^\prime= -\Phi(y^*,x,w)$. In other words, we are looking for $(y,w)\in (\mathcal{R}(x),\mathcal{F}(x))$ such that $(y-x)^T w +\gamma||y||^2_2  \geq (y^*-x)^T w + \gamma||y^*||^2_2 $ given the strong separation oracles representing $\mathcal{F}$ and $\mathcal{R}$. This separation oracle exists due to Theorem H.3 in \cite{concavegames}.

Now, we can give the constructed separation oracle for $\Psi$ as input to the computational Kakutani problem with accuracy parameter $\alpha=\frac{\epsilon^\prime}{\kappa^\prime}$ where $\epsilon^\prime=\epsilon h$. The output of this Kakutani instance will be two points $(x,w)\in ([-1,1]^m,[-1,1]^m)$ and $z=(x^*,w^*)\in\Psi(x,w)$  where $||(x,w)-(x^*,w^*) ||\leq \frac{\epsilon^\prime}{\kappa^\prime}$ and $d((x,w),\Psi(x,w))\leq \frac{\epsilon^\prime}{\kappa^\prime}$. Thus, $w^*\in \mathcal{F}(x)$ and $d(w,\mathcal{F}(x))\leq \frac{\epsilon^\prime}{\kappa^\prime} $. In addition, $x^*\in \Pi(x,w)$ and $d(x,\Pi(x,w))\leq \frac{\epsilon^\prime}{\kappa^\prime} $. By the definition of $\Pi$,  for every $y\in \mathcal{R}(x)$,  $\Phi\left(x^*,x,w\right)\geq \Phi\left(y,x,w\right)$. In conclusion:
\[ \left(y-x\right) ^T w+\gamma||y||^2_2  \geq \left(x^*-x\right) ^T w+ \gamma||x^*||^2_2 ,  \thickspace \forall y \in \mathcal{R}(x) \] So, recalling that $x,y, w\in [-1,1]^m $, for a is small enough  $\gamma$ , there exists a constant $u(\gamma)$:
\[ \left(y-x\right) ^T w \geq  \pm \frac{\epsilon^\prime}{\kappa^\prime} \pm u(\gamma) ,  \quad \forall y \in \mathcal{R}(x)\]

Finally, knowing that $||(x,w)-(x^*,w^*) ||^2_2\leq \frac{\epsilon^\prime}{\kappa^\prime}$ and $w^*\in \mathcal{F}(x)$, and using the Lipschitz continuity of $\mathcal{F}$, we can easily conclude that there exists a constant $\zeta$ such that:
\[ \left(y-x\right) ^T w^* \geq  -\zeta ,  \quad \forall y \in \mathcal{R}(x) \]
 Considering appropriate numbers for $h$ (recall that $\epsilon^\prime=\epsilon h$) and $\gamma$  results the following (setting $\beta=\epsilon$): 
\[ \left(y-x\right) ^T w^* + \beta \geq 0  ,  \quad \forall y \in \mathcal{R}(x) \]
\end{proof}

\begin{remark}
   Note that the convexity condition on $\mathcal{F}$ is only important for the construction of the separation oracle for $\Psi$. Indeed, $\Psi$ must be a convex-valued otherwise membership of an element in $\Psi$ cannot be assessed by using the strong separation oracle framework. In addition, we do not need "strong" convexity of $\mathcal{F}$ in general and we only provided it for completeness. One possibility is that the case violation of convexity happens because of the approximation error in the approximate projection problem. This means that there exits one $x$ point such that for some points $p$ and $q$ and some constant $\lambda\in (0,1)$:
    $$
    \begin{aligned}
\widetilde{\Pi}_{\mathcal{F}(x)}^{\epsilon}(\lambda p+ (1-\lambda)q ) & > \lambda \cdot \widetilde{\Pi}_{\mathcal{F}(x)}^{\epsilon}(p)+(1-\lambda) \cdot \widetilde{\Pi}_{\mathcal{F}(x)}^{\epsilon}(q)+\epsilon
\end{aligned}
$$
However, we might have the following relationship:
$$
\begin{aligned}
\Pi_{\mathcal{F}(x)}(\lambda p+ (1-\lambda)q ) \leq \lambda \cdot {\Pi}_{\mathcal{F}(x)}(p)+(1-\lambda) \cdot \Pi_{\mathcal{F}(x)}(q)
\end{aligned}
$$
In this situation, restricting the condition of convexity by adding $+\epsilon$ may not be the best solution. We may simply consider $\mathcal{F}^\prime(x)=\mathcal{F}(x)+\gamma^\prime\|x\|^2_2$ for some $\gamma^\prime$ and $\mathcal{F}^\prime(x)$ will become convex while we consider the approximate projection problem and we can repeat the proof with minor adjustments.

\end{remark}

\begin{proposition}\label{ppadhardqvi}
 The quasi-variational inequality (QVI) problem for a strong separation oracle $\mathrm{SO}_{\mathcal{R}}$ is PPAD-complete.
\end{proposition}
\begin{proof}
          Inclusion in PPAD is implied by Theorem \ref{gqviinppadnolinear}. PPAD-hardness of this problem is implied by the hardness VI\footnote{We may use alternative approaches such as using the concave games or other similar problems.}.

\end{proof}

\begin{theorem}
The variational inequality problem (VI) for a strong separation oracle $\mathrm{SO}_{\mathcal{R}}$  is PPAD-complete.
\end{theorem}
\begin{proof}
     Inclusion in PPAD is implied by Theorem \ref{gqviinppadnolinear}. PPAD-hardness of this problem is implied by the hardness of approximate Nash equilibrium. Finding a mixed (a well-known concept in game theory) Nash equilibrium in bi-matrix games is a special case that is shown to be PPAD-hard \cite{Chen} where maximizing (or equivalently minimizing the loss) of the expected payoff is the optimization goal. The rest of the proof is the same as Corollary \ref{nashvi2}.
\end{proof}

\section{Conclusion and Future Work}
We have initiated a study of the computational aspects of well-known instances of variational inequalities and considered their application to two different game-theoretic concepts with significant variations in constraints and formats. Variational inequalities have become an important tool for modeling solution concepts in game theory, and our results contribute to this approach by providing it with an algorithmic foundation. Our result, which establishes inclusion in PPAD, highlights the significance of PPAD with respect to a broad range of central problems in game theory under reasonable assumptions.  Below we suggest some possible research directions.

 
 We have not considered exact versions of variational inequality problems, but conjecture that it should be possible to prove FIXP-completeness for appropriate formulations (see \cite{Ratsikas}). Also, deriving computational complexity results for problems we studied under different assumptions could be of potential interest.   Deriving related computational results such as communication lower bounds could provide additional understanding of their application. Further investigation of the computational complexity of problems related to variational inequalities and other similar optimization problems could be of interest (for example, see \cite{Approximate}).  We also can easily combine two notions that we studied in order to derive a PPAD-completeness proof for a version of multi-leader follower games in which an approximate remedial equilibrium necessitates resistance to coalitions of leaders trying to be better off. This notion could be of interest in oligopoly markets and multi-party computation \cite{Abraham1,coalitionproofapplications,Pang2005}.  Moreover, there has been important work on games with uncertain data. Specifically, in addition to probabilistic models, distribution-free models based on a worst-case scenario have received attention in recent years and have a variety of applications \cite{Robustdef,Stocastic3bookapp,Survey,stocastic1}. Specifically, we believe GQVI could be used for the version of multi-leader-follower games with the form of uncertainty introduced in \cite{Robustuncertold}.  In general, we may also study computational aspects of multi-leader-follower games in which several remedial models are combined and consider their relationship to different variants of variational inequalities (see also \cite{Robustuncertold,mlfapp3,Robustuncert}). 

\paragraph*{Acknowledgement}

We would like to acknowledge the support of NSERC discovery grant RGPIN-2021-02481. The results of this paper are an extension of the poster that appeared on AAMAS accompanied by a 2-page abstract \cite{AAMAS24} which did not include resilient Nash equilibrium and related problems and all required extensions.

 \bibliography{refs}

\appendix

\newpage

\section*{Organization of the Appendices}

Appendix \ref{AppendixDefinitionssection} is dedicated to providing standard definitions that were omitted. Appendix \ref{CLassical} inspects various classical existence results relevant to the problems we study here.  Appendix \ref{AppendixElementaryProofsection} contains elementary proofs pertaining to the properties of (L2) distance, variational inequalities, and equilibrium problems that we use throughout the paper (e.g, Proof of Proposition \ref{geneqvari2}). The proof for membership in PPAD for the total resilient Nash equilibrium is available in Appendix \ref{Appendixtotalresilientinppad}.  We also relegated the discussion of weak separation oracles which require more cumbersome definitions to Appendix \ref{AppendixWeakoracles}.  Appendix \ref{AppendixCentralcutellipsoidsection} presents a more generalized version of the sub-gradient-cut method we require for the membership in PPAD of the resilient Nash equilibrium and related problems.  The decision version of resilient Nash equilibrium is also investigated thoroughly in Appendix \ref{AppendixDescionResilient}.

\section{Omitted Definitions}\label{AppendixDefinitionssection}

\subsection{Normal Form Games}

\begin{definition}
	A \emph{normal  form  game} is specified by a finite set of $k$ \emph{players}, sets $S_1,\dots, S_k$ as \emph{pure strategies} (or \emph{actions}), and  $(U_1,\dots,U_k)$ of \emph{utility functions} where $U_i:S^k\rightarrow\mathbb{R}$.
 	
Any such game has a representation where $U_i$ is given by a $n^k$-dimensional \emph{payoff matrix} $A_i$.
For the case of two players, this gives rise to the standard notion of a \emph{bi-matrix game.}

\end{definition}	

In many situations, it is useful to allow players to randomize over pure strategies. This leads to the following:

\begin{definition}

	 Let $\triangle_{n}$ be the the probability simplex in $n$-dimension, 
	    	\[
	\triangle_{n}=\left\{x\in \mathbb{R}^{n} \mid x_{i} \geq 0, \forall i, \sum_{i=1}^{n} x_{i}=1\right\}.
    \]
    A \emph{mixed strategy} is an element of $\triangle_{n}$, corresponding to a distribution on strategies. If  $x$ is a mixed strategy, $\mathrm{Supp}(x)$ (\emph{support} of $x$) denotes the set elements of $S$ to which $x$ assigns positive probability.
    \end{definition}
    Note that we may view a pure strategy $s$  as a mixed strategy assigning probability $1$ to $s$.
    
	\begin{definition}
	    A \emph{pure strategy profile} is a  $k$-dimensional vector $\textbf{s}=(s_1,\dots,s_k)$ of strategies, where $s_i$ is the pure strategy played by player $i$. We define a \emph{mixed strategy profile} $\textbf{x}=(x_1,\dots,x_k)$ similarly. 
	\end{definition}
If $\textbf{y}=(y_1,\dots,y_k)$ is a strategy profile and $x$ is a strategy, we use the standard notation $(x,\textbf{y}_{-i})$ to denote $(y_1,\dots,y_{i-1},x,y_{i+1},\dots,y_k)$.

\begin{definition}
	\label{expected}
	Given a pure strategy profile $\textbf{s}$, the \emph{payoff} for player $i$ is $U_i(\textbf{s})$.  We can extend this notion to mixed strategies by defining $U_i(\textbf{x})$ for  a mixed strategy profile $\textbf{x}$ to be $\mathbb{E}_{\textbf{s}\leftarrow \textbf{x}}[U_i(\textbf{s}))]$, i.e., the expected value of $U_i(\textbf{s})$ when $\textbf{s}$ is distributed according to $\textbf{x}$.

\end{definition}




\subsection{Pareto Efficiency (Maximization Version)}

\begin{definition}
    A strategy profile $s^*$ is \emph{strictly Pareto optimal} if there is no alternative strategy profile $s$ such that there exists at least one player $i$ such that for all $j \in [k]-\{i\}$:
    \[
    u_j(s^*)\leq  u_j(s)
    \mbox{ and }
 u_i(s^*)   < u_i(s)
    \]
    
\end{definition}
\begin{definition}
    A strategy profile $s^*$ is \emph{weakly Pareto optimal} if there is no strategy profile $s$ such that for all $j\in [k]$: 
    \[
  u_j(s^*)   <u_j(s)
    \]
\end{definition}

\begin{remark}
        Every strong Nash equilibrium is weak Pareto optimal. This condition does not necessarily hold in a coalition-proof Nash equilibrium. A strong Nash equilibrium with strong Pareto optimality is called \emph{super Nash equilibrium}\cite{Rozenfeld}.
\end{remark}

\begin{remark}
    Informally, a strong Pareto optimal strategy is a strategy profile such that any alternative strategy will make at least one player worse off. A weak Pareto optimal is a strategy profile such that any alternative strategy will make at least one player worse off but may not make any player worse off.
\end{remark}

\subsection{Karush–Kuhn–Tucker (KKT) conditions}
The KKT conditions play a crucial role in optimization theory and are often used to analyze and solve constrained optimization problems.  KKT conditions are necessary conditions for a solution to an optimization problem with equality and inequality constraints. Consider the following nonlinear optimization problem in standard form:
$$
\begin{aligned}
& \operatorname{Min} f(x) \\
& \text { s.t }
g_i(x) \leq 0 \text{ and } e_j(x)=0   \\
.
\end{aligned}
$$
The optimization variable $x$ is chosen from a convex $X \subseteq \mathbb{R}^m$ and $f$ is the objective function. For all $i\in [m]$ and $j\in [l]$, $g_i$ are the inequality constraint functions and $e_j$ are the equality constraint functions respectively. The Lagrangian function can be defined as follows:
\[
\mathcal{L}(x, \mu, \kappa)=f(x)+\mu^{T} g(x)+\kappa^{T} e(x)
\]

\[
g(x)=\left[\begin{array}{c}
g_1(x) \\
\vdots \\
g_i(x) \\
\vdots \\
g_m(x)
\end{array}\right], \thickspace e(x)=\left[\begin{array}{c}
e_1(x) \\
\vdots \\
e_j(x) \\
\vdots \\
e_{l}(x)
\end{array}\right], \thickspace \mu=\left[\begin{array}{c}
\mu_1 \\
\vdots \\
\mu_i \\
\vdots \\
\mu_m
\end{array}\right],\thickspace \kappa=\left[\begin{array}{c}
\kappa_1 \\
\vdots \\
\kappa_j \\
\vdots \\
\kappa_{\ell}
\end{array}\right]
\]

The conditions can be defined as follows for minimizing $f(x)$:
\begin{itemize}
    \item Stationarity:
\[\partial f\left(x^*\right)+\sum_{j=1}^{\ell} \kappa_j \partial e_j\left(x^*\right)+\sum_{i=1}^m \mu_i \partial g_i\left(x^*\right) \ni \mathbf{0}\]

    \item Primal feasibility:
\[
\begin{aligned}
& e_j\left(x^*\right)=0, \text { for } j=1, \ldots, l \\
& g_i\left(x^*\right) \leq 0, \text { for } i=1, \ldots, m
\end{aligned}
\]

    \item Dual feasibility:
\[
\mu_i \geq 0, \text { for } i=1, \ldots, m
\]
    \item Complementary slackness:
\[\sum_{i=1}^m \mu_i g_i\left(x^*\right)=0\]

\end{itemize}

\begin{remark}
   We use KKT conditions in the remedial models proposed by \cite{Pang2005}. We also summarize the dual feasibility and complementary slackness in one property. Furthermore, in this paper, we do not have equality constraints. 
\end{remark}

\section{Classic Existence Results}\label{CLassical}

In this section, for the sake of completeness, we restate some well-known existence results related to problems that we study.

\subsection{Existence of GQVI}

 The existence of a solution of a generalized quasi-variational inequality is guaranteed by the Eilenberg-Montgomery fixed-point theorem \cite{chan}.

\begin{proposition}[\cite{chan}]\label{chanproof}
 Let $\mathcal{F}$ and $\mathcal{R}$ be correspondences from $\mathbb{R}^n$ into itself. Suppose that there exists a nonempty compact convex set $C$ such that:
 \begin{enumerate}
     \item $\mathcal{R}(C) \subseteq C$;
     \item $\mathcal{F}$ is a nonempty contractible compact valued upper semicontinuous correspondence on $C$;
     \item $\mathcal{R}$ is a nonempty continuous convex valued correspondence on $C$.
 \end{enumerate}
Then there exists a solution to $GQVI(\mathcal{R}, \mathcal{F})$.     
\end{proposition}
\begin{remark}
     Under the following additional assumption, instead of Eilenberg-Montgomery, Kakutani's fixed point theorem may be applied :

{\em
\begin{enumerate}\setcounter{enumi}{3}
    \item $\mathcal{F}$ is convex valued mapping on $C$.
\end{enumerate}}

\end{remark}

\subsection{Existence Result for Generalized Nash Equilibrium}

Debreu's theorem (Theorem \ref{thm:SocialEq}) specifies necessary conditions that guarantee the existence of a social equilibrium.

\begin{theorem}[\cite{Social}]
\label{thm:SocialEq}
Consider a (Debreu) game represented by $\mathcal{U}=(u_1,\dots,u_k)$ and $\mathcal{S}=(S_1,\dots,S_k)$ with the constraints imposed on the players $\mathcal{R}=(\mathcal{R}_1,\dots,\mathcal{R}_k)$, $1 \le i \le n$, and suppose that:
\begin{enumerate}
    \item $u_i$ is continuous (for all $i \in [k]$);
    \item For any $\lambda \in [0,1]$ $s_i,s'_i \in S_i$ , $\textbf{s}_{-i} \in S_{-i}$, $u_i(\lambda s_i + (1 - \lambda)s'_i,\textbf{s}_{-i}) \ge \min\{u_i(s_i,\textbf{s}_{-i}),u_i(s'_i,\textbf{s}_{-i})\}$;
    \item $\mathcal{R}_i$ is upper and lower semicontinuous for all $i \in [k]$.
    \item  For every $\textbf{s}_{-i} \in S_{-i}$, $\mathcal{R}_i(\textbf{s}_{-i})$ is convex and non-empty (convex-valued and non-empty).
\end{enumerate}
Then, the game specified by $\mathcal{U}$, $\mathcal{S}$, $\mathcal{R}$ admits a social (generalized) equilibrium.
\end{theorem}

\subsection{Rosen's Contribution and Concave Games}

 In \cite{rosen}, Rosen gives an alternate approach to specifying generalized equilibrium, and an existence result that guarantees an equilibrium under potentially weaker conditions. Rosen's result applies in a setting where there is a single \emph{commonly coupled constraint}, i.e., a relation $\mathcal{R} \subseteq S_1\times\dots\times S_n$ which constrains the best response for all players. In particular, (\ref{eq:bestsocresp}) becomes;
\begin{equation}
    \label{eq:bestcommonresp}
        u_i(\textbf{s}^*)+\epsilon \ge u_i(s_i,\textbf{s}^*_{-i}), \quad\quad\quad\quad  \forall s_i \text{\ s.t.\ } \mathcal{R}(s_i,\textbf{s}^*_{-i})\tag{$\dagger$}
\end{equation}
An equilibrium in this case is a strategy profile $s^*\in\mathcal{R}$ that satisfies (\ref{eq:bestcommonresp}).

Theorem~\ref{thm:Kakutani} can be used to guarantee the existence of an equilibrium in these problems. The required continuity conditions on $\mathcal{R}_i$ are often too strong to be satisfied in many settings. Rosen  defines  $\rho:\mathcal{R}\times\mathcal{R}\rightarrow\mathbb{R}$ by:
$\rho(\textbf{s},\textbf{t})=\sum_{i\in [k]}u_i(\textbf{t}_i,\textbf{s}_{-i})$
, and the correspondence $\Gamma:\mathcal{R} \rightrightarrows \mathcal{R}$ by:
\[
\Gamma(\textbf{s})=\{\textbf{t}~|~\rho(\textbf{s},\textbf{t})=\max_{\textbf{r}} \rho(\textbf{s},\textbf{r})\}
\]

\begin{theorem}[ \cite{rosen}]
\label{thm:coupledeq}
Assume  that for $i \in [n]$, $u_i(\textbf{s})$ is continuous at every $\textbf{s}\in \mathcal{S}$ and for each fixed value of $\textbf{s}_{-i}\in S_{-i}$ , $u_i(s_i,\textbf{s}_{-i})$ is concave in $s_i$, and that $\mathcal{R}$ is convex and compact. Then:
\begin{enumerate}
    \item $\Gamma$ admits a fixed point $\textbf{s}^*\in\Gamma(\textbf{s}^*)$.
    \item Any fixed point $\textbf{s}^*$  of  $\Gamma$, is an equilibrium in the sense of {\rm (\ref{eq:bestcommonresp})}
\end{enumerate}
\end{theorem}

Note that $S_{-i}=(S_1\times \dots S_{i-1},S_{i+1} \dots S_{k})$. It also should be noted that  Theorem~\ref{thm:coupledeq} guarantees the existence of an equilibrium for games with coupled constraints under reasonable assumptions. Because of their close connection with variational inequalities \cite{Generalizedqvari,Harker}, equilibria obtained in this way have been termed \emph{variational equilibria}, and their properties have been extensively studied (see Proposition \ref{geneqvari2}). We present the proof only for the sake of completeness and comparison as one of the techniques used in the proof will be used later (see Proposition \ref{geneqvari2}).

\begin{proof}
    First note that $\rho$ is continuous at every $(\textbf{s},\textbf{t})\in\mathcal{R}\times\mathcal{R}$, and that for $\textbf{s}\in\mathcal{R}$, $\rho(\textbf{s},\textbf{t})$ is concave in $\textbf{t}$. By assumption $\mathcal{R}$ is compact and it is clear that for any $\textbf{s}\in\mathcal{R}$, $\Gamma(\textbf{s})$ is nonempty and convex. So, $\mathcal{R}$ and $\Gamma$ satisfy the assumptions of Theorem~\ref{thm:Kakutani} (Kakutani's fixed point) and there exists $\textbf{s}^*\in\mathcal{R}$ such that $\textbf{s}^*\in \Gamma(\textbf{s}^*)$, which means that:
\begin{equation}
\label{eq:gammafix}
\rho(\textbf{s}^*,\textbf{s}^*)=\max_{\textbf{r}\in\mathcal{R}}\rho(\textbf{s}^*,\textbf{r}).\tag{$\dagger\dagger$}
\end{equation}
The fixed point $\textbf{s}^*$ must satisfy (\ref{eq:bestcommonresp}). If not, for some $i\in [n]$ and $s_i$ such that $\mathcal{R}(s_i,\textbf{s}^*_{-i})$, $u_i(s_i,\textbf{s}^*_{-i})>u_i(s^*_i,\textbf{s}^*_{-i})$. But then, for $\overline{s}=(s_i,\textbf{s}^*_{-i})$, $\rho(\textbf{s}^*,\overline{s})>\rho(\textbf{s}^*,\textbf{s}^*)$, contradicting (\ref{eq:gammafix}).
\end{proof}

\subsection{Existence Result for Remedial L/F Equilibrium}

 Theorem 5 (implied by Theorem 2) in \cite{Pang2005} provides the required conditions for the existence of a remedial $L/F$ equilibrium and we prove a stronger result which is inclusion in PPAD.

  \begin{theorem}[\cite{Pang2005}]
      Let $X^{\mathrm{I}}$ and $X^{\mathrm{II}}$ be nonempty, bounded polyhedra and assume the following conditions:
      \begin{enumerate}
          \item For each $\left(x_{\mathrm{I}}, x_{\mathrm{II}}\right) \in X^{\mathrm{I}} \times X^{\mathrm{II}}$, the functions $\phi_{\mathrm{I}}\left(\cdot, x_{\mathrm{II}}, \cdot\right)$ and $\varphi_{\mathrm{II}}\left(x_{\mathrm{I}}, \cdot, \cdot\right)$ are convex and continuously differentiable;
          \item For all $i=1, \ldots, k$, the function $c_i\left(x_{\mathrm{I}}, x_{\mathrm{II}}, y_{(-i,\mathrm{I})}\right)$ is affine, and the graphs of the four set-valued maps $W_{(i,\mathrm{I})}, V_{(i,\mathrm{I})}, W_{(i,\mathrm{II})}$, and $V_{(i,\mathrm{II})}$ are polyhedra;
          \item For each $\left(x_{\mathrm{I}}, x_{\mathrm{II}}\right) \in X^{\mathrm{I}} \times X^{\mathrm{II}}, Z^{\mathrm{I}}\left(x_{\mathrm{I}}, x_{\mathrm{II}}\right)$ and $Z^{\mathrm{II}}\left(x_{\mathrm{I}}, x_{\mathrm{II}}\right)$ are nonempty;
\item $Z^{\mathrm{I}}\left(X^{\mathrm{I}}, X^{\mathrm{II}}\right)$ and $Z^{\mathrm{II}}\left(X^{\mathrm{I}}, X^{\mathrm{II}}\right)$ are bounded.
      \end{enumerate}
      Then there exists a remedial $L/F$ equilibrium.
  \end{theorem}

\section{Elementary Proofs}\label{AppendixElementaryProofsection}

\subsection{Elementary Properties of Quasi Variational Inequalities}

The following propositions will play an important part in establishing a relationship between quasi-variational inequalities and Debreu-Rosen style games (see Proposition \ref{geneqvari2}).

\begin{proposition}\label{minimprop}
   Let $x^*$ be a solution to the following optimization problem where $f$ is continuously differentiable and $\mathcal{R}$ is closed and convex correspondence:
\begin{equation}\label{minprob}
\begin{aligned}
    &\operatorname{Min}  \ f(x) \\
&\text { s.t }   \ x \in \mathcal{R}(x)
\end{aligned}
\end{equation}
Then $x^*$ is a solution of the quasi-variational inequality problem ($F=\nabla f$  where $\nabla$ refers to the gradient operator):
\begin{equation} \label{minim}
 \left(y-x^*\right) ^T F\left(x^*\right) \geq 0, \quad \forall y \in \mathcal{R}(x^*)   
\end{equation}

\end{proposition}

\begin{proof}
     Let $\Psi(t)=f\left(x^*+t\left(y-x^*\right)\right)$, for $t \in[0,1]$. Since $\Psi(t)$ gets its minimum at $t=0$, $0 \leq \Psi^{\prime}(0)= \left(y-x^*\right)^T \nabla f\left(x^*\right)$ which means $x^*$ is a solution of Equation \ref{minim}.
\end{proof}

\begin{proposition}\label{quasisolutionmin}
Suppose that $f$ is a differentiable convex function and $x^*$ is a solution to $\operatorname{QVI}(\nabla f, \mathcal{R})$, then $x^*$ is a solution to the minimization problem  (\ref{minprob}).

\end{proposition}

\begin{proof}
    
By the fact that $f$ is convex:
$$
f(y) \geq f\left(x^*\right)+  \left(y-x^*\right)^T \nabla f \left(x^*\right), \quad \forall y \in \mathcal{R}(x^*) 
$$
But $\left(y-x^*\right)^T  \nabla f\left(x^*\right)\geq 0$, since $x^*$ is a solution to $\operatorname{QVI}(\nabla f, \mathcal{R})$. Therefore, we can conclude that:
$$
f(y) \geq f\left(x^*\right), \quad \forall y \in \mathcal{R}(x^*) 
$$
that is, $x^*$ is a solution  (minimum point) for Equation \ref{minprob}.
\end{proof}


\subsection{Proof of Proposition \ref{geneqvari2}}
To prove Proposition \ref{geneqvari2}, we first inspect the exact version of the problem without considering the approximation\footnote{This version is easier to establish and also is a well-known proposition in literature.}.

\begin{proposition}\label{exgenqvi}
  Suppose that we have a game with concave and continuously differentiable utilities\footnote{When the utilities are non-differentiable, we need to introduce a more general notion of stationarity (see \cite{Hu2011}).} $\mathcal{U}=(u_1,\dots,u_k)$,  and strategies $\mathcal{S}=(S_1,\dots,S_k)$ and constraints $\mathcal{R}=(\mathcal{R}_1,\dots,\mathcal{R}_k)$  which are closed convex sets of $\mathcal{S}=(S_1,\dots,S_k)$.  The problem of finding a generalized equilibrium in this game can be transformed into a QVI problem.  
\end{proposition}
\begin{proof}
    Define $\theta_i=-u_i$ for each $i\in [k]$ to denote the \emph{loss} function of each player $i$. $\theta_i(\cdot,\textbf{x}_{-i})$ is convex in $\textbf{x}_{-i}$. Finding a generalized equilibrium in this game is equivalent to finding a solution to the following optimization problem where the goal is finding $\textbf{x}^*=(x^*_{i},\textbf{x}^*_{-i})$ such that  $\forall i\in [k] $  :

\begin{equation}\label{mingener}
\begin{aligned}
 & \operatorname{Min} \theta_i(x_i,\textbf{x}^*_{-i}) \\
&\text { s.t } x_i \in \mathcal{R}_i (\textbf{x}^*_{-i})
\end{aligned}
\end{equation}

    Define $\mathcal{R}(\textbf{x}) \equiv \prod_{i=1}^k \mathcal{R}_i\left(\textbf{x}_{-i}\right)$
and
$
F(\textbf{x})=\left(\nabla_{x_i} \theta_i(\textbf{x})\right)_{i=1}^k \in \mathbb{R}^{k}
$. Based on Proposition \ref{minimprop} and \ref{quasisolutionmin}, we can see that $x^*$ is a generalized equilibrium if and only if $\textbf{x}^* \in \mathcal{R} \left(\textbf{x}^*\right)$ and it satisfies:
\begin{equation} 
 \left(\textbf{y}-\textbf{x}^*\right) ^T F\left(\textbf{x}^*\right) \geq 0, \quad \forall \textbf{y} \in \mathcal{R}(\textbf{x}^*)   
\end{equation}

\end{proof}
\begin{corollary}\label{nashvi1}

 Suppose that we have a game with the same conditions on utilities $\mathcal{U}=(u_1,\dots,u_k)$,  and strategies $\mathcal{S}=(S_1,\dots, S_k)$. The problem of finding a Nash equilibrium of this game can be transformed into a VI problem.
\end{corollary}

Finally, we proceed to the approximate version. We provide detailed proof for the approximate version (Proposition \ref{geneqvari2}) for the sake of completeness\footnote{Also, this proposition is a starting point for Lemma \ref{reslient2k} in which we establish the relationship between the more generalized variation inequalities and resilient Nash equilibrium.}.

\subsubsection{Proof of Proposition \ref{geneqvari2}}
\begin{proof}
    Similarly, let $\theta_i=-u_i$ for each $i\in [k]$. Finding an approximate generalized equilibrium in this game is equivalent to finding a solution to the following optimization problem where the goal is  finding $\textbf{x}^*=(x^*_{i},\textbf{x}^*_{-i})$:

\begin{equation}\label{mingener2}
\begin{aligned}
 &  \forall i\in [k],\quad \theta_i(\textbf{x}^*) \leq \operatorname{Min} \theta_i(x_i,\textbf{x}^*_{-i}) +\epsilon \\
&\text {s.t } x_i \in \mathcal{R}_i (\textbf{x}^*_{-i})
\end{aligned}
\end{equation}

    Define $\mathcal{R}(\textbf{x}) \equiv \prod_{i=1}^k \mathcal{R}_i\left(\textbf{x}_{-i}\right)$, $
f(\textbf{x})=\prod_{i=1}^k\left(\theta_i(\textbf{x})\right) \in \mathbb{R}^{k}
$,
and
$
F(\textbf{x})=\prod_{i=1}^k \left(\nabla_{x_i} \theta_i(\textbf{x})\right) \in \mathbb{R}^{k}
$. We will show that $\textbf{x}^*$ is a generalized (constrained) equilibrium if and only if $\textbf{x}^* \in \mathcal{R} \left(\textbf{x}^*\right)$ and it satisfies:
\begin{equation}\label{gogo}
 \left(y-x^*\right) ^T F\left(x^*\right)+\epsilon \geq 0, \quad \forall y \in \mathcal{R}(x^*)   
\end{equation}

Assume that we have a solution for Equation \ref{gogo}. By the fact that $\theta_i(\cdot,\textbf{x}_{-i})$ is convex:
$$
\theta_i(y_i,\textbf{x}^*_{-i}) \geq \theta_i\left(x^*_i,\textbf{x}^*_{-i}\right)+  \left(y_i-x^*_i\right) \nabla_{x^*_i} \theta_i \left(x^*_i,x^*_{-i}\right), \quad \forall y_i \in \mathcal{R}_i(\textbf{x}^*_{-i}) 
$$
And also:
\begin{equation}\label{alleqy2}
    \sum_{i=1}^k \theta_i(y_i,\textbf{x}^*_{-i}) \geq \sum_{i=1}^k \theta_i\left(x^*_i,\textbf{x}^*_{-i}\right)+ \sum_{i=1}^k \left(y_i-x^*_i\right) \nabla_{x^*_i} \theta_i \left(x^*_i,\textbf{x}^*_{-i}\right), \quad \forall y_i \in \mathcal{R}_i(\textbf{x}^*_{-i})
\end{equation}

But $\left(y-\textbf{x}^*\right)^T  \nabla f\left(\textbf{x}^*\right)+\epsilon \geq 0$, since $\textbf{x}^*$ is a solution to $\operatorname{QVI}(\nabla f, \mathcal{R})$. Therefore, we can conclude that:
\begin{equation}\label{alleqy}
   \sum_{i=1}^k \theta_i(y_i,\textbf{x}^*_{-i}) + \epsilon \geq \sum_{i=1}^k \theta_i\left(x^*_i,\textbf{x}^*_{-i}\right)   \quad \forall y_i \in \mathcal{R}_i(\textbf{x}^*_{-i})  
\end{equation}
Similar to Rosen's proof, for the sake of contradiction, assume that there exists one player $i\in[k]$ such that for one alternative constrained strategy $y_i$:
$$  \theta_i(y_i,\textbf{x}^*_{-i})   <\theta_i\left(x^*_i,\textbf{x}^*_{-i}\right) +\epsilon $$

Let us focus on a specific example in Equation \ref{alleqy}. For all $j\in [k]-i$, let $y^*_j=x^*_j$ and let $y^*_i= \operatorname{argmin}_{y_i\in \mathcal{R}_i(x^*_{-i})} \theta_i(y_i,\textbf{x}^*_{-i})$. The equation will be simplified as follows:
$$  \theta_i(y^*_i,\textbf{x}^*_{-i}) +\epsilon  \geq \theta_i\left(x^*_i,\textbf{x}^*_{-i}\right) $$
This obviously is a contradiction as $y^*_i$ is a better alternative compared to $y_i$. In other words, this results $\left(\textbf{y}^*-\textbf{x}^*\right)^T  \nabla f\left(\textbf{y}^*\right)+\epsilon < 0$ which violates Equation \ref{gogo}. In conclusion for any $y_i \in \mathcal{R}_i(\textbf{x}^*_{-i})$, we have the following relationship:
$$ \forall i\in [k],\quad \theta_i(y_i,\textbf{x}^*_{-i})+\epsilon \geq \theta_i\left(\textbf{x}^*_i,\textbf{x}^*_{-i}\right) $$

To prove the other direction of the proof, we may equivalently consider the  alternative following form of Equation \ref{mingener2} where $f(y)$ has the loss function of all players in their respective coordinates \footnote{This proof is a just the approximate version of Proposition \ref{minimprop}}:

\begin{equation}\label{equivfx}
   f(\textbf{y})+\epsilon \mathds{1} \geq f\left(\textbf{x}^*\right), \quad \forall \textbf{y} \in \mathcal{R}(\textbf{x}^*) 
\end{equation} which means $\textbf{x}^*$ is a solution  (minimum point) for Equation \ref{mingener2}. Note that $\mathds{1}$ denotes all ones vector.

 Now, assume that $\textbf{x}^*$ is a solution of Equation \ref{mingener2} (or equivalently, Equation \ref{equivfx}).    Let $\Psi(t)=f\left(\textbf{x}^*+t\left(\textbf{y}-\textbf{x}^*\right)\right)$, for $t \in[0,1]$. Since $\Psi(t)$ gets its approximate minimum at $t=0$, $0 \leq \Psi^{\prime}(0)= \left(\textbf{y}-\textbf{x}^*\right)^T \nabla f\left(\textbf{x}^*\right)+\epsilon$ and this means $\textbf{x}^*$ is a solution of Equation \ref{gogo}.
\end{proof}

\subsection{Linear Arithmetic Circuits}\label{arithmeticcircuits}

In this section, we restate the theorem E.1 of \cite{Fearnley} in which the functions that can be computed by arithmetic circuits can be approximated by linear arithmetic circuits with minimal error.  An arithmetic circuit $f$ is well-behaved if, on any directed path that leads to output, there are at most $log(size(f))$ true multiplication gates. A true multiplication gate is one where both inputs are non-constant nodes of the circuit. Note that in linear arithmetic circuits, we do not have multiplication gates.

\begin{theorem}[\cite{Fearnley}]
    Given a well-behaved arithmetic circuit $f:[0,1]^n \rightarrow \mathbb{R}^d$, a purported Lipschitz constant $L>0$, and a precision parameter $\epsilon>0$, in polynomial time in $\operatorname{size}(f), \log L$ and $\log (1 / \varepsilon)$, we can construct a linear arithmetic circuit $F:[0,1]^n \rightarrow \mathbb{R}^d$ such that for any $x \in[0,1]^n$ it holds that: 
\end{theorem}
\begin{itemize}
    \item $\|f(x)-F(x)\|_{\infty} \leq \epsilon$, or
    \item given $x$, we can efficiently compute $y \in[0,1]^n$ such that:
    $$
\|f(x)-f(y)\|_{\infty}>L\|x-y\|_{\infty} .
$$

Here "efficiently" means in polynomial time in $\operatorname{size}(x), \operatorname{size}(f), \log L$ and $\log (1 / \epsilon)$.

\end{itemize}

\subsection{Some Properties of L2 Distance}
 Throughout the paper, we use some fundamental lemmas pertaining to the properties of L2 distance.

\begin{lemma}\label{norm2}
The function  $f(x)=\gamma \|x\|^2_2$
is $2\gamma$-strongly convex.
\end{lemma}
\begin{proof}
      The key lies in leveraging the fact that the $L2$ is arising from an inner product.
\[\begin{aligned}
& f(\lambda x+(1-\lambda) y)-\left(\lambda f(x)-(1-\lambda) f(y)\right) \\
& =\gamma \lambda^2\|x\|^2+2 \gamma \lambda(1-\lambda)\langle x, y\rangle+\gamma(1-\lambda)^2\|y\|^2 \\
& -\gamma \lambda\| x\left\|^2-\gamma(1-\lambda)\right\| y \|^2 \\
& =\gamma \lambda(\lambda-1)\|x\|^2+\gamma(1-\lambda)(1-\lambda-1)\|y\|^2s  +2 \gamma \lambda(1-\lambda)\langle x, y\rangle \\
& =-\gamma \lambda(1-\lambda)\|x\|^2+2 \gamma \lambda(1-\lambda)\langle x, y\rangle-\gamma \lambda(1-\lambda) \|y\|^2 \\
& =-\gamma \lambda(1-\lambda)\left(\|x\|^2-2\langle x, y\rangle+\|y\|^2\right) \\
& =-\gamma \lambda(1-\lambda)\|x-y\|^2 \\
& \leq-\frac{1}{2} \gamma \lambda(1-\lambda)\|x-y\|^2 \\
& \text { since } \gamma \lambda(1-\lambda)\|x-y\|^2 \geq 0 \text {. } \\
&
\end{aligned}\]
\end{proof}

\begin{corollary}
    The function  $f(x)=-\gamma \|x\|^2_2$
is $\gamma$-strongly concave.
\end{corollary}

\begin{lemma}\label{norm22}
    Suppose that $f$ and $g$ are $\gamma_f$ and $\gamma_g$ strongly concave functions where $f,g: X\rightarrow Y$. Then, $f+g$ is a  $\gamma_f+\gamma_g$-strongly concave function.
\end{lemma}
  \begin{proof}
    Suppose that $x,y \in X$. We have the following:
\[
\begin{aligned}
    &(f+g)(\lambda x+(1-\lambda) y)-(\left(\lambda (f+g)(x)-(1-\lambda) (f+g)(y)\right)\\
    &\geq \frac{\lambda(\lambda-1)}{2} \gamma_f \|x-y\|^2_2+ \frac{\lambda(\lambda-1)}{2}\gamma_g \|x-y\|^2_2=  \frac{\lambda(\lambda-1)}{2}(\gamma_f+\gamma_g) \|x-y\|^2_2
\end{aligned} \]

\end{proof}

\subsection{Formulating L/F Equilibrium as EPEC}

\begin{proposition}\label{qvileaderfollower}
Assume that in a multi-leader-follower game , we have $k$ followers  with continuously differentiable loss functions $\Theta=(\theta_1,\dots,\theta_k)$,  and strategies $\mathcal{S}=(S_1,\dots,S_k)$ and constraints $\mathcal{R}=(\mathcal{R}_1,\dots,\mathcal{R}_k)$. In addition, we have $2$ leaders with strategies $\mathcal{X}=(X^{\mathrm{I}},X^{\mathrm{II}})$ and utilities $\Phi=(\phi_{\mathrm{I}},\phi_{\mathrm{II}})$. The problem of finding a $L/F$-equilibrium of this game can be transformed into an EPEC problem.
\end{proposition}

\begin{proof}
    A pair $(x^{*}_\mathrm{I}, x^{*}_\mathrm{II}) \in X^{\mathrm{I}} \times X^{\mathrm{II}}$ is a \emph{$L/F$-equilibrium}, if there exists $(y^{*}_\mathrm{I}, y^{*}_\mathrm{II})$ such that $(x^{*}_\mathrm{I}, y^{*}_\mathrm{I})$ is an optimal solution of leader $\mathrm{I}$'s problem, which tries to find $(x_{\mathrm{I}}, y_{\mathrm{I}})$ such that:
\[
\begin{aligned}
& \operatorname{Min} \phi_{\mathrm{I}}\left(x_{\mathrm{I}}, x^{*}_\mathrm{II}, y_{\mathrm{I}}\right) \\
& \text { s.t } x_{\mathrm{I}} \in X^{\mathrm{I}} \\
& \text { and } y_{\mathrm{I}} \mbox{  solves } QVI(Y(x_\mathrm{I},x^{*}_\mathrm{II}), F_\mathrm{I} (x_\mathrm{I},x^{*}_\mathrm{II},\cdot))
\end{aligned}
\]where $F_\mathrm{I}$ is defined as follows:


\[F_\mathrm{I}\left(x_\mathrm{I}, x^*_\mathrm{II}, y_\mathrm{I}\right):=\left(\begin{array}{c}
\partial_{y_{(1,\mathrm{I})}} \theta_1\left(x_\mathrm{I},x^*_\mathrm{II}, y_{(1,\mathrm{I})}, y^*_{(-1,\mathrm{I})}\right) \\
\vdots \\
\partial_{y_{(k,\mathrm{I})}} \theta_k\left(x_\mathrm{I},x^*_\mathrm{II}, y_{(k,\mathrm{I})}, y^*_{(-k,\mathrm{I})}\right)
\end{array}\right)\]
 For clairty, all exogenous variables are marked by $*$. The notation $y^*_{(-k,\mathrm{I})}$ refers to all followers' anticipated strategies except $k$ for leader $\mathrm{I}$. A similar transformation can be used for the second leader.

\end{proof}


\newpage

\section{Total t-Resilient Nash is in PPAD}\label{Appendixtotalresilientinppad}

This section is organized as follows. First, we generalize variational problems to the case that a set of equations of variational inequalities is supported. We will also state and proof a more generalized format of the robust version of Berge's maximum theorem. Finally, the proof (of inclusion in PPAD for $t$-total resilient Nash) can be established by generalizing Corollary \ref{nashvi2} (see Lemma \ref{reslient2k}).

\subsection{A More Generalized Form of GQVI}
For completeness and brevity, we only investigate the most complicated case which is GQVI. Computational MQVI and MVI can be defined similarly. To avoid repetition, we do not investigate this problem for weak separation oracles. The computational problem (MGQVI) with strong separation oracles is defined as follows:

\problemStatement{$MGQVI(\mathcal{F},\mathcal{R})$ With strong Separation Oracles}{
  Input={ We receive as input all the following:
  \begin{itemize}
      \item  A circuit $C_{\mathcal{R}}$ which represents a strong separation oracle for a $L_{\mathcal{R}}$-Hausdorff Lipschitz correspondence $\mathcal{R}: \mathbb{R}^{m*} \rightrightarrows  \mathbb{R}^{m*}$,
      \item $\mathcal{F}$ which has $r$ circuits such that for each $i\in [r]$, $C^i_\mathcal{F}$ represents a strong separation oracle for a $L_\mathcal{F}$-Hausdorff Lipschitz and $\gamma$-strongly convex correspondence $\mathcal{F}^i: \mathbb{R}^{m*} \rightrightarrows  \mathbb{R}^{m*}$,
      \item    An accuracy parameter $\beta$. 
  \end{itemize}},
  Output={One of the following cases:
  \begin{itemize}
      \item (Violation of non-emptiness): A vector $x \in \mathbb{R}^{m*}$ such that $\overline{\mathrm{B}}(\mathcal{R}(x),-\epsilon)=\emptyset$ or an index $i$ such that for some $x \in \mathbb{R}^{m*}$ we have that $\overline{\mathrm{B}}(\mathcal{F}^i(x),-\epsilon)=\emptyset$,
\item  (Violation of $L_{\mathcal{R}}$-Hausdorff Lipschitzness of $\mathcal{R}$):
Four vectors $p, q, z, w \in \mathbb{R}^{m*}$ and a constant $\epsilon>0$ such that $w=\widetilde{\Pi}_{\mathcal{R}(q)}^{\epsilon}(q)$ and $z=\widetilde{\Pi}_{\mathcal{R}(p)}^{\epsilon}(w)$ but $\|z-w\|>L_{\mathcal{R}}\|p-q\|+3 \epsilon$.
\item  (Violation of $L_\mathcal{F}$-Hausdorff Lipschitzness of $\mathcal{F}$)
Four vectors $p, q, z, w \in \mathbb{R}^{m*}$ and a constant $\epsilon>0$ and one index $i$ such that $w=\widetilde{\Pi}_{\mathcal{F}^i(q)}^{\epsilon}(q)$ and $z=\widetilde{\Pi}_{\mathcal{F}^i(p)}^{\epsilon}(w)$ but $\|z-w\|>L_\mathcal{F}\|p-q\|+3 \epsilon$,
(Violation of strong convexity):
There vectors $x, p, q\in \mathbb{R}^{m*}$ and two constants $\epsilon>0$ and $\lambda\in (0,1)$ such that:
$$
\begin{aligned}
\widetilde{\Pi}_{\mathcal{F}^i(x)}^{\epsilon}(\lambda p+ (1-\lambda)q ) & > \lambda \cdot \widetilde{\Pi}_{\mathcal{F}^i(x)}^{\epsilon}(p)+(1-\lambda) \cdot \widetilde{\Pi}_{\mathcal{F}^i(x)}^{\epsilon}(q)
\\&-\frac{\lambda(1-\lambda)}{2} \cdot \gamma\cdot\left\|p-q\right\|_2^2+\epsilon
\end{aligned}
$$
\item Two tuples $(x,w)$ and $(x^*,w^*)$ with $\|x-x^*\|\leq \beta $ and $\|w-w^*\|\leq \beta$  such that  $x^*\in \mathcal{R}(x) $ and  $w^*\in \mathcal{F}(x)$ and $\left(y-x\right) ^T w^* +\beta \mathds{1} \geq 0, \quad \forall y \in \mathcal{R}(x)$  \footnote{Recall that $\mathds{1}$ is all ones vector.}.   
  \end{itemize}
  }
}
\begin{theorem}\label{D1theorem}
    The problem of finding a solution to MGQVI for strong separation oracles is in PPAD.
\end{theorem}
\begin{proof}
  We again consider $[-1,1]^m$ instead of $\mathbb{R}^{m*}$.   First, we define the following \footnote{Note that, unlike Theorem \ref{concaveppad},  $\Phi$ is a function that takes two vectors $x$ and $y$ (of size $m\times 1$), one matrix $w$ ( of size $m\times r$ which has $r$ vectors), and outputs one vector.}:
  \[\Phi(y,x,w)=-(y-x)^T w-\gamma(||y||^2_2)\mathds{1}\]
  \[\Pi(x,w)=\{y\in \mathcal{R}(x) ~|~\Phi(y,x,w)>max_{y\in \mathcal{R}(x)}\Phi(y,x,w)-\epsilon \mathds{1} \}\]
 Next, we show that for a constant $\kappa^\prime$, $\Psi(x,w)=(\Pi(x,w),\mathcal{F}(x))$ is $\kappa^\prime$-Lipschitz continuous in $(x,w)$.  The next step will similarly be constructing a strong separation oracle for $\Psi$.  Finally, we show that an approximate solution to Kakutani's fixed point of this correspondence will provide an approximate solution to the given MGQVI.

To show that $\Psi$ is $\kappa^\prime$-Lipschitz continuous, we need to show that $d_H(\Psi(x_1,w_1),\Psi(x_2,w_2))\leq \kappa^\prime||(x_1,w_1)-(x_2,w_2)||^q_p+c$ for some constants $q$,$p$ and $c$. 
\begin{remark}
Note that for any matrix $A$:
$$
\|A\|_2=\sqrt{\lambda_{\max }\left(A^T A\right)}=\sigma_{\max }(A) .
$$
where $\sigma_{\max }(A)$ represents the largest singular value of matrix $A$ and $\lambda_{max}$ denotes the largest eigenvalue value. We know that $w_1$ and $w_2$ are matrices, then we can define: 
    $||(x_1,w_1)-(x_2,w_2)||^2_2=||x_1-x_2||^2_2+||w_1-w_2||^2_2$
\end{remark}

$\Phi(y,x,w)$ is $(2\gamma)$-strongly concave (with respect to each coordinate) in $y$ as we have a minor change that $\Phi$ is a vector. In addition, $\Phi_i$ ($i$-th coordinate of $\Phi$) is $G$-Lipschitz continuous where $G=m+2\gamma m$ for all $x$ and $w$. 
We also have access to the sub-gradients of $\Phi_i$ for each $i\in [r]$.

Let $H(x,w)=\{y\in \mathcal{R}(x) ~|~\Phi(y,x,w)=max_{y\in \mathcal{R}(x)} \Phi(y,x,w)\}$. By optimality KKT conditions for maximization of a concave function with respect to the constraint set $\mathcal{R}(x)$ for all $y\in \mathcal{R}(x) $ we have:
\[
\partial \Phi(y,x,w)^T(y^*-y)\geq 0 \text{ where } y^*=argmax_{y\in \mathcal{R}(x)}\Phi(y,x,w)
\]
In addition, $(2 \gamma)$-strong-concavity of $\Phi(\cdot,x,w)$ results the following inequality:
\[
\Phi\left(y^*,x,w\right)-\Phi(y,x,w) \geq \partial \Phi\left(y^*,x,w\right)^{\top}\left(y^*-y\right)+\gamma\left\|y-y^*\right\|_2^2 \mathds{1}
\]
Finally, by combining the previous inequalities we have:
\[
\Phi\left(y^*,x,w\right)-\Phi(y,x,w) \geq \gamma\left\|y-y^*\right\|_2^2 \mathds{1}
\]
In conclusion, let $\gamma^\prime=r\gamma$, for $y^*\in H(x,w)$ and $y\in \Pi(x,w)$, we have the following inequality:

\begin{equation}\label{mkapp}
    \gamma^\prime\left\|y-y^*\right\|^2_2 \leq \epsilon \text{ or equivalently } \left\|y-y^*\right\|_2 \leq \sqrt{\frac{\epsilon}{\gamma^\prime}}
\end{equation}



 Now, we apply Theorem \ref{berge2} to $f((x,w),y)=\Phi(y,x,w)$ and $g((w,x))=\mathcal{R}(x)$ which is a modified version of Berge's theorem,  where $g$ is a non-empty, compact and convex-valued correspondence.  By Theorem \ref{berge2}, there exists a constant $\kappa$ such that $H$ is $\kappa$-Lipschitz (Holder) continuous. In conclusion, we have the following:
\begin{equation}\label{mkapp2}
    d_H(H(x_1,w_1),H(x_2,w_2))\leq \kappa||(x_1,w_1)-(x_2,w_2)||^{\frac{1}{2}}_2+c
\end{equation}
Combining Equations \ref{mkapp} and \ref{mkapp2}  results in:
\begin{equation} \label{mkapp3}
\begin{split}
 d_H(\Pi(x_1,w_1),\Pi(x_2,w_2)) &\leq d_H(H(x_1,w_1),H(x_2,w_2))\\
 & + d_H(\Pi(x_1,w_1),H(x_1,w_1))\\ &+d_H(\Pi(x_2,w_2),H(x_2,w_2))
\end{split}
\end{equation}
In conclusion, we have:
\begin{equation}\label{mkapp4}
    d_H(\Pi(x_1,w_1),\Pi(x_2,w_2))\leq \kappa||(x_1,w_1)-(x_2,w_2)||^{\frac{1}{2}}_2+2\sqrt{\frac{\epsilon}{\gamma^\prime}}+c
\end{equation}
Since for each $i\in[r]$, $\mathcal{F}^i$ is $L_\mathcal{F}$-Hausdorff Lipschitz, by the definition of $\Psi$ we can deduce there exists a constant $\kappa^\prime$: 

\begin{equation}\label{mkapp5}
    d_H(\Psi(x_1,w_1),\Psi(x_2,w_2))\leq \kappa^\prime||(x_1,w_1)-(x_2,w_2)||^{\frac{1}{2}}_2+c^\prime
\end{equation}

 To establish a reduction employing the computational Kakutani's problem variant, it is essential to establish the existence of a bounded-radius ball within the correspondence $\Psi(x,w)$. Let $y_x^*=argmax_{y\in \mathcal{R}(x)}\Phi(y,x,w)$. We know that that $\overline{\mathrm{B}} (\mathcal{R}(x),-\epsilon)$ and also $\overline{\mathrm{B}} (\mathcal{F}(x),-\epsilon)$ are  non-empty. Then, define a none-empty region $V_x= \overline{\mathrm{B}}(y_x^{\star},+\frac{\epsilon}{G^\prime})$  where $G^\prime=r\cdot G$ and $G$ is the Lipschitz constant of $\Phi_i$. By Lipschitzness of $\Phi\left(\cdot,x,w\right)$, we have:
\[ \forall y \in V_x, \quad \left\|\Phi\left(y, x,w\right)-\Phi\left(y^*_x,x,w\right)\right\| \leq G^\prime\frac{\epsilon}{G^\prime}=\epsilon\]
 In conclusion, $\Phi(V_x,x,w)\subseteq [\Phi\left(y^*_x,x,w\right)-\epsilon,\Phi\left(y^*_x,x,w\right)]$ is non-empty. Then, we can show that  $\Phi(V_x,x,w) \subseteq\Pi(x,w)$. We also know that for any $x\in \mathcal{R}(x)$,  $\overline{\mathrm{B}} (\mathcal{F}(x),-\epsilon)$ is non-empty,  then, we can find a non-empty region $C_x$ such that  $C_x \subseteq\overline{\mathrm{B}} (\mathcal{F}(x),-\epsilon)$.  This is sufficient to show that for all $x$ and $w$, $ \Psi(x,w)$ is non-empty. 


Now we are ready to construct a strong separation oracle for $\Psi(x,w)$ leveraging a modified version of the Strong Constrained Convex Optimization framework stated in the following.

\problemStatement{Modified Strong Constrained Convex Optimization Problem (2)}{
  Input={A zeroth and first order oracle for the concave function $G: \mathbb{R}^{m*}\rightarrow \mathbb{R}^{r*}$, two rational numbers $\delta,\epsilon>0$ and a strong separation oracle $\mathrm{SO}_{\mathcal{R}}$ for a non-empty closed convex-valued correspondence $\mathcal{R}:\mathbb{R}^{m*}\rightrightarrows  \mathbb{R}^{m*}$ and one input $x$.
  },
  Output={One of the following cases:
  \begin{itemize}
      \item (Violation of non-emptiness)
A failure symbol $\perp$ with a polynomial-sized witness that certifies that $\overline{\mathrm{B}}(\mathcal{R}(x),-\delta)=\emptyset$,
\item (Approximate Maximization)
A vector $z \in \mathbb{Q}^m \cap \mathcal{R}(x)$, such that $G(z) +\epsilon \mathds{1} \geq \max _{y \in \mathcal{R}(x)} G(y) $ for all values.
  \end{itemize}
  }
}

We can compute a solution $y^*\in \mathcal{R}(x)$ using the modified sub-gradient ellipsoid central cut method such that $\Phi\left(y^*,x,w\right)\geq max_{y\in \mathcal{R}(x)} \Phi\left(y,x,w\right)-\epsilon\mathds{1} $.  Thus, we can substitute a SO for $\Psi$ by considering a strong separation oracle for the following set:
\[
\overline{\Psi}_{s}(x,w)=\{(y,w)\in (\mathcal{R}(x),\mathcal{F}(x))~|~- \Phi(y,x,w)\leq \gamma^\prime\}
\]
where $\gamma^\prime= -\Phi(y^*,x,w)$. In other words, we are looking for $(y,w)\in (\mathcal{R}(x),\mathcal{F}(x))$ such that $(y-x)^T w +\gamma||y||^2_2 \mathds{1}  \geq (y^*-x)^T w+ \gamma||y^*||^2_2 \mathds{1} $ given the strong separation oracles representing $\mathcal{F}$ and $D$\footnote{This modified strong constrained convex optimization problem could be solved by the sub-gradient ellipsoid central too (see Appendix \ref{AppendixCentralcutellipsoidsection}).}.

Now, we can give $\Psi$ as input to the Kakutani problem presented in \cite{concavegames}  with accuracy parameter $\alpha=\frac{\epsilon^\prime}{\kappa^\prime}$ where $\epsilon^\prime=\frac{\epsilon h}{\mathrm{i}}$ and a slightly modified input. Our input is a separation oracle for $\Psi$ which takes one vector $x$ and one matrix $w$. Kakutani's problem takes a correspondence in the format of $\mathcal{R}:\mathbb{R}^{d*}\rightrightarrows \mathbb{R}^{d*}$ . Here, we need to consider $d=mr+m$ and encode the given matrix $w$ with $r$-vectors to make it completely compatible with the conditions of the Kakutani's fixed point\footnote{Alternatively, we can consider a format of Kakutani's fixed point such that it works with the correspondences of this format $\mathcal{R}:\mathbb{R}^{a*}\rightrightarrows \mathbb{R}^{b*}$. However, this will require a complete rewrite.}.


The (decoded) output of this Kakutani instance will be two points $(x,w)\in ([-1,1]^m,[-1,1]^m)$ and $z=(x^*,w^*)\in\Psi(x,w)$  where $||(x,w)-(x^*,w^*) ||^2_2\leq \frac{\epsilon^\prime \mathrm{i} }{\kappa^\prime}$ and $d((x,w),\Psi(x,w))\leq \frac{\epsilon^\prime \mathrm{i}}{\kappa^\prime}$ where $\mathrm{i}$ is the constant error caused by the simple decoding\footnote{This error is caused by the fact that the norm that we defined for the input given as a matrix and a vector is slightly different from $L2$ distance. This does not cause a problem as the matrix has constant columns and these different norms are related.}. Thus, $w^*\in \mathcal{F}(x)$ and $d(w,\mathcal{F}(x))\leq \frac{\epsilon^\prime \mathrm{i}}{\kappa^\prime} $. In addition, $x^*\in \Pi(x,w)$ and $d(x,\Pi(x,w))\leq \frac{\epsilon^\prime \mathrm{i}}{\kappa^\prime} $.
By the definition of $\Pi$,  for every $y\in \mathcal{R}(x)$,  $\Phi\left(x^*,x,w\right)\geq \Phi\left(y,x,w\right)$. In conclusion:
\[ \left(y-x\right) ^T w+2\gamma^\prime||y||^2_2 \mathds{1}  \geq \left(x^*-x\right) ^T w+ 2\gamma^\prime||x^*||^2_2 \mathds{1} ,  \thickspace \forall y \in \mathcal{R}(x) \]
So, recalling that $x,y, w\in [-1,1]^m $, if $\gamma^\prime$ is small there exists a vector $u(\gamma^\prime)$:
\[ \left(y-x\right) ^T w \geq  \pm \frac{\epsilon^\prime}{\kappa^\prime}\cdot m \mathds{1}   \pm u(\gamma^\prime) ,  \quad \forall y \in \mathcal{R}(x)\]
\begin{remark}
    We know that $||x||_2<m ||x||_\infty$.
\end{remark}

Finally, by knowing the facts that $||(x,w)-(x^*,w^*) ||^2_2\leq \frac{\epsilon^\prime \mathrm{i}}{\kappa^\prime}$ and $w^*\in \mathcal{F}(x)$, we can simply deduce that there exists a vector $\zeta^\prime$ such that:
\[ \left(y-x\right) ^T w^* \geq  -\zeta^\prime ,  \quad \forall y \in \mathcal{R}(x) \]
Considering appropriate small numbers for $h$ and $\gamma^\prime$ (recall that $\epsilon^\prime=\frac{\epsilon h}{\mathrm{i}}$ and consider $\epsilon=\beta$)  will imply the following for $x\in \mathcal{R}(x)$: 
\[ \left(y-x\right) ^T w^* \geq - \beta \mathds{1}  ,  \quad \forall y \in \mathcal{R}(x) \]

\end{proof}

\subsection{More Generalized Form of the Robust Berge Maximum Theorem}

 We can easily extend the robust form of Berg's maximum theorem to $f:A\times B \rightarrow \mathbb{R}^t$ for any $t$ by some minor changes in the following. \footnote{For MGQVI, we only need to prove the Lipchitzness of $g^*$.}. 
 
Assuming $f(x,y)=(f_1(x, y),\dots,f_t(x,y))$ then if  $f^*(a)=(f_1(a, b),\dots,f_t(a,b))$ we can conclude that $(f_1(a, b),\dots,f_t(a,b))\geq (f_1(a, c),\dots,f_t(a,c)) $ for all $c\in g(a)$. In other words for all $1 \leq i\leq t$, $f_i(a, b) \geq f_i(a, c)$.  A similar analogy also holds for $g^*$.

\begin{theorem}\label{berge2}
  (Generalized Robust Berge Maximum Theorem). Let $A \subseteq \mathbb{R}^{n_1}\times \mathbb{R}^{n_2}$ and $B \subseteq \mathbb{R}^{m}$. Consider a continuous function $f: A \times B \rightarrow \mathbb{R}^t$ that is $\mu$-strongly concave $\forall a \in A, L$-Lipschitz in $A \times B$ and a $L^{\prime}$-Hausdorff Lipschitz, non-empty, convex-set, compact-valued correspondence $g: A \rightrightarrows B$. Define $f^*(a)=\max _{b \in g(a)} f(a, b)$ and $g^*(a)=\arg \max _{b \in g(a)} f(a, b)$. Then, we observe $f^*$ is continuous and $g^*$ is upper semi-continuous and single-valued, i.e., continuous. Furthermore,  $g^*$ is Lipschitz and $\left(L^{\prime}+2 \sqrt{\frac{4}{\mu}} \sqrt{\left(L+L \cdot L^{\prime}\right)}\right)-(1 / 2,p)$-Holder continuous respectively (for sufficiently small differences)\footnote{We only need the cases where $n_1=m$, $n_2=r\cdot  m$ and $t=r$.}.
\end{theorem}

\begin{proof}
     We will follow the proof of Theorem 3.20 in \cite{concavegames} leveraging the generalized Apollonius theorem. 
     We first prove that $g^*$ is a singled-value function. For the sake of contradiction, assume that we have $b_1, b_2 \in g^*(a)$. By the definition of $f^*(a)$, we have the following:
   \[
   f^*(a)=\max \{f(a, b): b \in g(a)\}=f\left(a, b_1\right)=f\left(a, b_2\right)
   \]
      Let $b_k=b_1 \cdot k+(1-k) \cdot b_2$ for some $k \in(0,1)$. By the fact that $g(a)$ is a convex set correspondence, it holds that also $b_k \in g(a)$. Then by concavity of $f$ we have that:
\[    \begin{aligned}
        f^*(a) \geq f\left(a, b_k\right)=f\left(a, b_1 \cdot t+(1-t)-b_2\right) &\geq f\left(a, b_1\right) \cdot k+(1-k) \cdot f\left(a, b_2\right) \\ &=f^*(a) \cdot k+(1-k) \cdot f^*(a)=f^*(a)
    \end{aligned}\]
So, by definition $b_k \in g^*(a)$. However, since $f$ is strongly concave, it has a unique maximizer. This means that $g^*(a)$ is a singled-valued correspondence. By  Berge's oringinal theorem, we know that $g^*$ would be upper semi-continuous which for the singled-value case corresponds to the classical notion of continuity.

Next, we will prove now that $f^*(a), g^*(a)$ will have some form of Lipschitz continuity. More precisely, for two arbitrary inputs $a_1, a_2 \in A$ it holds that:
$$
\left\{\begin{array}{l}
f^*\left(a_1\right)=\left\{\max f\left(a_1, b\right): b \in g\left(a_1\right)\right\}=f\left(g^*\left(a_1\right), a_1\right) \\
f^*\left(a_2\right)=\left\{\max f\left(a_2, b\right): b \in g\left(a_2\right)\right\}=f\left(g^*\left(a_2\right), a_2\right)
\end{array}\right.
$$
Furthermore, we have the following:
\[
\begin{aligned}
f^*\left(a_1\right)=\left\{\max f\left(a_1, b\right): b \in g\left(a_1\right)\right\} & \geq f\left(a_1, b\right) \thickspace \forall b \in g\left(a_1\right) \\
f^*\left(a_1\right) & \geq f\left(a_1,\Pi_{g\left(a_1\right)}\left(  g^*\left(a_2\right)\right)\right) \\
& \geq-L\left\|\left(\Pi_{g\left(a_1\right)}\left(g^*\left(a_2\right)\right), a_1\right)-\left(g^*\left(a_2\right), a_2\right)\right\|\mathds{1}\\&+f^*\left(a_2\right)
\end{aligned}
\]
Note that $\Pi_{g\left(a_1\right)}\left(g^*\left(a_2\right)\right)$ denotes the exact projection of $g^*(a_2)$ to $g(a_1)$. The second inequality is implied by the fact that $f^*(a_1)=f\left(a_1,g^*\left(a_1\right)\right) \geq f(a_1,b)$ for all $b\in g(a_1)$. Finally, the last inequality follows by Lipchitzness of $f$ in $A\times B$. Next:
\[
\begin{aligned}
& f^*\left(a_2\right)-f^*\left(a_1\right) \leq L\left\|a_1-a_2\right\| \mathds{1} +L\left\|\Pi_{g\left(a_1\right)}\left(g^*\left(a_2\right)\right)-g^*\left(a_2\right)\right\|\mathds{1} \\
& f^*\left(a_2\right)-f^*\left(a_1\right) \leq L\left\|a_1-a_2\right\| \mathds{1}+L \mathrm{~d}_H\left(g\left(a_1\right), g\left(a_2\right)\right) \mathds{1} \\
& f^*\left(a_2\right)-f^*\left(a_1\right) \leq\left(L+L \cdot L^{\prime}\right)\left\|a_1-a_2\right\| 
\mathds{1}
\end{aligned}
\]
Note that the second inequality again follows by the definition of Hausdorff distance and projection. Applying symmetrically the same argument for $f^*\left(a_2\right)$:
\[
\left|f^*\left(a_1\right)-f^*\left(a_2\right)\right| \leq\left(L+L \cdot L^{\prime}\right)\left\|a_1-a_2\right\|\mathds{1}
\]
Or equivalently:
\[
\left\|f^*\left(a_1\right)-f^*\left(a_2\right)\right\| \leq \sqrt{t} \cdot \left(L+L \cdot L^{\prime}\right)\left\|a_1-a_2\right\|
\]

For $g^*(a)$, in the first step, we will leverage the generalization of the Apollonius theorem\footnote{We assume that $f=(f_1,\dots,f_m)$ and derive the bounds for each $f_i$ ( for $i\in[m]$) separately similar to \cite{concavegames}.}.  More precisely, for a $\mu$-strongly concave function $f$, we have the following relationship:

\[-f\left(a, \frac{x+y}{2}\right) \leq-\frac{f(a, x)+f(a, y)}{2}-\frac{\mu}{8}\|x-y\|_2^2\mathds{1} \quad \forall a \in A\]
or equivalently,
$$
\frac{f(a, x)+f(a, y)}{2} \leq f\left(a, \frac{x+y}{2}\right)-\frac{\mu}{8}\|x-y\|_2^2 \mathds{1} \quad  \forall a \in A
$$
Since $f\left(a, \frac{x+y}{2}\right) \leq \max \{f(a, x), f(a, y)\}$ we get that:
\begin{equation}
    \label{xxx3}\|x-y\| \mathds{1} \leq \sqrt{\frac{8}{\mu} \cdot \frac{\max \{f(a, x), f(a, y)\}-\min \{f(a, x), f(a, y)\}}{2}} \quad \forall a \in A
\end{equation}
Or equivalently: 
\begin{equation}
    \label{xxx4}\|x-y\| \sqrt{t} \leq \sqrt{\frac{4}{\mu}} \|\sqrt{  \max \{f(a, x), f(a, y)\}-\min \{f(a, x), f(a, y)\}}\| \quad \forall a \in A
\end{equation}
In the interest of clarity, we introduce extra variables similar to \cite{concavegames}. $K_1=g^*\left(a_1\right), K_2=g^*\left(a_2\right)$ and $K_3=\Pi_{g\left(a_1\right)}\left(K_2\right)$, we get the following:
$$
\begin{aligned}
 \left\|g^*\left(a_1\right)-g^*\left(a_2\right)\right\|&=\left\|K_1-K_2\right\| \leq\left\|K_1-K_3\right\|+\left\|K_2-K_3\right\| \\
& \leq \mathrm{d}_{\mathrm{H}}\left(g\left(a_1\right), g\left(a_2\right)\right)+\left\|\left(K_1-K_3\right)\right\|
\end{aligned}$$
The second inequality again follows by the definition of Hausdorff distance and projection. Now we are ready to apply Equation \ref{xxx3} knowing $f\left(a_1, K_1\right)>f\left(a_1, K_3\right)$.

\[
\begin{aligned}
&\left\|g^*\left(a_1\right)-g^*\left(a_2\right)\right\| \leq L^{\prime}\left\|a_1-a_2\right\|+\sqrt{\frac{4}{t \mu }}\| \sqrt{\left(f\left(a_1, K_1\right)-f\left(a_1, K_3\right)\right)} \|\\
& \leq L^{\prime}\left\|a_1-a_2\right\|+\sqrt{\frac{4}{t\mu}} \|\sqrt{\left|f\left(a_1, K_1\right)-f\left(a_1, K_3\right)+f\left(a_2, K_2\right)-f\left(a_2, K_2\right)\right|}\| \\
& \left.\leq L^{\prime}\left\|a_1-a_2\right\|+\sqrt{\frac{4}{t\mu}} \| \sqrt{\left\{\begin{array}{c}
\mid f\left(a_1, K_1\right)-f\left(a_2, K_2\right)| \\
+ \\
|f\left(a_2, K_2\right)-f\left(a_1, K_3\right)|
\end{array}\right\}}\right\} \|\\
& \leq L^{\prime}\left\|a_1-a_2\right\|+\sqrt{\frac{4}{t\mu}}\|\left\{\begin{array}{c}
\sqrt{\left|f^*\left(a_1\right)-f^*\left(a_2\right)\right|} \\
+ \\
\sqrt{\left|f\left(a_2, K_2\right)-f\left(a_1, K_3\right)\right|}
\end{array}\right\}\|
\end{aligned}
\]

Next, we apply the bound that we have for $f^*$:

$$
\begin{aligned}
 \left\|g^*\left(a_1\right)-g^*\left(a_2\right)\right\| &\leq L^{\prime}\left\|a_1-a_2\right\|+\sqrt{\frac{4}{t\mu}}\left\{\begin{array}{c}
\sqrt{\sqrt{t}\left(L+L \cdot L^{\prime}\right)\left\|a_1-a_2\right\|} \\
+ \\
\sqrt{L\left\|K_2-K_3\right\|+L \| a_2-a_1 \|}
\end{array}\right\} \\
& \leq L^{\prime}\left\|a_1-a_2\right\|+\sqrt{\frac{4}{t\mu}}\left\{\begin{array}{c}
\sqrt{\sqrt{t}\left(L+L \cdot L^{\prime}\right)\left\|a_1-a_2\right\|} \\
+ \\
\sqrt{L \mathrm{~d}_{\mathrm{H}}\left(g\left(a_1\right), g\left(a_2\right)\right)+L \| a_2-a_1 \|}
\end{array}\right\}\\
& \leq L^{\prime}\left\|a_1-a_2\right\|+2 \sqrt{\frac{4}{\sqrt{t}\mu}} \sqrt{\left(L+L \cdot L^{\prime}\right)\left\|a_1-a_2\right\|} \\
& \leq \underbrace{L^{\prime}+2 \sqrt{\frac{4}{\mu}} \sqrt{\left(L+L \cdot L^{\prime}\right)}}_\kappa \max \left\{\left\|a_1-a_2\right\|^{1 / 2},\left\|a_1-a_2\right\|\right\} \\
& \leq \kappa \max \left\{\left\|a_1-a_2\right\|^{1 / 2},\left\|a_1-a_2\right\|\right\} \\
&
\end{aligned}
$$

\end{proof}

\subsection{Proof of Inclusion in PPAD}
Finally, in the following lemma, we convert the problem of finding a $t$-resilient Nash equilibrium (for any constant $t$) to a series of variational equations as follows. Roughly speaking, the definition of $t$-resilient Nash equilibrium needs $\sum_{i=1}^{t} {k \choose i}$ conditions (immunity to deviations) to be checked, and through this lemma, we make sure that all of these conditions are satisfied by using a single variational inequality problem.

\begin{remark}
    Note that the membership of a strategy profile in $\mathcal{S}=\Pi_{j\in [k]} S_j$ could be done in polynomial time. Therefore, a strong separation oracle for $\mathcal{S}$ exists.
\end{remark}

\begin{lemma}\label{reslient2k}
  Suppose that we have a game with multi-concave and continuously differentiable utilities $\mathcal{U}=(u_1,\dots,u_k)$,  and strategies $\mathcal{S}=(S_1,\dots,S_k)$. The problem of finding an approximate $t$-resilient Nash equilibrium of this game can be reduced to an MVI (multi-variational inequality) problem. 
\end{lemma}
\begin{proof}
  For simplicity, we only solve the case $t=2$. Let $\theta_j=-u_j$ for each $j\in [k]$.  Finding an approximate $2$-resilient equilibrium in this game is equivalent to finding a solution to the following optimization problem where the goal is  finding $\textbf{s}^*$ such that:

\begin{equation}\label{mmingener2}
\begin{aligned}
 & \forall J\in \mathcal{J},  \forall j\in J,\quad \theta_j(\textbf{s}^*) \leq \operatorname{Min} \theta_j(\textbf{s}_J,\textbf{s}^*_{-J}) +\epsilon
 \\ &\mbox{s.t } s_j \in S_J
\end{aligned}
\end{equation}

  Define  $F_{i}(x)=(\nabla_{x_{\pi(i,j)}} \theta_j(x))_{j=1}^k \in \mathbb{R}^{k}$ where the function $\pi$ will help us indicate which one of the partial derivations should be considered. $F$ is the matrix whose columns are $F_{i}$ (for each $i$). And if $\pi(i,j)=0$, $F(i,j)=0$ will be $0$.

\[\pi(i,j)=
\left\{
	\begin{array}{lll}
		i  & \mbox{if }  1 \leq i\leq k \thickspace \& \thickspace i=j  \\
		b(i,j) & \mbox{if }  k+1 \leq i\leq k+1+ {k \choose 2} \thickspace  \& \thickspace b(i,j)\neq 0 \thickspace \\
       0 &  \mbox{o.w}
	\end{array}
\right.
\]

\[b(i,j)=
\left\{
	\begin{array}{ll}
		p=Mod((i-k)+1,k)  & \mbox{if }  i=j  \thickspace  \\
		0  & \mbox{o.w } 
	\end{array}
\right.
\]
We will show that $\textbf{s}^*$ is a $2$-resilient Nash equilibrium if and only if $\textbf{s}^* \in \mathcal{S}=\Pi_{i\in[k]} S_i$ satisfies:
\begin{equation}\label{mgogo}
 \left(\textbf{y}-\textbf{s}^*\right) ^T F\left(\textbf{s}^*\right)+\epsilon\mathds{1} \geq 0, \quad \forall \textbf{y}\in \mathcal{S}  
\end{equation}

Assume that we have an approximate solution for Equation \ref{mgogo}. First, for each $J\in \mathcal{J}$ with one element (all $j\in [k])$, we have:

$$
\theta_j(y_j,\textbf{s}^*_{-j}) \geq \theta_j\left(s_j,\textbf{s}^*_{-j}\right)+  \left(y_j-s^*_j\right) \nabla_{y^*_j} \theta_j \left(s_j,\textbf{s}^*_{-j} \right), \quad \forall y_j \in S_j
$$
Similar to Proposition \ref{geneqvari2}, $\left(y-\textbf{s}^*\right)^T  F\left(\textbf{s}^*\right)+\epsilon\mathds{1} \geq 0$, since $\textbf{s}^*$ is a solution to $\operatorname{MVI}(\nabla f, \mathcal{S})$. Then, we can conclude that $s^*$ is a $1$-resilient Nash equilibrium (which is also a Nash equilibrium): 

$$
\theta_j(y_j,\textbf{s}^*_{-j})+\epsilon \geq \theta_j(s_j,\textbf{s}^*_{-j})
$$

Next, for each coalition of size $2$, for a pair $(i,p)$ of players cannot favor themselves by forming a coalition and deviating from the equilibrium.  By the fact that $\theta_j(\cdot,s_{-J})$ is $2$-multi-convexity of $\theta_j$ for any $j\in [k]$, for any possible strategy $y_p \in S_p$ for player $p$, we have (note that either $j=p$ or $i=j$):

$$
\theta_j(y_p,y_i,\textbf{s}^*_{-\{i,p\}}) \geq \theta_j\left(\textbf{s}^*_{\{i,p\}},\textbf{s}^*_{-\{i,p\}}\right)+  \left(y_i-s^*_i\right) \nabla_{s^*_i} \theta_j \left(\textbf{s}^*_{\{i,p\}},\textbf{s}^*_{-\{i,p\}}\right), \quad \forall y_i \in S_i
$$And also for any possible fixed strategy $y_i$ for player $i$, we have:

$$
\theta_j(y_p,y_i,\textbf{s}^*_{-\{i,p\}}) \geq \theta_j\left(\textbf{s}^*_{\{i,p\}},\textbf{s}^*_{-\{i,p\}}\right)+  \left(y_p-s^*_p\right) \nabla_{s^*_p} \theta_j \left(\textbf{s}^*_{\{i,p\}},\textbf{s}^*_{-\{i,p\}}\right), \quad \forall y_p \in S_p
$$

Since $\textbf{s}^*$ is a solution to $\operatorname{MVI}(\nabla f, S)$, we have $\left(y-\textbf{s}^*\right)^T  F\left(\textbf{s}^*\right)+\epsilon\mathds{1} \geq 0$. This also means that the inequalities $\left(y_i-s^*_i\right) \nabla_{s^*_i} \theta_j \left(\textbf{s}^*_{\{i,p\}},\textbf{s}^*_{-\{i,p\}}\right)+\epsilon>0$ and  $\left(y_p-s^*_p\right) \nabla_{s^*_p} \theta_j \left(\textbf{s}^*_{\{i,p\}},\textbf{s}^*_{-\{i,p\}}\right)+\epsilon>0$ hold. Therefore, we can conclude that   $s^*$ is a $2$-resilient Nash equilibrium as any two pairs $(i,p)$ cannot form a coalition to get a better payoff:

\begin{equation}\label{malleqy}
\theta_j(\textbf{s}^*_J,\textbf{s}^*_{-J})  \leq \theta_j(\textbf{y}_J,\textbf{s}^*_{-J})+\epsilon, \quad \forall y_J \in S_J
\end{equation}
Now, assume that $x^*$ is a solution of Equation \ref{mgogo}.  We can use a similar approach used in Proposition \ref{geneqvari2} (corollary \ref{nashvi2}).
\end{proof}

\newpage

\section{Equivalent Results for Weak Separation Oracle Variation}\label{AppendixWeakoracles}
Weak separation oracles are useful for establishing general distinctions between problems and deriving lower bounds on computational complexity. They offer a broader perspective on the relative difficulty of problems. The choice between weak and strong separation oracles depends on the goals of the analysis and the level of detail required to draw meaningful conclusions about the computational landscape of specific problems.  Within this section, we reproduce the procedural steps elucidated in Section \ref{AppendixvariationalSection}, maintaining a consistent methodology\footnote{For $t$-resilient Nash, we have a strong separation oracle inherently.}. First, we define weak separation oracles.
\subsection{Essential Elements}

In short, a weak separation oracle (WSO) for a set-valued map  is a circuit that takes as input the point and the accuracy of the separation oracle $\delta$ (compared to a strong separation oracle $\delta=0$) and produces either an almost-membership or a guarantee of almost-separation.

\problemStatement{ A Weak Separation Oracle (via a circuit $C_{\mathcal{R}(x)}$)}{
  Input={A vector $z \in \mathbb{Q}^m \cap\mathbb{R}^{m*}$ and a constant $\delta$ as inputs:
  },
  Output={$(a, b) \in \mathbb{Q}^m \times \mathbb{Q}$ such that the threshold $b \in[0,1] \cap \mathbb{Q}$ denotes almost membership of $z$ in $\mathcal{R}(x)$. More precisely:
  \begin{itemize}
      \item If $z \in \overline{\mathrm{B}}( \mathcal{R}(x),\delta)$ then $b>\frac{1}{2}$ and the vector $a \in \mathbb{Q}^m$ will be $\bot$. In other words, $a$ is meaningful only when $b \leq \frac{1}{2}$
\item  $b \leq \frac{1}{2}$ and vector $a$ , with $\|a\|_{\infty}=1$ which defines an almost separating hyperplane $\mathcal{H}(a, z):=\left\{y \in\mathbb{R}^{m*}:\langle a, y-z\rangle=0\right\}$ between the vector $z$ and the set $\mathcal{R}(x)$ such that $\langle a, y-z\rangle \leq 0$ for every $y \in \overline{\mathrm{B}}(\mathcal{R}(x),-\delta)$.
  \end{itemize}
  }
}

Similarly, we need the weak version of the constrained convex optimization problem.

\problemStatement{Weak Constrained Convex Optimization Problem}{
  Input={A zeroth and first order oracle for the convex function $F: \mathbb{R}^m \rightarrow \mathbb{R}$, two rational numbers $\delta,\epsilon>0$ and a weak separation oracle $\mathrm{WSO}_{\mathcal{R}}$ for a non-empty closed convex set $\mathcal{R} \subseteq$  $\mathbb{R}^{m*}$.
  },
  Output={One of the following cases:
  \begin{itemize}
      \item (Violation of non-emptiness):
A failure symbol $\perp$ with a polynomial-sized witness that certifies that $\overline{\mathrm{B}}(\mathcal{R},-\delta)=\emptyset$.
\item (Approximate maximization):
A vector $z \in \mathbb{Q}^m \cap \overline{\mathrm{B}}(\mathcal{R},\delta)$, such that $F(z) +\epsilon \geq \max _{y \in \overline{\mathrm{B}}(\mathcal{R},-\delta)} F(y)$.
  \end{itemize}
  }
}

The weak approximate version of the projection problem, denoted by $\widehat{\Pi}^{\epsilon,\delta}_{X}(x)$, is an instance of a weak constrained optimization problem. The problem is defined as follows:
\problemStatement{Weak Approximate Projection Problem}{
  Input={Two rational numbers $\epsilon,\delta>0$ and a weak separation oracle $\mathrm{WSO}_{X}$ for a non-empty closed convex set $X \subseteq \mathbb{R}^{m*}$  and a vector $x$ that belongs to $\mathbb{Q}^m \cap X$. 
  },
  Output={One of the following cases:
  \begin{itemize}
      \item (Violation of non-emptiness):
A failure symbol $\perp$ followed by a polynomial-sized witness that certifies that $\overline{\mathrm{B}}(X,-\delta)=\emptyset$.
\item  (Approximate projection):
A vector $z \in \mathbb{Q}^m \cap \overline{\mathrm{B}}(X,-\delta)$, such that:
$$
\|z-x\|_2^2 \leq \min _{y \in \overline{\mathrm{B}}(X,\delta)}\|x-y\|_2^2+\epsilon.
$$
  \end{itemize}
  }
}

For the case of weak separation oracles, a disparity issue arises (see Remark B.4 in \cite{concavegames}) and to alleviate this problem, the authors assumed that the $L$-Hausdorff-set-valued maps to be $(\eta, \sqrt{m}, L)$-well conditioned, i.e., $\forall x \in \mathbb{R}^{m*}$, there exists $a \in \mathcal{R}(x)$ such that $ \overline{\mathrm{B}}(a, \eta) \subseteq \mathcal{R}(x)$. Considering this assumption and leveraging the ellipsoid method will assure the existence of an oracle-polynomial time algorithm (Algorithm 1 in \cite{concavegames}) that has a polynomial non-emptiness certificate using the following modification on the weak constrained convex optimization.  We also can define $(r,R)$-\emph{well-boundedness} of convex set $X$ by having the property that $\exists \boldsymbol{a}_0 \in \mathbb{R}^m: \overline{\mathrm{B}}\left(\boldsymbol{a}_0, r\right) \subseteq X \subseteq \overline{\mathrm{B}}(0, R)$.  Recall that $\mathbb{R}^{m*}$ can represent any well-bounded subset of  $\mathbb{R}^{m}$. We define the notation \emph{$\operatorname{vol}(A)$} representing the Lebesgue volume measure of the set $A$. Assume that all of our sets have one element that we call $0$.

\problemStatement{Weak Convex (Feasibility/Projection/Optimization) Problem}{
  Input={A zeroth and first order oracle for the convex function $F: \mathbb{R}^{m*} \rightarrow \mathbb{R}$, two rational numbers $\delta,\epsilon>0$ and a weak separation oracle $\mathrm{WSO}_{\mathcal{R}}$ for a non-empty closed convex set $\mathcal{R} \subseteq$  $\mathbb{R}^{m*}$.
  },
  Output={One of the following cases:
  \begin{itemize}
      \item (Violation of non-emptiness):
A failure symbol $\perp$ with a polynomial-sized witness that certifies that either $\mathcal{R}=\emptyset$ or $\operatorname{vol}(\mathcal{R})\leq \operatorname{vol} (\overline{\mathrm{B}}(0,\eta))$.
\item (Approximate minimization):
A vector $z \in \mathbb{Q}^m \cap \overline{\mathrm{B}}(\mathcal{R},\delta)$, such that $F(z)  \leq \min _{y \in \overline{\mathrm{B}} (\mathcal{R},-\delta)} F(y)+ \epsilon$.
  \end{itemize}
  }
}

\begin{definition}
     $\left((\epsilon, \eta \right)$-Approximate Generalized Equilibrium). Let $(\mathcal{U},\mathcal{S},\mathcal{R})$ be a concave game with $k$ players where $\mathcal{R}$ represents the common constraint set (that includes all strategies that are acceptable given this constraint). Define $\mathcal{R}_\eta=$ $\overline{\mathrm{B}}(\mathcal{R}, \eta)$, and $\mathcal{R}_{-\eta}=\overline{\mathrm{B}}(\mathcal{R},-\eta)$. A vector $\textbf{x}^{\star} \in S_\eta$ is an $(\epsilon, \eta)$- approximate generalized equilibrium of this game if for every $i \in[k]$ and every $x_i \in S_i $ such that $\left(x_i, x_{-i}^{\star}\right) \in \mathcal{R}_{-\eta}$ it holds that:
\begin{equation*}\label{eq:bestrespeta}
    u_i\left(\textbf{x}^{\star}\right) +\epsilon  \geq u_i\left(x_i, \textbf{x}_{-i}^{\star}\right)
\tag{$\ddagger$}
\end{equation*}

\end{definition}

Next, we establish an equilibrium focusing on weak separation oracles. Subsequently, we need to introduce the computational problem of approximating a generalized equilibrium similar to \cite{concavegames}.
\begin{theorem}[\cite{concavegames}]
    The  Kakutani fixed-point problem  for a weak separation oracle is PPAD-complete.
\end{theorem}
\problemStatement{Kakutani With A Weak Separation Oracle (via $C_{\mathcal{R}(x)}$)}{
  Input={A circuit $C_\mathcal{R}$ that represents weak separation oracle for an $(\eta,\sqrt{m},L)$ well-conditioned correspondence $\mathcal{R}:\mathbb{R}^{m*} \rightrightarrows\mathbb{R}^{m*}$ and an accuracy parameter $\alpha$. 
  },
  Output={One of the following cases:
  \begin{itemize}
      \item (Violation of $\eta$-non-emptiness): A vector $x \in$ $\mathbb{R}^{m*}$ such that $\operatorname{vol}(\mathcal{R}(x))\leq \operatorname{vol} (\overline{\mathrm{B}}(0,\eta))$,
\item  (Violation of L-Hausdorff Lipschitzness):
Four vectors $p, q, z, w \in\mathbb{R}^{m*}$ and a constant $\epsilon>0$ such that $w=\widehat{\Pi}_{\mathcal{R}(q)}^{\epsilon,\epsilon}(q)$ and $z=\widehat{\Pi}_{\mathcal{R}(p)}^{\epsilon,\epsilon}(w)$ but $\|z-w\|>L\|p-q\|+3(1+c_{\eta,m})\epsilon$\footnote{$c_{\eta,m}$ is a constant based on the dimension $m$ and $\eta$. For more information, see \cite{concavegames}.},
\item  Vectors $x, z \in\mathbb{R}^{m*}$ such that $\|x-z\| \leq$ $\alpha$ and $z \in \mathcal{R}(x) \Leftrightarrow d(x, \mathcal{R}(x)) \leq \alpha$.
  \end{itemize}
  }
}

\problemStatement{Strongly Concave Games Problem with WSO}{
  Input={We receive as input all the following:
  \begin{itemize}
      \item $k$ circuits representing the utility functions $\left(u\right)_{i=1}^k$ for all $k$ players,
      \item A Lipschitzness parameter $L$, a strong concavity parameter $\mu$, and accuracy parameter $\epsilon$, 
      \item $S=\Pi_{i=1}^k S_i$ a convex set called \emph{the strategy domain} where $S_i$ represent the strategy domain for each player $i$,
      \item An arithmetic circuit representing a weak separation oracle for the well-bounded set  $\mathcal{R}$ (strategies that satisfy the common constraint) that is a non-empty, convex, and compact subset of $S$,
      \item Accuracy parameters $\epsilon,\delta$.
  \end{itemize}
  },
  Question={We output one of the following:
  \begin{itemize}
   \item (Violation of almost non-emptiness)
A certificate that $\operatorname{vol}(\mathcal{R})\leq \operatorname{vol} (\overline{\mathrm{B}}(0,\eta))$.
 \item (Violation of Lipschitz Continuity)
A certification that there exist at least two vectors $x, y \in S $ and an index $i \in[n]$ such that $\left|u_i(x)-u_i(y)\right|>L \|x-y\|$.
\item  (Violation of Strong Concavity)
An index $i \in[n]$, three vectors $x_i, y_i \in S_i,~\textbf{x}_{-i} \in S_{-i}=\Pi_{j=1,j\neq i}^k S_j$ and a number $\mu \in[0,1]$ such that:
\[\begin{aligned}
     u_i\left(\lambda  x_i+(1-\lambda)  y_i, \textbf{x}_{-i}\right)
& <\lambda  u_i\left(x_i, \textbf{x}_{-i}\right)+(1-\lambda) u_i\left(y_i, \textbf{x}_{-i}\right) \\
 &+\frac{\lambda(1-\lambda)}{2} \mu \cdot\left\|\left(x_i, \textbf{x}_{-i}\right)-\left(y_i, \textbf{x}_{-i}\right)\right\|_2^2
\end{aligned}\]
\item An $(\epsilon,\delta)$-approximate generalized equilibrium  in the sense of {\rm (\ref{eq:bestrespeta})}.
  \end{itemize}
  }
}
\begin{theorem}[\cite{concavegames}]\label{concaveppad2}
    The computational problem Strongly Concave Games for a weak separation oracle is PPAD-complete.
\end{theorem}

\begin{remark}
    
When $\eta=0$, we have  $\epsilon$-approximate generalized equilibrium.

\end{remark}

\subsection{Variational Inequalities with Weak Separation Oracles}

Next, we investigate the variational problems for weak separation oracles.

\problemStatement{$GQVI(\mathcal{F},\mathcal{R})$ With weak Separation Oracles}{
  Input={ We receive as input all the following:
  \begin{itemize}
            \item  A circuit $C_{\mathcal{R}}$ which represents a weak separation oracle for a $(\eta,\sqrt{m},L_{\mathcal{R}})$ well-conditioned correspondence $\mathcal{R}: \mathbb{R}^{m*} \rightrightarrows  \mathbb{R}^{m*}$,
      \item A circuit $C_\mathcal{F}$ which represents  a $(\eta,\sqrt{m},L_{\mathcal{F}})$ well-conditioned correspondence $\mathcal{F}: \mathbb{R}^{m*} \rightrightarrows  \mathbb{R}^{m*}$ which is also $\gamma$-strongly convex valued, 
      \item  Accuracy parameters $\beta,\eta$. 
  \end{itemize}},
  Output={One of the following cases:
  \begin{itemize}
      \item (Violation of almost non-emptiness):
 A vector $x\in \mathbb{R}^{m*}$ with a certificate that either $\operatorname{vol}(\mathcal{R}(x))\leq \operatorname{vol} (\overline{\mathrm{B}}(0,\eta))$ or  $\operatorname{vol}(\mathcal{F}(x))\leq \operatorname{vol} (\overline{\mathrm{B}}(0,\eta))$.    
\item  (Violation of $L_{\mathcal{R}}$-Hausdorff Lipschitzness of $\mathcal{R}$)
Four vectors $p, q, z, w \in \mathbb{R}^{m*}$ and a constant $\epsilon>0$ such that $w=\widehat{\Pi}_{\mathcal{R}(q)}^{\epsilon,\epsilon}(q)$ and $z=\widehat{\Pi}_{\mathcal{R}(p)}^{\epsilon,\epsilon}(w)$ but $\|z-w\|>L_{\mathcal{R}}\|p-q\|+3(1+c_{\eta,m})\epsilon$,
\item  (Violation of $L_\mathcal{F}$-Hausdorff Lipschitzness of $\mathcal{F}$):
Four vectors $p, q, z, w \in \mathbb{R}^{m*}$ and a constant $\epsilon>0$ such that $w=\widehat{\Pi}_{\mathcal{F}(q)}^{\epsilon,\epsilon}(q)$ and $z=\widehat{\Pi}_{\mathcal{F}(p)}^{\epsilon,\epsilon}(w)$ but $\|z-w\|>L_\mathcal{F}\|p-q\|+3(1+c_{\eta,m})\epsilon$,
\item 
(Violation of strong convexity):
Two vectors $x,p,q \in \mathbb{R}^{m*}$ and two constants $\epsilon>0$ and $\lambda\in (0,1)$ such that:
$$
\begin{aligned}
\widehat{\Pi}_{\mathcal{F}(x)}^{\epsilon,\epsilon}(\lambda p+ (1-\lambda)q ) & > \lambda \cdot \widehat{\Pi}_{\mathcal{F}(x)}^{\epsilon,\epsilon}(p)+(1-\lambda) \cdot \widehat{\Pi}_{\mathcal{F}(x)}^{\epsilon,\epsilon}(q)
\\&+\left(-\frac{\lambda(1-\lambda)}{2} \cdot \gamma\cdot\left\|p-q\right\|_2^2+c_{\eta,m} \cdot \epsilon\right)
\end{aligned}
$$
\item Two tuples $(x,w)$ and $(x^*,w^*)$ with  $\|(x,w)-(x^*,w^*)\|\leq \beta $ such that $x^*\in \overline{\mathrm{B}}( \mathcal{R}(x),\eta)$ and  $w^*\in \overline{\mathrm{B}}(\mathcal{F}(x),\eta)$ such that  $\left(y-x\right) ^T w^* +\beta  \geq 0, \quad \forall y \in \overline{\mathrm{B}}( \mathcal{R}(x),-\eta)$   
  \end{itemize}
  }
}

Using the following proposition, we show that the computational version of GQVI for weak separation oracles is in PPAD.
\begin{theorem}\label{gqviinppad2}
The generalized quasi-variational inequality problem $GQVI(\mathcal{F},\mathcal{R})$ for weak separation oracles $\mathrm{SO}_{\mathcal{R}}$ and $\mathrm{SO}_\mathcal{F}$ is in PPAD.
\end{theorem}

\begin{proof}  Without loss of generality, we can consider $[-1,1]^m$ instead of $\mathbb{R}^{m*}$.  Define  $ \mathcal{R}(x)=\{y\in \overline{\mathrm{B}}(\mathcal{R}(x),\delta)\}$ and $\mathcal{R}^\prime(x)=\{y\in \overline{\mathrm{B}}(\mathcal{R}(x),-\delta)\}$.  By definition, $\mathcal{R}(x)$ and $\mathcal{R}^\prime(x)$ have a weak separation oracle and convexity also follows. 
The proof is organized as follows. First, we define the following:
  \[\Phi(y,x,w)=-(y-x)^T w-\gamma(||y||^2_2)\]
  \[\Pi(x,w)=\{y\in \mathcal{R}(x) ~|~\Phi(y,x,w)>max_{y\in \mathcal{R}^\prime(x)}\Phi(y,x,w)-\epsilon \}\]
 Next, we show that for a constant $\kappa^\prime$, $\Psi(x,w)=(\Pi(x,w),\mathcal{F}(x))$ is $\kappa^\prime$-Lipschitz continuous in $(x,w)$.  Then, we construct a weak separation oracle for $\Psi$.  Finally, we show that an approximate Kakutani's fixed point of this function will provide an approximate solution to the given GQVI with WSO.


To show that $\Psi$ is $\kappa^\prime$-Lipschitz continuous, we  show that $d_H(\Psi(x_1,w_1),\Psi(x_2,w_2))\leq \kappa^\prime||(x_1,w_1)-(x_2,w_2)||^q_p+c$ for some constants $q$,$p$ and $c$. Similar to the strong separation oracle case,  $\Phi(y,x,w)$ is $(2\gamma)$-strongly concave function of $y$.   In addition, $\Phi$ is $G$-Lipschitz continuous where $G=m+2\gamma m$ for all $x$ and $w$. 
We also have access to the sub-gradients of $\Phi$. Now, we define the following sets:

\[
H_+(x,w)=\{y\in \mathcal{R}(x) ~|~\Phi(y,x,w)=max_{y\in \mathcal{R}(x)} \Phi(y,x,w)\}
\]
\[
H^{\epsilon}_+(x,w)=\{y\in \mathcal{R}(x) ~|~\Phi(y,x,w)\geq max_{y\in \mathcal{R}(x)} \Phi(y,x,w)-\epsilon\}
\]
\[
H_-(x,w)=\{y\in \mathcal{R}^\prime(x) ~|~\Phi(y,x,w)=max_{y\in \mathcal{R}^\prime(x)} \Phi(y,x,w)\}
\]
\[
H^{\epsilon}_-(x,w)=\{y\in \mathcal{R}^\prime(x) ~|~\Phi(y,x,w)\geq max_{y\in \mathcal{R}^\prime(x)} \Phi(y,x,w)-\epsilon\}
\]
We also need the following lemma  (Lemma E2 from \cite{concavegames}) that help us relate these two sets\footnote{Both $\mathcal{R}(x)$ and $\mathcal{R}^\prime(x)$ are convex sets. Recall that we could consider $\mathbb{R}^{m*}$ alternatively for well-boundedness.}.

 \begin{lemma}[\cite{concavegames}] \label{lemmae2}
     For a function $f: A \subseteq \mathbb{R}^{d} \rightarrow \mathbb{R}^{d}$, which is $\mu$-strongly concave and $G$-lipschitz and $a$ well-bounded convex set $S$, i.e. $\exists a_0 \in \mathbb{R}^d: \overline{\mathrm{B}}\left(a_0, r\right) \subseteq S \subseteq \overline{\mathrm{B}}(0, R) \subseteq A$, it holds that:
$$
\left|\max _{\boldsymbol{a} \in \overline{\mathrm{B}}(S,\eta_1)} f(\boldsymbol{a})-\max _{\boldsymbol{a} \in \overline{\mathrm{B}}(S,\eta_2)} f(\boldsymbol{a})\right| \leq C_{G, \mu}\left\|\eta_1-\eta_2\right\|
$$
for some constant $C_{G, \mu}$.
 \end{lemma}

By optimality KKT conditions for maximization of a concave function with respect to the constraint set $\mathcal{R}(x)$ for all $y\in \mathcal{R}(x) $ we have:
\[
\partial \Phi(y,x,w)^T(y^*-y)\geq 0 \text{ where } y^*=argmax_{y\in \mathcal{R}(x)}\Phi(y,x,w)
\]
In addition, $(2 \gamma)$-strong-concavity of $\Phi(\cdot,x,w)$ results the following inequality:
\[
\Phi\left(y^*,x,w\right)-\Phi(y,x,w) \geq \partial \Phi\left(y^*,x,w\right)^{\top}\left(y^*-y\right)+\gamma\left\|y-y^*\right\|_2^2 
\]
By combining the previous inequalities we have:
\[
\Phi\left(y^*,x,w\right)-\Phi(y,x,w) \geq \gamma\left\|y-y^*\right\|_2^2 
\]
In conclusion, for $y^*=argmax_{y\in \mathcal{R}(x)}\Phi(y,x,w)$ and $y \in H^{\epsilon}_+(x,w) $, we have the following inequality:
\begin{equation}\label{kapp55}
    \gamma\left\|y-y^*\right\|^2_2 \leq \epsilon \text{ or equivalently } \left\|y-y^*\right\|_2 \leq \sqrt{\frac{\epsilon}{\gamma}}
\end{equation}
Applying Lemma \ref{lemmae2}, we have that for every $y \in \Pi(x,w)$ it holds that there exists a constant $k$ such that:
\[
\Phi(y,x,w) \geq \max _{y \in \mathcal{R}^\prime(x)} \Phi(y,x, w)-\epsilon \geq \max _{y \in \mathcal{R}(x)} \Phi(y,x, w)-\epsilon -k 
\]
Therefore for every $y \in \Pi(x,w)$, it holds that $y \in H_{+}^{\epsilon+k}(x,w)$. Consequently,
$$
\left\{\begin{array}{l}
\mathrm{d}_{\mathrm{H}}\left(\Pi\left(x_1,w_1\right), H_{+}\left(x_1,w_1\right)\right) \leq \sqrt{\frac{\epsilon+k}{\gamma}} \\
\mathrm{d}_{\mathrm{H}}\left(\Pi\left(x_2,w_2\right), H_{+}\left(x_2,w_2\right)\right) \leq \sqrt{\frac{\epsilon+k }{\gamma}} \\
\mathrm{d}_{\mathrm{H}}\left(H_{+}\left(x_1\right), H_{+}\left(x_2\right)\right)=\mathrm{d}\left(H_{+}\left(x_1\right), H_{+}\left(x_2\right)\right) \leq \kappa\left\|x_1-x_2\right\|_2^{1 / 2}+c
\end{array}\right.
$$
The last inequality comes from the application of the Theorem \ref{berge} to $f((x,w),y)=H_+(y,x,w)$ and $g((w,x))=\mathcal{R}(x)$,  where $g$ is a non-empty, convex compact valued correspondence. This shows that $\Pi$ is approximate Hausdorff-$(\frac{1}{2},2)$-Holder continuous:

\begin{equation}\label{kapp45}
    d_H(\Pi(x_1,w_1),\Pi(x_2,w_2))\leq \kappa||(x_1,w_1)-(x_2,w_2)||^{\frac{1}{2}}_2+2\sqrt{\frac{\epsilon+k}{\gamma}}+c
\end{equation}
Since $\mathcal{F}$ is $L_\mathcal{F}$-Hausdorff Lipschitz, by the definition of $\Psi$ we can  there exists a constant $\kappa^\prime$: 
\begin{equation}\label{kapp51}
    d_H(\Psi(x_1,w_1),\Psi(x_2,w_2))\leq \kappa^\prime||(x_1,w_1)-(x_2,w_2)||^{\frac{1}{2}}_2+c^\prime
\end{equation}
This shows that $\Psi$ is  $\kappa^\prime$-Hausdorff Lipschitz (and also $(\frac{1}{2},2)$-Holder) continuous. 

 To establish a reduction employing the computational Kakutani's Problem variant, it is essential to establish the existence of a bounded-radius ball within the correspondence $\Psi(x,w)$. Let $y_{(x,-\eta)}^*=argmax_{y\in \mathcal{R}^\prime(x)}\Phi(y,x,w)$. By almost non-emptiness, $\operatorname{vol}(\mathcal{R}(x))> \operatorname{vol} (\overline{\mathrm{B}}(0,\eta))$ and  $\operatorname{vol}(\mathcal{F}(x))> \operatorname{vol} (\overline{\mathrm{B}}(0,\eta))$. 

\begin{enumerate}
    \item $\exists \hat{v}:\|\hat{v}\|=1 \thickspace \& \thickspace \mathcal{V}_x:=\overline{\mathrm{B}}\left(y_{(x,-\eta)}^{\star}-\frac{\eta}{2} \hat{v}, \min \{\eta / 2, \epsilon / G\}\right) \subseteq \mathcal{R}^\prime(x)=\overline{\mathrm{B}}(\mathcal{R}(x),-\eta)$
    \item By Lipschitzness of $\phi(\cdot,x,w)$, $y \in \mathcal{V}_x$:
    
    $$\left|\phi(x, \boldsymbol{y})-\phi\left(x, \boldsymbol{y}_{(\boldsymbol{x},-\eta)}^{\star}\right)\right| \leq G \min \{\eta / 2, \epsilon / G\}$$

\end{enumerate}

From (2) we can  $\Phi( \mathcal{V}_x,x,w) \subseteq\left[\Phi\left( y_{(x,-\eta)}^{\star},x,w\right)-\epsilon, \Phi\left(y_{(x,-\eta)}^{\star},x,w\right)\right] \stackrel{(1)}{\Rightarrow} \mathcal{V}_x \subseteq \Pi(x,w)$. Therefore, $\Pi(x,w)$ always contains a ball of radius $\min \{\eta / 2, \epsilon / G\}$. Similar to the strong separation oracle case, $\Psi$ is non-empty as well (we know $\operatorname{vol}(\mathcal{F}(x))> \operatorname{vol} (\overline{\mathrm{B}}(0,\eta))$ and $\mathcal{F}$ is Lipchitz).


Now we are ready to construct a weak separation oracle for $\Psi(x,w)$ leveraging a slightly modified version of the Weak Convex (Feasibility/Projection/Optimization) Problem stated in the following.

\problemStatement{Modified Weak Convex (Feasibility/Projection/Optimization) Problem}{
  Input={A zeroth and first-order oracle for the concave function $G: \mathbb{R}^{m*}\rightarrow \mathbb{R}$, two rational numbers $\delta,\epsilon>0$ and a weak separation oracle $\mathrm{SO}_{\mathcal{R}}$ for a non-empty closed convex-valued correspondence $\mathcal{R}:\mathbb{R}^{m*}\rightrightarrows  \mathbb{R}^{m*}$ and one input $x$.
  },
  Output={One of the following cases:
  \begin{itemize}
      \item (Violation of non-emptiness):
A failure symbol $\perp$ with a polynomial-sized witness that certifies that either $\mathcal{R}(x)=\emptyset$ or $\operatorname{vol}(\mathcal{R}(x))\leq \operatorname{vol} (\overline{\mathrm{B}}(0,\eta))$.
\item (Approximate maximization):
A vector $z \in \mathbb{Q}^m \cap \overline{\mathrm{B}}(\mathcal{R}(x),\delta)$, such that $G(z) +\epsilon \geq \max _{y \in \overline{\mathrm{B}}(\mathcal{R}(x),-\delta)} G(y)$.
  \end{itemize}
  }
}
Let us define $ \mathcal{R}^{\prime\prime}(x)=\{y\in \overline{\mathrm{B}}(\mathcal{R}(x),min\{\delta,\epsilon\})\} \subseteq \mathcal{R}^{\prime}(x)$. We can compute a solution $y^*\in \mathcal{R}^{\prime\prime}(x)$ such that  $\Phi\left(y^*,x,w\right) \geq$ $\max _{y \in \mathcal{R}^{\prime\prime}(x)} \Phi(y,x,w)-\min \{\epsilon, \eta\} \geq \max _{y \in \mathcal{R}^{\prime}(x)} \phi(x, y)-\epsilon $ using sub-gradient ellipsoid central cut method (see Appendix \ref{AppendixCentralcutellipsoidsection}). And as discussed, such a separation oracle exists. Thus, we can substitute a WSO for $\Psi$ by considering a separation oracle for the following set:

\[
\overline{\Psi}_{s}(x,w)=\{(y,w)\in \left(\mathcal{R}(x),\overline{\mathrm{B}}\left(\mathcal{F}(x),\epsilon\right)\right)~|~- \Phi(y,x,w)\leq \gamma^\prime\}
\]
where $\gamma^\prime= -\Phi(y^*,x,w)$. In other words, are looking for $(y,w)\in (\mathcal{R}(x),\overline{\mathrm{B}}\left(\mathcal{F}(x),\epsilon\right))$ given the strong separation oracles representing $\mathcal{F}$ and $D$ such that:

\[(y-x)^T w +\gamma||y||^2_2  \geq (y^*-x)^T w+ \gamma||y^*||^2_2 \]

Now, we can give $\Psi$ as input to the computational Kakutani problem with accuracy parameter $\alpha=\frac{\epsilon^\prime}{\kappa^\prime}$ where $\epsilon^\prime=\epsilon h$. The output of this Kakutani instance will be two points $(x,w)\in ([-1,1]^m,[-1,1]^m)$ and $z=(x^*,w^*)\in\Psi(x,w)$  where $||(x,w)-(x^*,w^*) ||\leq \frac{\epsilon^\prime}{\kappa^\prime}$ and $d((x,w),\Psi(x,w))\leq \frac{\epsilon^\prime}{\kappa^\prime}$. Thus, $w^*\in \mathcal{F}(x)$ and $d(w,\mathcal{F}(x))\leq \frac{\epsilon^\prime}{\kappa^\prime} $. In addition, $x^*\in \Pi(x,w)$ and $d(x,\Pi(x,w))\leq \frac{\epsilon^\prime}{\kappa^\prime} $.
By the definition of $\Pi$,  for every $y\in \mathcal{R}^\prime(x)$,  $\Phi\left(x^*,x,w\right)\geq \Phi\left(y,x,w\right)$. In conclusion:
\[ \left(y-x\right) ^T w+\gamma||y||^2_2  \geq \left(x^*-x\right) ^T w+ \gamma||x^*||^2_2 ,  \quad \forall y \in \mathcal{R}^\prime(x) \]
 Recall that $x,y, w\in [-1,1]^m $, consider $\gamma$ to be a very small number, and there exists a small constant $u$:
\[ \left(y-x\right) ^T w \geq  \pm \frac{\epsilon^\prime}{\kappa^\prime} \pm u ,  \quad \forall y \in \mathcal{R}^\prime(x)\]
Finally, by knowing the facts that $||(x,w)-(x^*,w^*) ||^2_2\leq \frac{\epsilon^\prime}{\kappa^\prime}$ and $w^*\in \mathcal{F}(x)$, Lipchitzness of $\mathcal{F}$, we can simply deduct  that there exists a constant $\beta^\prime$ such that:
\[ \left(y-x\right) ^T w^* \geq  -\beta^\prime ,  \quad \forall y \in \mathcal{R}^\prime(x) \]
 Considering appropriate numbers for $h$ (recall that $\epsilon^\prime=\epsilon h$) and $\gamma$ will imply the following (let $\epsilon=\beta$): 
\[ \left(y-x\right) ^T w^* \geq - \beta ,  \quad \forall y \in \overline{\mathrm{B}}(\mathcal{R}(x),-\eta) \]

\end{proof}

\newpage
\subsection{Especial Cases: GVI and VI}

Next, similar to the main part of the paper, we investigate QVI and VI for weak separation oracles.

\problemStatement{$QVI(F,\mathcal{R})$ With a weak Separation Oracle}{
  Input={We receive as input all the following:
  \begin{itemize}
      \item  A circuit $C_{\mathcal{R}}$ which represents a weak separation oracle for a $(\eta,\sqrt{m},L_{\mathcal{R}})$ well-conditioned correspondence $\mathcal{R}: \mathbb{R}^{m*} \rightrightarrows  \mathbb{R}^{m*}$.
      \item A circuit $C_F$ which represents a $L_F$-Hausdorff Lipschitz function $F: \mathbb{R}^{m*} \rightarrow  \mathbb{R}^{m*}$  
      \item   Accuracy parameters $\beta,\eta$.
  \end{itemize}},
  Output={One of the following cases:
  \begin{itemize}
      \item (Violation of almost non-emptiness):  A certificate that  $\operatorname{vol}(\mathcal{R}(x))\leq \operatorname{vol} (\overline{\mathrm{B}}(0,\eta))$.  
\item  (Violation of $L_{\mathcal{R}}$-Hausdorff Lipschitzness of $\mathcal{R}$):
Four vectors $p, q, z, w \in \mathbb{R}^{m*}$ and a constant $\epsilon>0$ such that $w=\widehat{\Pi}_{\mathcal{R}(q)}^{\epsilon,\epsilon}(q)$ and $z=\widehat{\Pi}_{\mathcal{R}(p)}^{\epsilon,\epsilon}(w)$ but $\|z-w\|>L_{\mathcal{R}}\|p-q\|+3(1+c_{\eta,m})\epsilon$.
\item  (Violation of $L_F$-Hausdorff Lipschitzness of $F$):
Two vectors $p, q \in \mathbb{R}^{m*}$ such that $\|F(p)-F(q)\|>L_F\|p-q\|$.
\item Two vectors $x$ and $x^*$ with the condition $\|x-x^*\|\leq \beta$ such that $x^*\in \overline{\mathrm{B}}( \mathcal{R}(x),\eta)$ such that  $\left(y-x\right) ^T F(x) +\beta\geq 0, \quad \forall y \in \overline{\mathrm{B}}( \mathcal{R}(x),-\eta)$ 
  \end{itemize}
  }
}

\problemStatement{$VI(F)$ With a weak Separation Oracle}{
  Input={We receive as input all the following:
  \begin{itemize}
      \item  A circuit $C_{\mathcal{R}}$ which represents a weak separation oracle for a non-empty closed convex set $\mathcal{R}$,
      \item A circuit $C_F$ which represents a $L_F$-Hausdorff Lipschitz  function $F: \mathbb{R}^{m*} \rightarrow  \mathbb{R}^{m*}$,
      \item  Accuracy parameters $\beta,\eta$. 
  \end{itemize}},
  Output={One of the following cases:
  \begin{itemize}
      \item (Violation of almost non-emptiness): A certificate that  $\operatorname{vol}(\mathcal{R})\leq \operatorname{vol} (\overline{\mathrm{B}}(0,\eta))$
\item  (Violation of $L_F$-Hausdorff Lipschitzness of $F$):
 Two vectors $p, q \in \mathbb{R}^{m*}$ such that $\|F(p)-F(q)\|>L_F\|p-q\|$,
\item One vector with the condition $x\in \overline{\mathrm{B}}( \mathcal{R},\eta)$ such that  $\left(y-x\right) ^T F(x) +\beta  \geq 0, \quad \forall y \in \overline{\mathrm{B}}(\mathcal{R},-\eta)$  
  \end{itemize}
  }
}

\begin{theorem}
   The generalized quasi-variational inequality problem (GQVI),  quasi-variational inequality problem (QVI) variational inequality problem (VI) for strong separation oracles are PPAD-complete. 
\end{theorem}

\begin{proof}
    Proof follows by combining the techniques used for PPAD-completeness of the above-mentioned variational inequality problems by repeating the same procedure considering the differences between weak and separation oracles.  
\end{proof}

\subsection{Remedial L/F Equilibrium With Weak Separation Oracles}

Building on previous contributions and theorems, we aim to further explore them within the context of weak separation oracles. For the partial responses while considering the definitions of $Z^{\mathrm{I}}(x_\mathrm{I},x^*_\mathrm{II})$, $W_{(i,\mathrm{I})}(x^*_{\mathrm{II}})$ and $V_{(i,\mathrm{I})}(x^*_{\mathrm{II}})$ for leader $\mathrm{I}$'s optimization problem and considering $x^*_{\mathrm{II}}$ as an exogenous variable, we find a solution $\left(x_{\mathrm{I}}, y_{\mathrm{I}}\right)$ to: 

\begin{equation}\label{leader12}
\begin{aligned}
& \operatorname{Min} \phi_{\mathrm{I}}\left(x_{\mathrm{I}}, x^*_{\mathrm{II}}, y_{\mathrm{I}}\right) \\
& \text { s.t } x_{\mathrm{I}} \in \overline{\mathrm{B}}(X^{\mathrm{I}},\eta) \\
& \text { and } \quad\left(x_{\mathrm{I}}, y_{\mathrm{I}}\right) \in \overline{\mathrm{B}}\left( \operatorname{graph} Z^{\mathrm{I}}\left(\cdot, x^*_{\mathrm{II}}\right),\eta\right)
\end{aligned}
\end{equation}

For leader $\mathrm{II}$, the optimization problem with the surrogate complementarity conditions can be defined similarly by considering $Z^{\mathrm{II}}(x^*_\mathrm{I},x_\mathrm{II})$, $ W_{(i,\mathrm{II})}(x^*_{\mathrm{I}})$ and $V_{(i,\mathrm{II})}(x^*_{\mathrm{I}})$ where $x^*_{\mathrm{I}}$ will be an exogenous variable:
\begin{equation}\label{leader22}
    \begin{aligned}
& \operatorname{Min} \phi_{\mathrm{II}}\left(x^*_{\mathrm{I}}, x_{\mathrm{II}}, y_{\mathrm{I}}\right) \\
& \text { s.t } x_{\mathrm{II}} \in \overline{\mathrm{B}}(X^{\mathrm{II}},\eta) \\
& \text { and } \quad\left(x_{\mathrm{II}}, y_{\mathrm{II}}\right) \in \overline{\mathrm{B}}\left(\operatorname{graph} Z^{\mathrm{II}}\left( x^*_{\mathrm{I}},\cdot \right),\eta\right)
\end{aligned}
\end{equation}

Finally, we are now ready to define the computational version of finding an equilibrium in this setting:

\begin{proposition}\label{remedial2}
In a multi-leader-follower game, the problem of finding remedial $(\epsilon,\delta)$-$L/F$-equilibrium where the constraints of the followers are given by weak separation oracles is PPAD-complete.
\end{proposition}
\begin{proof}
    The hardness simply follows by considering $\delta=0$. Inclusion in PPAD is similar to \ref{remedial} with the difference that we consider the weak separation oracle cases discussed in this section.
\end{proof}
\newpage

\problemStatement{Remedial $L/F$ Equilibrium With  Weak Separation Oracles}{
  Input={We receive as input all the following:
  \begin{itemize}
   \item  Two linear arithmetic circuits representing the convex loss functions  $(\phi_{\mathrm{I}}$,$\phi_{\mathrm{II}})$ for two leaders,
      \item Two linear arithmetic circuits representing weak separation oracles for non-empty, convex, and compact sets $X^{\mathrm{I}}$ and $X^{\mathrm{II}}$ for the leaders $\mathrm{I}$ and $\mathrm{II}$,
      \item Two linear arithmetic circuits representing weak separation oracles for  $Z^{\mathrm{I}}$ and $Z^{\mathrm{II}}$ that represent restricted or relaxed responses of the followers for each leader respectively that are two non-empty, convex, and compact correspondences. 
      \item Accuracy parameters $\beta$ and $\eta$. 
  \end{itemize}},
  Output={One of the following cases:
  \begin{itemize}
      \item (Violation of almost non-emptiness): A certificate indicating at least one of the following cases is almost empty:
    \begin{itemize}
        \item $X^{\mathrm{I}}$ or  $X^{\mathrm{II}}$
        \item $Z^{\mathrm{I}}\left(\cdot, x\right)$ for some $x\in X^{\mathrm{I}}$
            \item $Z^{\mathrm{II}}\left( x,\cdot\right)$ for some $x\in X^{\mathrm{II}}$
    \end{itemize}
       \item (Violation of convexity): of any of the loss functions of the inputs,
\item (Approximate Minimization): 
Vectors $(x^*_{\mathrm{I}}, y^*_{\mathrm{I}}, x^*_ {\mathrm{II}},y^*_ {\mathrm{II}})$  having the following relationship:
\begin{itemize}
    \item $\phi_{\mathrm{I}}\left(x^*_{\mathrm{I}}, x^*_{\mathrm{II}}, y^*_{\mathrm{I}}\right)\leq \beta+ \text{ Sol(\ref{leader12}) }  $
    \item $\phi_{\mathrm{II}}\left(x^*_{\mathrm{I}}, x^*_{\mathrm{II}}, y^*_{\mathrm{II}}\right)\leq \beta+ \text{ Sol(\ref{leader22}) }  $
\end{itemize}
  \end{itemize}
  }
}

\section{A Slightly Generalized Oracle Polynomial-Time Subgradient Ellipsoid Central Cut}\label{AppendixCentralcutellipsoidsection}

 For the sake of completeness (Theroems \ref{gqviinppad} and \ref{D1theorem}), we will present a more generalized version of the sub-gradient-cut method to solve the modified convex-constrained optimization problem in polynomial time given the separation oracles. Our generalized polynomial-time oracle algorithm is general enough to capture all desired properties and also future similar applications.  For example, for the exception of violation of almost non-emptiness in multi-leader-follower games, we need to output specific almost-emptiness exceptions to denote whether $\operatorname{vol}(\mathcal{R}(x))$ or $\operatorname{vol}(\mathcal{F}(x))$ is very small and our algorithm gives a suitable exception output. 
 
 Assume that $f(x)=(f_1(x),\dots,f_k(x))$ (compared to $f$ being a single coordinate valued function in \cite{concavegames}) and $\mathcal{X}=(X_1,\dots,X_s)$ (compared to $X$ being one set). We also know that for each $i\in[s]$, $X_i$ has a weak separation oracle.  Similar to \cite{concavegames}, for each $i\in[k]$, the approximate value and subgradient oracle for the objective function $f_i$ are available. One main difference of our algorithm is that it changes one index of $x^{(t)}\in (X_1,\dots,X_s)$ in each step $t$ while other coordinates stay the same.

The cutting plane methods are distinguished by their construction of sets $M^{(t)}$ and selection of query points $x^{(t)}$. These methods exhibit an exponential decrease in the volume of $M^{(t)}$ as $t$ increases, which leads to linear convergence guarantees in the presence of gradient and value oracles\footnote{Here, we only have a more complex function with constant coordinates and this will not ruin the exponential decrease in size guarantee.}. Given a desirable approximate parameter $\epsilon$  and also margin $\delta$, we can set some thresholds $T_{emptiness}$ and $T_{ellipsoid}$ so that we achieve the desirable outputs for the minimization (and also maximization) problems such as the modified convex-constrained optimization problem that we discussed before.

\SetKwComment{Comment}{/* }{ */}
\SetKwInput{KwData}{Input}
\SetKwInput{KwResult}{Output}

\begin{algorithm}
\caption{Subgradient Central-Cut Ellipsoid Method}\label{alg:two}
\KwData{Gradient and value oracles $\mathrm{O}_{\mathrm{grad}}^{f_{i}}, \mathrm{O}_{\mathrm{val}}^{f_{i}}$ with accuracies $(\epsilon_{\mathrm{grad}}, \epsilon_{\mathrm{val}})$, for each $i\in[k]$.}
\KwData{Weak separation oracles $\mathrm{O}^{i}_{\text {sep }}$ for set $X_i$ with margin $\delta$, for each $i\in [s]$ }
$index \leftarrow 1;$\\
\For{ $t \in\left[T_{\text {ellipsoid }}\right]$  }{
 
\eIf{ $x^{(t)}\in \overline{\mathrm{B}}(\mathcal{X}, \delta)$ \Comment*[r]{$x^{(t)}=(x^{(t)}_1,\dots,x^{(t)}_k)$} }{

\For{ $i \in\left[k\right]$ }{
Call a gradient oracle $g^{(t)}_i \leftarrow \mathrm{O}_{\text {grad }}^{f_i}\left(x^{(t)}\right)$\;
}

  \eIf{ $\forall i \in [k]$, $\left\|g^{(t)}_i\right\| \leq G_{\text {threshold }}$}{
  \KwResult{$x^{(t)};$} 
  }{
  
  \For{ $i \in [k]$ }{  $w^{(t)}_i \leftarrow g^{(t)}_i /\left\|g^{(t)}_i\right\|_{\infty}$ (Output A);}
}

    }{
     \If{ $ mod(index,s)==0$ }{
    Call the separation oracle $w^{(t)}_i \leftarrow \mathrm{O}_{\text {sep }}^f\left(x^{(t)}_i\right)$;\\
    $calls_{index}\leftarrow calls_{index}+1;$\\
   $ index\leftarrow \mod(index,s)+1;$\\
    \If{$calls_{index} > T_{\text {emptiness }}$ }{   \KwResult{$\perp_{index}$ \Comment*[r]{Emptiness of the respective index in the domain}}}
   
    }
}
For each $i\in [k]$, construct an ellipsoid $M_i^{(t+1)}$ such that : $\left\{x \in M_i^{(t)}: {w_i^{(t)}}^{\top}\left(x-x^{(t)}\right) \leq \delta \right\} \subseteq M_i^{(t+1)}$;\\
Let $x^{(t+1)}$ have all of the centroids of $M_i^{(t+1)}$ for each $i\in[k]$;
}
\KwResult{The iteration $\bar{x} \in \operatorname{argmin}\left\{\mathrm{O}_{\text {val }}^f(x) \mid x \in\left\{x^{(1)}, \cdots, x^{(T_{\text {ellipsoid }})}\right\} \cap \overline{\mathrm{B}}(\mathcal{X}, \delta)\right\}$}
\end{algorithm}

\begin{proposition}

    There exists a cutting plane method, referred to as the "central-cut Ellipsoid method", with a decay rate of $\theta=O(1 / d)$, such that for all $i\in [k]$\footnote{$\operatorname{vol}$ is the usual d-dimensional volume.}:
    $$
\frac{\operatorname{vol}\left(M_i^{(t)}\right)}{\operatorname{vol}\left(M_i^{(1)}\right)} \leq e^{-\theta t}
$$

\end{proposition}

Define the following sets:
$$\mathcal{R}_{(i,\epsilon)}=\left\{x \in \overline{\mathrm{B}}(\mathcal{X},-\delta): \min _{x\in\overline{\mathrm{B}}(\mathcal{X},-\delta)} f_i(x) \leq f_i(x) \leq \min _{x \in \overline{\mathrm{B}}(\mathcal{X},-\delta)}f_i(x)+\frac{\epsilon}{2}  \right\}$$

This set is the set of all $\epsilon$-approximate and $\delta$-marginally inside $\mathcal{X}$ minimum solutions of $f_i$.  First, we need to show that,  $\mathcal{R}_{(i,\epsilon)}$ has non-zero volume. If we assume that $\overline{\mathrm{B}}(\mathcal{X},-\delta) \neq \emptyset$, then by $L$-lipschitzness of $f_i$ for each $i\in [k]$, we know that $\mathcal{R}_{(i,\epsilon)}$ has a ball of radius $r(\epsilon, \delta)=\min \{\delta, \epsilon / L\}$. 


Next, we define $\mathcal{T}_{\text {active }}:=\left\{t \in [T_{\text {ellipsoid}}] ~|~ x^{(t)} \in \overline{\mathrm{B}}(\mathcal{X}, \delta)\right\}$, and $w^{(t)}=\Pi_{i=1}^k w_i^{(t)}$. There are three possible cases:

\begin{itemize}
    \item Case 1: Assume that for any $t \in \mathcal{T}_{\text {active }}$, any $i\in [k]$, and for any $x_{(i,\epsilon)} \in \mathcal{R}_{(i,\epsilon)}$, we  have that ${w_i^{(t)}}^{\top}\left(x_{(i,\epsilon)}-x^{(t)}\right) \leq \delta$. This can imply that $\forall i\in[k]$, $\forall t \in\left[T_{\text {ellipsoid }}\right],  \forall x_{(i,\epsilon)} \in \mathcal{R}_{(i,\epsilon)}: {w_i^{(t)}}^{\top}\left(x_{(i,\epsilon)}-x^{(t)}\right) \leq \delta $. The reasoning follows by the definition of the separation oracles we have (${w_i^{(t)}}^{\top}\left(x-x^{(t)}\right) \leq \delta$ ) for all $x \in \overline{\mathrm{B}}(\mathcal{X},-\delta )$. Thus, for all $i\in [k]$, it holds that:
$$
\forall t \in\left[T_{\text {ellipsoid }}\right]: D_{(i,\epsilon)} \subseteq M_i^{(t)} \Rightarrow \operatorname{vol}\left(\mathcal{R}_{(i,\epsilon)}\right) \leq \operatorname{vol}\left(M_i^{(t)}\right)
$$

Next, we show that the above condition can happen only if $t \leq C_0 \cdot d^2 \log (d / r(\epsilon, \delta))$. For any $i\in [k]$ some positive constant $C_0$ independent of $d, \delta$, we have the following inequalities:
$$
\left\{\begin{array}{l}
\frac{\operatorname{vol}\left(M_i^{(t)}\right)}{\operatorname{vol}\left(M_i^{(1)}\right)} \leq e^{-\theta t}, \quad \theta=\Theta\left(\frac{1}{d}\right) \\
\frac{\pi^d}{(d / 2+1) !} r(\delta, \epsilon)^d=\operatorname{Vol}(\overline{\mathrm{B}}(\mathcal{X}, r(\delta, \epsilon))) \\
\operatorname{vol}(\overline{\mathrm{B}}(\mathcal{X}, r(\delta, \epsilon))) \leq \operatorname{Vol}\left(\mathcal{R}_{(i,\epsilon)}\right) \leq \operatorname{vol}\left(M_i^{(t)}\right) \\
\operatorname{Vol}\left(M^{(1)}\right) \leq \operatorname{vol}(B o x)
\end{array} \quad \Longrightarrow t \leq C_0 \cdot d^2\left(\log \left(\frac{d}{2 r(\epsilon, \delta)}\right)\right)\right.
$$
 If for any $i\in [s]$ the number of used separation oracle calls are greater than $T_{\text {emptiness}}$, then $\operatorname{Vol}(\mathcal{X}) \leq \operatorname{Vol}(\overline{\mathrm{B}}(0, \delta))$, or consequently $\overline{\mathrm{B}}(\mathcal{X},-\delta)=0$. Otherwise, if we set $T_{\text {ellipsoid }}=\max \left\{C_0, 10\right\} d^2\left(\log \left(\frac{d}{2 r(\epsilon, \delta)}\right)\right.$, then for $C_0 \cdot d^2\left(\log \left(\frac{d}{2 r(\epsilon, \delta)}\right)<t \leq T_{\text {ellipsoid, }}\right.$, either Case 2 or 3 hold.
\item  
Case 2: If for all $i\in [k]$, $\left\|g_i^{(t)}\right\| \leq G_{\text {threshold}}$ for appropriate choice of $G_{\text {threshold }}$, we will show that $x^{(t)}$ is an $\epsilon$-approximate minimizer. Indeed, for each $i\in [k]$, by convexity $\min _{x \in \overline{B}(X,-\delta)} f_i(x) \geq f_i\left(x^{(t)}\right)+\min _{x \in \overline{\mathbb{B}}(X,-\delta)} \partial f_i\left(x^{(t)}\right)^{\top}\left(x-x^{(t)}\right)$. By choosing $G_{\text {threshold }}=O\left(\operatorname{poly}\left(d, \epsilon, \epsilon_{\text {grad }}\right)\right)$ such that $\epsilon \geq\left(G_{\text {threshold }}-\epsilon_{\text {grad }}\right) \sqrt{d}$, we have, for all $i\in[k]$, $f_i\left(x^{(t)}\right) \leq$ $\min _{x \in \overline{B}(\mathcal{X},-\delta)} f_i(x)+\epsilon$ and $f\left(x^{(t)}\right) \leq$ $\min _{x \in \overline{B}(\mathcal{X},-\delta)} f(x)+\epsilon\mathds{1}$.

\item Case 3: Assume that for any index $i$, the element $x_{\epsilon}$ at iteration $t^{\star} \in\left[T_{ellipsoid}\right]$ has the property that $w^{(t^{\star})}_i\left(x_{\epsilon}-x^{(t^*)}\right)>$ $\delta$. In this case, using the convexity of objective can imply $f_i\left(x^{(t^*)}\right) \leq f_i\left(x_{\epsilon}\right)-\nabla f_i\left(x^{(t^*)}\right)^{\top}\left(x_{\epsilon}-x^{(t^*)}\right)$. The last expression is equal to:

$f_i\left(x_{\epsilon}\right)-\left(\nabla f_i\left(x_{t^*}\right)-g^{(t^*)}_i\right)^{\top}\left(x_{\epsilon}-x_{t^*}\right)-{g^{(t^*)}_i }^{\top}\left(x_{\epsilon}-x_{t^*}\right) \leq f_i\left(x_{\epsilon}\right)+\epsilon_{\text {grad }} \sqrt{d}-\delta$. If we set $\epsilon_{\text {grad }} \leq \frac{\epsilon}{\sqrt{d}}$ and $\delta \leq \frac{\epsilon}{2}$, then $f_i\left(x_{t^*}\right) \leq f_i\left(x_{\epsilon}\right)+\frac{\epsilon}{2} \leq \min _{x \in \bar{B}(X,-\delta)} f_i(x)+\epsilon$. This means that: $f\left(x_{t^*}\right) \leq f\left(x_{\epsilon}\right)+\frac{\epsilon}{2} \mathds{1} \leq \min _{x \in \bar{B}(X,-\delta)} f(x)+\epsilon\mathds{1}$
 
\end{itemize}
\begin{remark}
    It is noteworthy to highlight that the recent study on the optimal fusion of subgradient descent and the ellipsoid method, conducted by Rodomanov and Nesterov \cite{Nesterov}, has exclusively concentrated on the scenario of a strong oracle. We investigate only the most complex case similar to \cite{concavegames}
\end{remark}

\section{Decision Approximate Resilient Nash}\label{AppendixDescionResilient}

Here, we briefly review some hardness results of \cite{vadhan,Conitzer,MFCS} to establish NP-hardness of the decision version of $t$-resilient Nash equilibrium problem.  This generic reduction provided by \cite{Conitzer} (from SAT which is an NP-complete problem) improved the NP-completeness results for Nash equilibrium with certain constraints given in \cite{Gilboa} (which are only concerned with exact versions).  However, the proof of inapproximability in \cite{Conitzer} does not apply to the form of approximation that we consider in this paper (the additive approximation error introduced in \cite{ComplexityNASH}). A reduction provided by Schoenebeck and Vadhan\cite{vadhan} (which modifies the proof provided in \cite{Conitzer}) shows that the problem of finding a Nash equilibrium with a certain guaranteed payoff for all players even under poly-approximation is NP-complete. Our approach is similar to \cite{MFCS} combining the techniques and ideas given in both reductions given in \cite{Conitzer} (Corollary $6$) and \cite{vadhan} (Theorem $8.6$).

\subsection{Approximate Guaranteed Nash Equilibrium}

We first restate the game showing that deciding whether a bi-matrix game has a $\epsilon$-approximate Nash equilibrium with a guaranteed payoff of $n-1-\epsilon$ is NP-complete. This problem is called \emph{guaranteed Nash}.

\begin{definition}
	 A \emph{Boolean formula $\phi$  in CNF} (conjunctive normal form) is specified by a set $V$  of \emph{variables} (with $|V|=n$), a set of $L$ of \emph{literals} consisting of variables and their negations, and a set $C$ of clauses, where each clause is a set of literals. A \emph{3CNF formula} is a CNF formula in which each clause has exactly 3 literals. 3CNFSAT is the problem of deciding whether there is a satisfying assignment for a 3CNF formula $\phi$ (i.e. a setting of the variables to \emph{true} or \emph{false} under which $\phi$ evaluate to \emph{true}.)
  \end{definition}

 \begin{proposition}
     The decision problem of whether given a sat formula has a solution is NP-complete.
 \end{proposition}

The following definition of the game $\mathcal{SV}(\phi,\epsilon)$ from \cite{vadhan} (Theorem $8.6$) is a simpler game compared to \cite{Conitzer} that single strategy $f$  from strategy set of the players. The strategy set of the players is $S=L\cup C \cup V $. Let  $\mathrm{v}: L \rightarrow V$ be the function that gives the variable corresponding to a literal, e.g. $\mathrm{v}\left(x_{1}\right)=\mathrm{v}\left(-x_{1}\right)=x_{1}$. For example, if $x_1$ is a variable, $x_1$ and $-x_1$ are literals that are \textbf{representatives} of variables $x_1$ being true or false respectively.  The following symmetric bi-matrix game $\mathcal{SV}(\phi,\epsilon)$ is defined as follows:

	\begin{enumerate}
	    \item $u_1\left(l^1, l^2\right)=n-1$, where $l^1 \neq -l^2$ for all $l^1, l^2 \in L$. This will ensure each player gets a high payoff for playing the aforementioned strategy.
	    \item $u_1(l, -l)=n-4$ for all $l \in L$. This will ensure that each player does not play a literal and its negation at the same time.
	    \item $u_1(v, l)=0$, where $\mathrm{v}(l)=v$, for all $v \in V, l \in L$. This, along with rule 4, ensures that for each variable $v$, each agent plays either $l$ or $-l$ with a probability of at least $\frac{1}{n}$, where $\mathrm{v}(l)=\mathrm{v}(-l)=v$.
	    
	    \item $u_1(v, l)=n$, where $\mathrm{v}(l) \neq v$, for all $v \in V, l \in L$.
	    
	    \item $u_1(l, x)=n-4$, where $l \in L,~ x \in V \cup C$. This, along with rules 6 and 7, ensures that if both players do not play the literals, then the payoffs cannot meet the guarantees.
	    \item $u_1(v, x)=n-4$ for all $v \in V, x \in V \cup C$.
	    \item $u_1(c, x)=n-4$ for all $c \in C, x \in V \cup C$.
	    \item $u_1(c, l)=0$ where $l \in c$ for all $c \in C, l \in L$. This, along with rule 9, ensures that for each clause $c$, each agent plays a literal in the clause $c$ with probability least $\frac{1}{n}$.
	    \item $u_1(c, l)=n$, where $l \notin c$ for all $c \in C, l \in L$.

	\end{enumerate}

 \begin{remark}
     Note that the games introduced in both \cite{Conitzer,vadhan} are symmetric games which means $\forall s_1,s_2 \in S,\space  u_1\left(s_1, s_2\right)=u_2\left(s_2, s_1\right) $.
 \end{remark}
	
\begin{proposition}[\cite{vadhan}]\label{SVgame}
Let $\epsilon=\frac{1}{2 n^3}$ and let the guarantee to each player be $n-1$. Given a 3CNF $\phi$ is satisfiable iff there exists a $\epsilon$-approximate Nash equilibrium in $\mathcal{SV}(\phi,\epsilon)$ where each player has a guaranteed payoff of $n-1-\epsilon$. 
\end{proposition}

We break the proof down into multiple parts and use them later\footnote{The complete proofs of the lemmas are available in \cite{MFCS}.}.

\begin{lemma}
	\label{maxsw}
    In any $\epsilon$-approximate Nash equilibrium with the guaranteed payoff $n-1-\epsilon$ in $\mathcal{SV}(\phi,\epsilon)$, clauses and variables are played with a probability of at most $\epsilon$.
\end{lemma}

\begin{lemma}\label{lornegativel}
     In any $\epsilon$-approximate Nash equilibrium with a guaranteed payoff of $n-1-\epsilon$ in $\mathcal{SV}(\phi,\epsilon)$, for any $l\in L$, the probability that the row player plays $l$ or $-l$ is at least  $\frac{1}{n}-2 \epsilon$.
\end{lemma}

\begin{lemma}\label{bothliterals}
      In any $\epsilon$-approximate Nash equilibrium of $\mathcal{SV}(\phi,\epsilon)$ with a guaranteed payoff of $n-1-\epsilon$, for each player and any literal $l\in L$, either $l$ or $-l$ is played with  probability $\geq$ $\frac{1}{n} -2\epsilon-\frac{1}{ n^{2}}$ while the other is played with probability less than $\frac{1}{ n^{2}}$
\end{lemma}

Now, we are prepared to define the computational problem, focusing our proof on the simple case of bi-matrix games.

\problemStatement{$t$-Resilient Nash in Bi-matrix Games}{
  Input={We receive as input all the following:
  \begin{itemize}
      \item  A bi-matrix game $\mathcal{G}$ with $k$ players represented by $k$ utility matrices ($u_i$ for all $i\in[k]$),
      \item The strategy sets $\mathcal{S}$, 
      \item  An accuracy parameter $\epsilon$.
  \end{itemize}},
  Output={The following case:
  \begin{itemize}   
\item Whether there exists a vector $s^*$ which represents the strategy profile of all players that satisfies:
$$
\left(\forall J\in \mathcal{J} \mbox{ s.t } |J|\leq t\right), \forall s_J^\prime\in S_J, \forall j\in J: u_j\left(\textbf{s}^*\right) +\epsilon \geq u_j\left(\textbf{s}^\prime_J, \textbf{s}^*_{-J}\right)
$$  
  \end{itemize}
  }
}

\begin{theorem}\label{resides}
    The problem of $t$-resilient Nash for 3-player games is NP-complete.
\end{theorem}

\begin{proof}
   We need to add the following strategies to $\mathcal{SV}(\phi,\epsilon)$. We name the modified game $\mathcal{G}(\phi,\epsilon)$ where the strategy set is now $S=L\cup C \cup V \cup \{f\}$ similar to \cite{MFCS}. We also assume that there exists a third player that has the same payoff on any strategies in $S$.
   
   \begin{itemize}
       	\item $u_{1}(x, f,y)=u_{2}(f, x,y)=0$ for all $x \in S-\{f\}$ and for all $y\in S$;
	
	\item $u_{1}(f, f,y)=u_{2}(f, f,y)=0$ for all $y\in S$;
	
	\item $u_{1}(f, x,y)=u_{2}(x, f,y)=n-1$ for all $x \in S-\{f\}$ and for all $y\in S$;
 \item $u_{3}(x,x^\prime,y)=1$ for all $x,x^\prime,y \in S$.
   \end{itemize}

   Other rules can be modified to three players similarly. We now show that if $\phi$ is satisfiable, there is a uniform EXACT Nash equilibrium of $\mathcal{SV}(\phi,\epsilon)$ where the first and second players play $l^i \in L$ or $-l^i\in L$ (for all $i$) uniformly with probability $\frac{1}{n}$ getting the expected payoff $n-1$ while the third player plays an arbitrary strategy. The proof is the same as \cite{Conitzer,vadhan,MFCS}. Furthermore, this Nash equilibrium is an EXACT strong Nash equilibrium as well \cite{Conitzer} and similarly can be shown that it is an EXACT $2$-resilient Nash equilibrium. This is followed by the fact that there exists no other EXACT Nash equilibrium other than this equilibrium unless both the first and the second players play $f$ \cite{Conitzer}.

Now suppose that $\phi$ is not satisfiable. We can show that in any $\epsilon$-approximate Nash equilibrium, at least one player always receives an expected payoff of less than $n-1-\epsilon$ which causes one of two first players to deviate to $f$. According to \cite{vadhan}, for the approximate version, we need to do the following modification to the correspondence between literals and truth assignments compared to \cite{Conitzer}. We consider a literal is \emph{true} if it is played more often than its negation. If an assignment does not satisfy the formula, there is at least one clause that does not have a satisfying literal. Any of the first two players will turn to this clause strategy to receive a payoff of $n$ whenever the opponent plays a literal that is not in that clause. We know that in the game $\mathcal{SV}$ the second player plays literals with probability $>1-\epsilon$ by Lemma \ref{maxsw}, and there are only $3$ literals in each clause each of which the second player plays with probability $\leq \frac{1 }{n^{2}} $ by Lemma \ref{bothliterals}. By changing the strategy to this clause, the first player will receive at least $\left(1-\epsilon-\frac{3}{n^{2}} \right) n>n-1+2 \epsilon$. The proofs can be extended to the case where the game includes the strategy $f$ and a third trivial player exists if both the first and the second player do not use $f$ (see \cite{MFCS}). So either the first player can do $\epsilon$ better by changing his strategy or he is already receiving $n-1+2\epsilon$ and so the second player does not have a guaranteed payoff of at least $n-1-\epsilon$. Reaching the payoff less than  $n-1-\epsilon$ will cause either of the first two players to deviate to $f$. If either the first or the second player uses $f$, no approximate $2$-resilient Nash could not be formed as the coalition of the first, and the second player can always gain more than $0$ (particularly, $\frac{n}{2}>\epsilon$ each player) for either of the players.   
\end{proof}

\end{document}